\documentclass[12pt]{amsart} 
\usepackage[english]{babel}
\usepackage[utf8]{inputenc}
\usepackage{amsmath}
\usepackage{graphicx}
\usepackage{subfigure}
\usepackage{amssymb}
\usepackage{amsthm}
\usepackage{bm}
\usepackage{tikz-cd}
\usepackage{mathrsfs}
\usepackage[colorinlistoftodos,prependcaption,textsize=tiny]{todonotes}
\usepackage{enumitem}
\usepackage{yfonts}
\usepackage{ dsfont }
\usepackage{bbm}
\usepackage{geometry}
\usepackage{setspace}
\usepackage{tikz}
\usepackage{pgfplots}
\usepackage{cite}
\pgfplotsset{compat=newest}
\usepgfplotslibrary{fillbetween}
\usetikzlibrary{positioning}
\usetikzlibrary{decorations.pathreplacing}
\usetikzlibrary{decorations,arrows}
\usetikzlibrary{decorations.markings}
\usetikzlibrary{patterns}
\usetikzlibrary{quotes,angles}
\usetikzlibrary{arrows}

\numberwithin{equation}{section}
\newtheorem{theorem}{Theorem}[section]
\newtheorem{lemma}[theorem]{Lemma}

\newtheorem{corollary}[theorem]{Corollary}

\newtheorem{remark}[theorem]{Remark}
\newtheorem{proposition}[theorem]{Proposition}
\newtheorem{assumption}[theorem]{Assumption}
\allowdisplaybreaks[4]

\title{Robustness of Valley-Hall Interface Modes Against Sharp Bending}

\date{}

\author{Habib Ammari} %
\address[H. Ammari]{Department of Mathematics, ETH Z\"{u}rich, R\"{a}mistrasse 101, CH-8092 Z\"{u}rich, Switzerland}
\email{habib.ammari@math.ethz.ch}

\author{Jiayu Qiu} %
\thanks{The second author is the corresponding author.}
\address[J. Qiu]{Department of Mathematics, ETH Z\"{u}rich, R\"{a}mistrasse 101, CH-8092 Z\"{u}rich, Switzerland}
\email{jiayu.qiu@sam.math.ethz.ch}

\begin{document}

\begin{abstract}
It is well known that band inversion across a straight interface in a periodic medium gives rise to interface modes that are localized near the interface and propagate along it inside the bulk spectral gap. This phenomenon constitutes the key mechanism underlying the valley-Hall effect. In this paper, we address the problem of the robustness of such interface modes. We prove that, when the interface is bent through an angle of $\frac{2\pi}{3}$, the interface modes persist for every frequency in the bulk spectral gap where the group velocity is non-vanishing, except for a finite exceptional set. We also show that corner-localized modes, if they occur, can appear only at these exceptional frequencies and have finite multiplicity. To the best of our knowledge, this is the first rigorous mathematical theory of the bending immunity of valley-Hall interface modes.
\end{abstract}

\maketitle

\bigskip

 \par

\bigskip

\\

\tableofcontents

\section{Introduction}

\subsection{Background} \label{sec_background}

Recent developments in topological insulators have opened new avenues for guiding waves robustly. For example, it is known from the bulk-edge correspondence (BEC) principle that the interface between two insulating media with distinct topological invariants, such as the gap Chern number or the Kane-Mele index, supports unidirectional propagating modes \cite{qizhang11topo_insulator,bernivig13topo_insulator}. As a consequence of their topological origin, such interface modes exhibit high robustness against disorder and impurities present in the system as long as the topological indices remain unchanged; therefore, these interface modes are referred to as \textbf{topologically protected}. Remarkably, this unidirectional guiding of waves is also realized in classical wave systems, not only in condensed-matter settings, for example in photonic systems \cite{ozawa19topological_photonics}. In that case, wave propagation is governed by the Maxwell operator, or by its scalar reduction in the case of plane-polarized waves, whose band structure still encodes the relevant topological information.

However, the application of BEC to achieve robust guiding of waves is limited by the scarcity of nontrivial topological indices in practical systems. For example, in photonic structures, the emergence of a nonvanishing Chern number generally requires breaking time-reversal symmetry through magneto-optical effects, which in turn necessitates strong external magnetic fields \cite{ozawa19topological_photonics}. This difficulty has motivated the investigation of alternative mechanisms that may produce robust interface modes without relying on conventional topological indices.

Remarkably, a mechanism of this type arises in the study of the (photonic) valley Hall effect. Its basic structure may be summarized as follows:
\begin{itemize}
    \item[(i)] We begin with a periodic medium described by an elliptic operator $\mathcal{L}$ that is invariant under a symmetry group $\mathcal{S}$. The presence of symmetry may enforce special spectral features, including degeneracies between adjacent Bloch bands. For instance, when $\mathcal{S}$ is generated by $\mathcal{R}$ ($\frac{2\pi}{3}$-rotation) and $\mathcal{P}$ (inversion symmetry), the band structure may exhibit conical degeneracies, i.e. the so-called Dirac points, at high-symmetry quasi-momenta of the Brillouin zone \cite{fefferman12jams}.
    \item[(ii)] Next, if we perturb the system by a symmetry-breaking term by replacing $\mathcal{L}$ with
    $$
    \mathcal{L}+\mathcal{L}_{per},
    $$ 
    where $\mathcal{L}_{per}$ does not commute with the inversion symmetry $\mathcal{P}$, the degeneracy is lifted and a local spectral gap opens in a neighborhood of the Dirac point (see Figure \ref{fig_perturbed_structure_and_bands} for an illustration). Importantly, for suitably chosen perturbations, the Bloch eigenspaces associated with the upper and lower bands near the gap exhibit opposite symmetry character for the operator $\mathcal{L}+ \mathcal{L}_{per}$ and $\mathcal{L}-\mathcal{L}_{per}$, as measured, for example, by the rotation symmetry; see Figure \ref{fig_perturbed_band} for an illustration. This phenomenon is commonly referred to as \textbf{band inversion} \cite{vanderbilt2018berry}.
    
    \item[(iii)] Finally, we form a heterostructure by adjoining the two bulk media governed by $\mathcal{L}+ \mathcal{L}_{per}$ and $\mathcal{L}-\mathcal{L}_{per}$ along an interface, with the governing operator for the interface model denoted by $\mathcal{L}_{int}$. It is conjectured and has been supported by both formal and rigorous analysis in specific settings that a weaker form of bulk–edge correspondence holds in this context: namely, $\mathcal{L}_{int}$ possesses interface modes whose frequencies lie in the bulk spectral gap and whose profiles are localized near the interface whenever the adjoining media exhibit band inversion. In this regime, the interface states may be viewed as bifurcating from the degenerate spectral point at which the original band touching occurs.
\end{itemize}
Compared with the classical BEC mechanism based on nontrivial topological indices, the band-inversion mechanism is often much easier to realize in concrete models. In photonic systems, for example, it may be produced through suitable geometric deformations of a periodic structure, without breaking time-reversal symmetry by external magnetic fields \cite{ozawa19topological_photonics}. Such geometric deformations are readily implemented in photonic crystals, making them a promising platform for robust wave guiding. Remarkably, even though the resulting interface modes do not arise from a global topological invariant in the usual sense, experiments show that they still possess a striking degree of robustness, making them highly advantageous over the guided modes of conventional waveguides. Two manifestations of this robustness are particularly important:
\begin{itemize}
    \item[(i)] \textbf{symmetry protection}: suppose that the interface operator $\mathcal{L}_{int}$ is invariant under some symmetry $\mathcal{O}$, then the corresponding interface modes are expected to be robust against perturbations that preserve $\mathcal{O}$ \cite{fu2011crystalline}.
    \item[(ii)] \textbf{bending immunity}: even more remarkably, experiments indicate that interface modes arising from band inversion remain robust when the interface is sharply bent. In other words, when a straight interface is replaced by a bent one that forms a corner (see Figure \ref{fig_bended_interface_structure}), the interface modes can pass through the corner with negligible reflection.
\end{itemize}
Of these two aspects, symmetry protection is comparatively well understood and was analyzed rigorously by the authors in \cite{ammari2026symmetry_protect}. By contrast, sharp-bend immunity is much less understood. Nevertheless, a complete understanding of this phenomenon is highly desirable, since it underlies the proposed applications of valley-Hall-type photonic crystals, for example in chiral quantum optics \cite{Hafezi2020chiral_optic}. Despite a variety of informal explanations in the physics literature, there is still no consensus on the precise mechanism responsible for sharp-bend immunity; see the discussions in \cite{liu2024bending_immune_1,dai2025bending_immune_2,xue2021device_bend_immune_3,ma2016all_si_valley,arora2021direct,leveque2023scattering}, etc. The purpose of this paper is to provide a rigorous mathematical theory for this robustness under sharp bends. To the best of our knowledge, this is the first mathematical theory of the bending immunity of valley-Hall interface modes.

\subsection{Summary of Results}

The main result of this paper is summarized as follows. We begin with a straight-interface model that supports, for each frequency in the bulk spectral gap, two linearly independent interface modes arising from band inversion. Then we consider a bent-interface configuration in which \textbf{the interface is bent by an angle of $\frac{2\pi}{3}$}, forming a corner at the origin. Our main result, Theorem \ref{thm_bent_interface_mode}, proves that:
\begin{itemize}
    \item[(i)] for every frequency in the bulk spectral gap $\mathcal{I}_0$ where the group velocity of straight-interface modes is non-vanishing, except for a finite exceptional set $\mathcal{I}_0^{exc}$, there exist exactly two linearly independent interface modes propagating along the bent interface;
    \item[(ii)] the far-field profiles of the bend-interface modes are explicitly determined;
    \item[(iii)] for every frequency in $\mathcal{I}_0\backslash \mathcal{I}_0^{exc}$, there is no mode localized near the corner;
    \item[(iv)] for every frequency in the exceptional set $\mathcal{I}_0^{exc}$, a corner-localized mode may appear, but the multiplicity is at most finite.
\end{itemize}
Statement (i) and (ii) shows that the number of propagating interface modes is unchanged by the sharp bend, thereby establishing the bending immunity of valley-Hall interface modes. Statements (iii) and (iv) clarify the role of corner-localized modes. Such modes are known in the physics literature to be a serious obstacle in applications where interface modes are used as transport channels or cavities, yet, to the best of our knowledge, no rigorous mathematical result addressing their existence or absence has previously been available. Our theorem gives a definitive answer: corner-localized modes can occur only at exceptional frequencies, and even there they form at most a finite-dimensional space. In this sense, they are generically scarce.

\subsection{Relation to Previous Results} \label{sec_relation_previous}

From a mathematical perspective, the analysis of interface modes induced by band inversion originates in the pioneering work of Fefferman, Lee-Thorp and Weinstein. In \cite{fefferman2017topologically}, the authors establish the existence of interface modes in a one-dimensional \textbf{domain wall model}, which smoothly interpolates between two periodic Schrödinger operators exhibiting band inversion near a Dirac point. Their multiscale analysis has subsequently been extended to various two-dimensional domain-wall systems, proving the existence of interface modes bifurcating from Dirac points in honeycomb-type structures; see, for example, \cite{fefferman2016honeycomb_edge,lee2019elliptic,drouot2020edge}. In the context of chains of subwavelength resonators, Ammari et al. have adapted the fictitious source method and Toeplitz theory to study the interface modes \cite{SWP3,SWP4}. More recently, Qiu et al. have developed a layer-potential framework which applies to studying interface modes in \textbf{sharp interface models}, in which two bulk media are directly adjoined without adiabatic modulation. This approach applies to interface modes bifurcating from Dirac or quadratic degeneracies in one- and two-dimensional systems \cite{qiu2026waveguide_localized,li2024interface,qiu2024square_lattice}. 

We emphasize that \textbf{a common feature of these mentioned works is that the interface is assumed to be straight}. This assumption is crucial for the use of Floquet-type techniques in the analysis of interface modes. In fact, the straight interface case is by now rather well understood: all straight interfaces with rational slopes are treated in \cite{fefferman2024discrete,li2024interface}, and the irrational-slope situation is studied in \cite{amenoagbadji2026continuum}. The focus of the present paper is therefore no longer the existence theory for straight interfaces, but rather the following challenging problem:
\begin{center}
Are interface modes arising from band inversion robust under a sharp bend of the interface?
\end{center}
We need to point out that band inversion can, in certain special situations, give rise to a nontrivial topological invariant. In such cases, it is well-known that the interface modes are robust against arbitrary deformation of the interface geometry based on the (strong) BEC principle. The mathematical theory of the BEC principle has been extensively developed in a variety of settings, including Schrödinger operators \cite{kellendonk2004quantization,taarabt2014equality,Combes2005edge_impurity,Bourne2016ChernNL,cornean2021landau+functional,ludewig2020shortrange+coarse,gontier2023edge_channel}, Dirac Hamiltonians \cite{bal2019dirac+functional,bal2023dirac+microlocal}, tight-binding Hamiltonians \cite{graf2018shortrange+transfer,avila2013shortrange+transfer,graf2013shortrange+scattering,graf2005equality,tauber2022chiral_finite_chain,qiu2025generalized,bourne2017ktheory,kubota2017ktheory,qiu2025bec_disorder_finite,drouot2024bec_curvedinterfaces,bols2018quantization}, and classical wave operators \cite{lin2022transfer,thiang2023transfer,ammari2024toeplitz_1,ammari2024toeplitz_2,
ammari2024toeplitz_3,SWP1,SWP2,SWP3,SWP4,qiu2025bec_finite,craster}; we also refer to the excellent monograph \cite{prodan2016ktheory} and the recent review \cite{bal2024review} for comprehensive accounts. In contrast, \textbf{the system considered in this paper does not possess any nontrivial topological invariant; consequently, the BEC principle is not applicable to account for the robustness of the interface modes under sharp bends}. To the best of our knowledge, this paper serves as the first mathematical theory addressing the question of bending-immunity of interface modes in such a non-topological setting. As will become clear in the sequel, the robustness of valley-Hall interface modes against the $\frac{2\pi}{3}$ bending is a natural consequence of \textbf{the lattice symmetry of the interface model}. More precisely, the deep reason is that
\begin{center}
\textit{After the $\frac{2\pi}{3}$ bending, the bent-interface structure is \textbf{reflectional symmetric} about the corner, which leads to a lossless scattering process and hence the persistence of number of interface modes.}
\end{center}
In fact, $\frac{2\pi}{3}$ is the only angle that satisfies this symmetry property, which explains why the $\frac{2\pi}{3}$-bent-structure is deemed as the most promising setup to realize bending-immuned interface modes in the physics literature \cite{liu2024bending_immune_1,dai2025bending_immune_2,xue2021device_bend_immune_3}. See Section \ref{sec_further_discussion} for more detailed discussion.

\subsection{Outline of Paper}

This paper is organized as follows:

\begin{itemize}
    \item In Section \ref{sec_main_result}, we present the detailed setup and main result of this paper. We begin with the unperturbed periodic elliptic operator $\mathcal{L}^{a}$, which serves as the starting point for the band-inversion mechanism described above. We assume that $\mathcal{L}^{a}$ respects the $\mathcal{R}$ ($\frac{2\pi}{3}$-rotation), $\mathcal{F}$ (reflection) and $\mathcal{P}$ (inversion symmetry), which gives rise to Dirac points at high-symmetry quasi-momenta of the Brillouin zone. After breaking the inversion symmetry $\mathcal{P}$, a (bulk) spectral gap $\mathcal{I}_0$ is opened near the Dirac point, and the band inversion occurs across the straight-interface $E$. As a result, the corresponding interface operator $\mathcal{L}^{E}$ possesses in-gap eigenvalues whose eigenfunctions are localized near $E$. We then introduce the bent-interface operator $\mathcal{L}^{bend}$, obtained by bending the straight interface through an angle of $\frac{2\pi}{3}$; see \eqref{eq_bended_interface_operator}. The main result, Theorem \ref{thm_bent_interface_mode}, states that the number of propagating in-gap interface modes is preserved under this bend and that corner-localized modes are absent except possibly at finitely many exceptional frequencies.
    \item In Section \ref{sec_bend_immunity}, we present the main strategy of the proof of Theorem \ref{thm_bent_interface_mode}, which proceeds in two steps. First, in Section \ref{sec_LP_formulation}, we reduce the bent-interface problem to a transmission problem across an auxiliary interface. More precisely, for any bent-interface mode $u^{bend}$, we decompose it into two parts, with the first part supported in the pre-bending medium, denoted as $u_{bend}\cdot\mathbbm{1}_{\Omega_{L}}$, and the second part supported in the after-bending medium, denoted as $u_{bend}\cdot\mathbbm{1}_{\Omega_{R}}$. By matching $u_{bend}\cdot\mathbbm{1}_{\Omega_{L}}$ and $u_{bend}\cdot\mathbbm{1}_{\Omega_{R}}$ across an auxiliary interface separating $\Omega_{L}$ and $\Omega_{R}$, we obtain a set of boundary integral equations. The number of solutions to this system is shown to coincide with the number of bent-interface modes, up to the possible contribution of corner-localized modes; see Proposition \ref{prop_layer_potential_formulation}. Secondly, in Section \ref{sec_LP_solution}, we solve the resulting boundary integral equation by combining analytic Fredholm theory with layer-potential methods. This step is highly nontrivial because the relevant layer-potential operators are built from \textbf{the out-going Green function} defined through the limiting absorption principle, rather than from the standard logarithmic kernel, and are posed on \textbf{a non-compact boundary}.  Establishing their regularity, Fredholm properties, and analyticity is one of the main technical contributions of the paper; see Theorem \ref{thm_analytic_Fredholm_SL} and Section \ref{sec_LP_solution} for a detailed discussion.
    \item In Section \ref{sec_further_discussion}, we discuss several related questions on the robustness of interface modes that remain open for future study.
    \item In Section \ref{sec_prelim}, we collect basic properties of the straight-interface modes that will be used later in the analysis of the boundary integral equations obtained in Section \ref{sec_LP_formulation}.
    \item In Section \ref{sec_LA_out_green}, we introduce the out-going Green function for frequencies in the bulk spectral gap, which forms the foundation of our layer-potential framework. Following the lines of \cite{joly2016solutions}, the out-going Green function is defined via the limiting absorption principle as presented in Section \ref{sec_LA_out_green_def}, in which we also list its basic properties, including the far-field asymptotics. Then, in Section \ref{sec_LA_out_green_analytic}, we construct an analytic extension of the out-going Green function with the frequency continued to the complex plane, following the lines of \cite{qiu2026embedded}, in order to prepare for the use of analytic Fredholm theory.
    \item In Section \ref{sec_LP_operators}, we establish the Fredholm property of the layer-potential operator, which is the final ingredient needed to apply the analytic Fredholm theory. The main difficulty here comes from the fact that the layer-potential operators are defined on an unbounded interface. To overcome this, we exploit the special bulk-interface-bulk geometry of the model together with the spectral gap of the bulk operator. Roughly speaking, this implies that the relevant Dirichlet modes remain concentrated near the interface, which in turn leads to finite-dimensional kernel spaces for the layer-potential operator. See Section \ref{sec_LP_fredholm} for the detailed argument.
\end{itemize}

\subsection{Notation} \label{sec_notation}
Here, we list some frequently-used notation and conventions in this paper for the reader's convenience.
\subsubsection{Geometry}

\begin{itemize}
    \item Background lattice $\Lambda:=\mathbb{Z}\bm{e}_1\oplus \mathbb{Z}\bm{e}_2$ with  unit vectors $\bm{e}_1:=(1,0)^{\top}$, $\bm{e}_2:=(-\frac{1}{2} , \frac{\sqrt{3}}{2})^{\top}$. The associated unit cell $Y:=\{s\bm{e}_1+t(\bm{e}_1+\bm{e}_2):\, s,t\in (0,1)\}$.
    \item Reciprocal lattice $\Lambda^{*}:=\mathbb{Z}\bm{e}_1^{*}\oplus \mathbb{Z}\bm{e}_2^{*}$ with $\bm{e}_1^*:=2\pi(1,\frac{\sqrt{3}}{3})^\top$, $\bm{e}_2^*:=2\pi(0,\frac{2\sqrt{3}}{3})^\top$ and the unit cell (Brillouin zone) $Y^*:=\{s\cdot \bm{e}_1^*+t\cdot \bm{e}_2^*: -\frac{1}{2}\leq s,t\leq \frac{1}{2}\}$.
    \item Straight interface (without bending) $E:=\mathbb{R}\times \{0\}$; see Figure \ref{fig_straight-interface structure}. The straight-interface structure is $\mathbb{Z}\bm{e}_1$-periodic, with the unit strip $S:=\{s\bm{e}_1+t(\bm{e}_2+2\bm{e}_1):\, s\in (0,1),\, t\in \mathbb{R}\}$. The translations of strip are denoted as $S_{n}:=S+n\bm{e}_1$, with $S_{0}=S$.    
    \item The bent interface $E^{bend}:=\mathbb{R}^{-}\bm{e}_1\cup \mathbb{R}^{-}(\bm{e}_1+\bm{e}_2)$; see Figure \ref{fig_bended_interface_structure}.
    \item Auxiliary interface (to solve the bent-interface mode) $\Gamma:=\mathbb{R}(2\bm{e}_1+\bm{e}_2)$, with the right-pointing normal vector $\nu:=\frac{1}{2}(1,-\sqrt{3})^{\top}$. The auxiliary interface $\Gamma$ separates the whole plane into the pre-bending part $\Omega_{L}:=\{\bm{x}\in \mathbb{R}^2:\, -\bm{x}\cdot \nu>0\}$ and the after-bending part $\Omega_{R}:=\{\bm{x}\in \mathbb{R}^2:\, \bm{x}\cdot \nu>0\}$. See Figure \ref{fig_bended_interface_structure}.
    \item Translations of the auxiliary interface $\Gamma_{n}:=\Gamma+n\bm{e}_1$, which form the boundary of translated strips, i.e., $\partial S_{n}=\Gamma_{n}\cup \Gamma_{n+1}$.
\end{itemize}

\subsubsection{Symmetry operations}

\begin{itemize}
    \item Operation on spatial variables $O:\, \bm{x}\to O\bm{x}$, with $O\in \{R,F,V,F_{\Gamma}\}$ where
    \begin{equation*}
    R:=\begin{pmatrix}
-\frac{1}{2} & -\frac{\sqrt{3}}{2} \\
\frac{\sqrt{3}}{2} & -\frac{1}{2} 
\end{pmatrix},\quad
F:=\begin{pmatrix}
-1 & 0 \\
0 & 1
\end{pmatrix},\quad
V:=\begin{pmatrix}
-1 & 0 \\
0 & -1
\end{pmatrix},\quad
F_{\Gamma}:=F\cdot R=\begin{pmatrix}
\frac{1}{2} & \frac{\sqrt{3}}{2} \\
\frac{\sqrt{3}}{2} & -\frac{1}{2} 
\end{pmatrix}
    \end{equation*}
denote the $\frac{2\pi}{3}$-rotation, ($y$-axis) reflection, inversion, and $\Gamma$-reflection, respectively.
    \item Induced operator on functions $\mathcal{O}:\, f(\cdot)\mapsto f(O\cdot)$ with $\mathcal{O}\in \{\mathcal{R},\mathcal{F},\mathcal{V},\mathcal{F}_{\Gamma}\}$ corresponding to $O\in \{R,F,V,F_{\Gamma}\}$, respectively. Note that $\mathcal{F}_{\Gamma}=\mathcal{R}\mathcal{F}$, reversing the order of multiplication in the level of spatial variable.  
\end{itemize}

\subsubsection{Function spaces and brackets}

\begin{itemize}
    \item (quasi-periodic space along $\mathbb{Z}\bm{e}_1$) $L_{\kappa,\bm{e}_1}^2:=\{u\in L^2_{loc}(\mathbb{R}^2):\, \|u\|_{L^2(S)}<\infty,\, u(\bm{x}+\bm{e}_1)=e^{i\kappa}u(\bm{x})\}$, equipped with $L^2(S)$-inner product ($\kappa\in [-\pi,\pi]$).
    \item ($\mathbb{Z}\bm{e}_1$-quasi-periodic Sobolev space) $H_{\kappa,\bm{e}_1}^n:=\{u\in H^n_{loc}(\mathbb{R}^2):\, \partial^{\alpha}u\in L_{\kappa,\bm{e}_1}^2 \, (|\alpha|\leq n)\}$, equipped with $H^n(S)$-inner product ($\kappa\in [-\pi,\pi]$, $n=1,2$).
    \item (boundary function spaces) $H^{\frac{1}{2}}(\Gamma):=\gamma H^1(\mathbb{R}^2)$, with $\gamma$ being the trace operator from $\mathbb{R}^2$ to $\Gamma$. $H^{-\frac{1}{2}}(\Gamma)$: the dual space of $H^{\frac{1}{2}}(\Gamma)$ induced by the $L^2$-product. The $H^{-\frac{1}{2}}-H^{\frac{1}{2}}$ dual pair is denoted as $\langle \varphi,\phi\rangle_{\Gamma}=\int_{\Gamma}\varphi\cdot\overline{\phi}$ for $\varphi\in H^{-\frac{1}{2}}(\Gamma)$ and $\phi\in H^{\frac{1}{2}}(\Gamma)$.
    \item (dual space $H^{-1}$) for either $U=\mathbb{R}^2$ or $U=S$, $H^{-1}(U)$ denotes the dual space of $H^1(U)$ induced by the $L^2$-product.\footnote{It is different from the conventional notation where $H^{-1}(U)$ refers to the dual of $H^{1}_{0}(U)$, with Dirichlet boundary conditions applied.} The $H^{-1}-H^{1}$ dual pair is denoted by $\langle u,v\rangle_{U}=\int_{U}u\cdot\overline{v}$ for $u\in H^{-1}(U)$ and $v\in H^{1}(U)$.
    \item ($H^{-1}$ functions with compact support along $E$) $H^{-1}_{x-comp}:=\big\{u\in H^{-1}(\mathbb{R}^2):\, \text{supp }(u)\subset \overline{\cup_{n\leq j\leq m}S_{j}} \quad \text{for some $n,m\in\mathbb{Z}$} \big\}$. In other words, any $u\in H^{-1}_{x-comp}$ is compactly supported in finitely many strips.
    \item (transversely integrable functions) $H^{n}_{y}:=\big\{u\in H^{n}_{loc}(\mathbb{R}^2):\, \sup_{k\in\mathbb{Z}}\|u\|_{H^{n}(S_{k})}<\infty \big\}$ ($n=0,1,2$). The functions in $H^{n}_{y}$ are $H^n$-bounded in each strip (i.e., transversely to the straight-interface $E$).
\end{itemize}

\subsubsection{Traces and conormal derivatives}

\begin{itemize}
    \item (one-side) traces on $\Gamma$: the right trace $\Gamma^{+}:\, H^1(\Omega_{R})\to H^{\frac{1}{2}}(\Gamma)$, and the left trace $\Gamma^{-}:\, H^1(\Omega_{L})\to H^{\frac{1}{2}}(\Gamma)$. For $u\in H^1_{y}$, the right/left traces coincide, i.e., $\gamma^{+}u=\gamma^{-}u$, which will be simply denoted as $\gamma u$.
    \item (one-side) conormal derivatives on $\Gamma$: for $u\in H^1_{y}$ satisfying the elliptic equation $-\nabla \cdot c\nabla u=f$ in the strip $S$ with $f\in H^{-1}(S)$, the right conormal derivative $\partial_{\nu,c}^{+}u\in H^{-\frac{1}{2}}(\Gamma)$ is defined in the standard way by the Green identity
    \begin{equation*}
    \langle \partial_{\nu,c}^{+}u,\phi\rangle_{\Gamma}:=-\int_{S}c\nabla u\cdot \overline{\nabla (\Xi^{+}\phi)}+\int_{S}f\cdot\overline{\Xi^{+}\phi} \quad (\forall \phi\in H^{\frac{1}{2}}(\Gamma)),
    \end{equation*}
    where the extension map $\Xi^{+}:H^{\frac{1}{2}}(\Gamma)
    \to H^{1}(S)$ satisfies $\gamma^{+}\Xi^{+}=\mathbbm{1}_{H^{\frac{1}{2}}(\Gamma)}$. If in addition $u\in H^2_{y}$, then $\partial_{\nu,c}^{+}u=\gamma^{+}(\nu\cdot c\nabla u)$. The left conormal derivative $\partial_{\nu,c}^{-}$ is defined similarly. When $\partial_{\nu,c}^{+}u=\partial_{\nu,c}^{-}u$, it will be simply denoted as $\partial_{\nu,c}u$.
\end{itemize}

\section{Setup and Main Results} \label{sec_main_result}
We start with the elliptic operator in $L^2(\mathbb{R}^2)$:
\begin{equation}
\label{eq_unperturbed_operator}
\mathcal{L}^{a}:H^2(\mathbb{R}^2) \subset L^2(\mathbb{R}^2)\to L^2(\mathbb{R}^2),\quad
\mathcal{L}^{a}u:= -\nabla\cdot a(\bm{x}) \nabla u.
\end{equation}
The coefficient function $a$ is assumed to be smooth, real-valued and uniformly positive. Moreover, we assume that $a$ is periodic with respect to a triangular lattice, and respects the $C_{3v}$ ($2\pi/3$-rotation and reflection) and inversion symmetry, as specified below. Consequently, the elliptic operator $\mathcal{L}^{a}$ inherits all symmetry properties of the function $a$.
\begin{assumption} \label{asmp_unperturbed_coef}
$a\in C^{\infty}(\mathbb{R}^2;\mathbb{R})$, and there exists $m^{a},M^{a}>0$ such that $m^{a}\leq a(\bm{x})\leq M^{a}$ for all $\bm{x}\in\mathbb{R}^2$. Moreover, letting
\begin{equation*}
\bm{e}_1:=\begin{pmatrix}
1 \\ 0    
\end{pmatrix},\quad
\bm{e}_2:=\begin{pmatrix}
-\frac{1}{2} \\ \frac{\sqrt{3}}{2}    
\end{pmatrix},\quad
R:=\begin{pmatrix}
-\frac{1}{2} & -\frac{\sqrt{3}}{2} \\
\frac{\sqrt{3}}{2} & -\frac{1}{2} 
\end{pmatrix},\quad
F:=\begin{pmatrix}
-1 & 0 \\
0 & 1
\end{pmatrix},\quad
V:=\begin{pmatrix}
-1 & 0 \\
0 & -1
\end{pmatrix},
\end{equation*}
and
\begin{equation*}
\mathcal{O}:L^{\infty}(\mathbb{R}^2)\to L^{\infty}(\mathbb{R}^2),\quad
(\mathcal{O}f)(\bm{x}):=f(O\bm{x}) \quad (O\in \{R,F,V\}),
\end{equation*}
then it holds that
\begin{equation*}
a(\bm{x}+\bm{e}_i)=a(\bm{x}),\quad \mathcal{O}a=a \quad \big(i\in\{1,2\},\, \mathcal{O}\in \{\mathcal{R},\mathcal{F},\mathcal{V}\},\, \forall \bm{x}
\in\mathbb{R}^2\big).
\end{equation*}
\end{assumption}
Physically, the operator \(\mathcal{L}^{A}\) models the propagation of time-harmonic transverse electric (TE) polarized electromagnetic waves in a two-dimensional photonic crystal (see \cite[Appendix A.4.2]{lee2019elliptic}\footnote{We refer the reader to \cite{lee2019elliptic} for a rigorous first-principle derivation of the divergence-form operator \(\mathcal{L}^A\) governing the in-plane wave propagation.}). As a classical choice in experiments, the underlying photonic crystal comprises dielectric rods that are homogeneous along the \(z\)-direction and periodically arranged within the triangular background lattice $\Lambda:=\mathbb{Z}\bm{e}_1\oplus \mathbb{Z}\bm{e}_2$. In that case, the coefficient function $a$ differs from constant only within the cross sections of these rods in the \(xy\)-plane; see Figure \ref{fig_unperturbed_lattice} for an illustration. 

\begin{figure}
\centering
\subfigure[(Unperturbed) periodic structure]{
\label{fig_unperturbed_lattice}
\begin{tikzpicture}[scale=1]
%disks
\foreach \x in {-2,-1,0,1} { 
\foreach \y in {-2,-1,0,1} {
\draw[fill=black,opacity=0.5] ({1/2+\x+\y/2},{sqrt(3)/6+sqrt(3)*\y/2}) ellipse(0.1 and 0.1);
\draw[fill=black,opacity=0.5] ({1+\x+\y/2},{sqrt(3)/3+sqrt(3)*\y/2}) ellipse(0.1 and 0.1);
}
}
%boundary
\foreach \t in {-2,2} {
\draw[dashed] ({-2+\t/2},{sqrt(3)*\t/2})--({2+\t/2},{sqrt(3)*\t/2});
\draw[dashed] ({1+\t},{sqrt(3)})--({-1+\t},{-sqrt(3)});
}
%unit cell
\draw[fill=blue,opacity=0.1] (0,0)--(1,0)--({1/2+1},{sqrt(3)/2})--({1/2},{sqrt(3)/2})--(0,0);
%lattice vector
\draw[very thick,->,red] ({0},{0})--({1},{0});
\draw[very thick,->,red] ({0},{0})--({-1/2},{sqrt(3)/2});
\node[right,red,scale=0.9] at ({1},{0}) {$\bm{e}_1$};
\node[left,red,scale=0.9] at ({-1/2},{sqrt(3)/2}) {$\bm{e}_2$};
\end{tikzpicture}
}
\subfigure[Reciprocal cell]{
\label{fig_reciprocal_cell}
%reciprocal vector
\begin{tikzpicture} [scale=1]
\draw[very thick,->,blue] ({0},{0})--({1},{sqrt(3)/2});
\draw[very thick,->,blue] ({0},{0})--({0},{2*sqrt(3)/3});
\node[right,blue,scale=0.9] at ({1},{sqrt(3)/2}) {$\bm{e}_1^{*}$};
\node[above,blue,scale=0.9] at ({0},{2*sqrt(3)/3}) {$\bm{e}_2^{*}$};
\draw[fill=blue,opacity=0.1] ({2/3},{0})--({1/3},{sqrt(3)/3})--({-1/3},{sqrt(3)/3})--({-2/3},{0})--({-1/3},{-sqrt(3)/3})--({1/3},{-sqrt(3)/3})--({2/3},{0});
%high symmetry point
\node[above,red,scale=0.6] at ({1/3},{sqrt(3)/3}) {$\bm{K}$};
\node[right,red,scale=0.6] at ({2/3},{0}) {$\bm{K}^{\prime}$};
%just to fill
\path (-1,-1) rectangle (1,-1.5);
\end{tikzpicture}
}
\subfigure[Band structure]{
\label{fig_unperturbed_band}
\includegraphics[height=4cm]{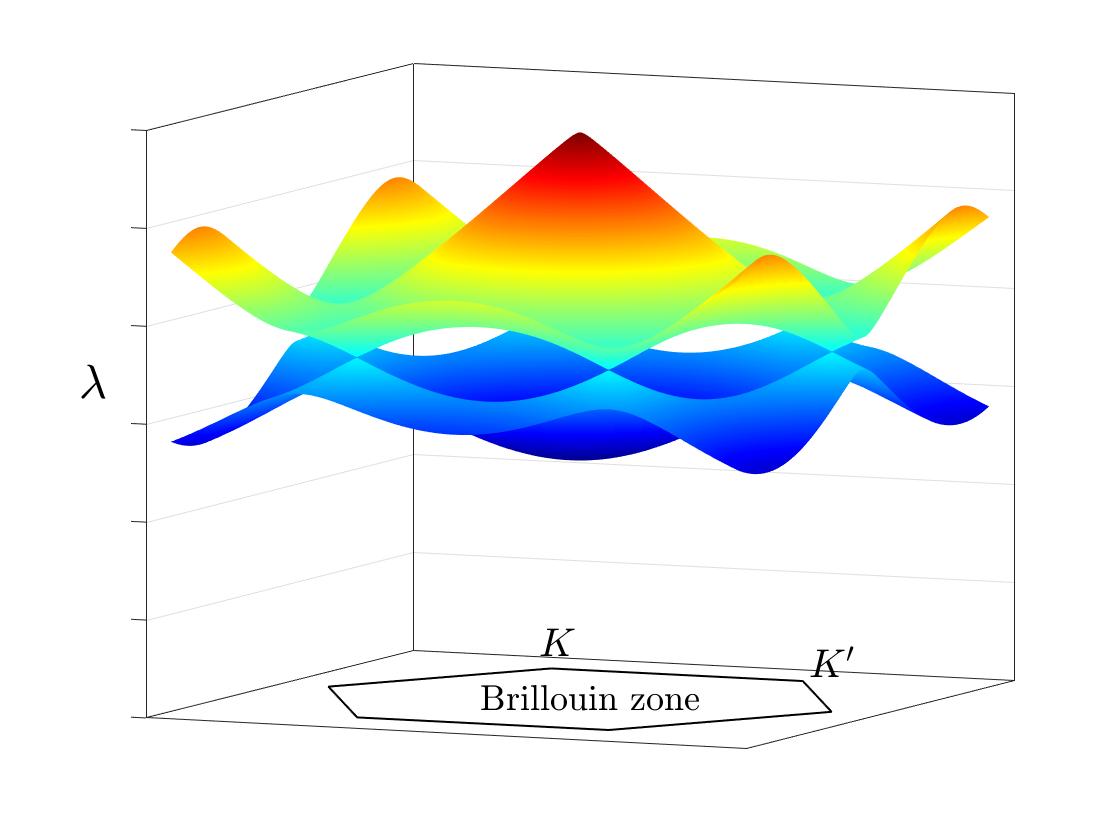} 
}
\caption{Unperturbed lattice and band structure. The inversion symmetry leads to the conic intersection between the first two bands.}
\label{fig_unperturbed_structure_and_bands}
\end{figure}
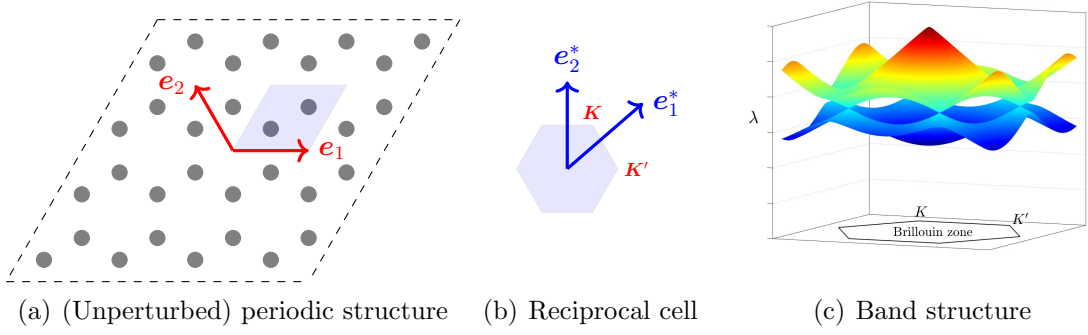

By $\Lambda$-periodicity and Floquet-Bloch theory, the spectrum of $\mathcal{L}^{a}$ is decomposed as a countable union of bands \cite{kuchment2016overview}. Specifically, we let $\Lambda^*:=\mathbb{Z}\bm{e}_1^*\oplus \mathbb{Z}\bm{e}_2^*$ be the dual lattice of $\Lambda$, where $\bm{e}_1^*:=2\pi(1,\frac{\sqrt{3}}{3})^\top$, $\bm{e}_2^*:=2\pi(0,\frac{2\sqrt{3}}{3})^\top$, whose unit cell (Brillouin zone) is taken as $Y^*:=\{s\cdot \bm{e}_1^*+t\cdot \bm{e}_2^*: -\frac{1}{2}\leq s,t\leq \frac{1}{2}\}$; see Figure \ref{fig_reciprocal_cell}. Denoting the restriction of $\mathcal{L}^a$ to $L^2_{\bm{\kappa}}(\mathbb{R}^2)$ by $\mathcal{L}^a_{\bm{\kappa}}$, where $L^2_{\bm{\kappa}}(\mathbb{R}^2):=\{f\in L_{loc}^2(\mathbb{R}^2):f(\bm{x}+\bm{e})=e^{i\bm{e}\cdot \bm{\kappa}}f(\bm{x}),\, \forall \bm{e}\in \Lambda\}$, the Floquet-Bloch theory indicates
\begin{equation*}
\text{Spec}(\mathcal{L}^{a})=\bigcup_{\bm{\kappa}\in Y^*}\text{Spec}(\mathcal{L}^{a}_{\bm{\kappa}})=\bigcup_{\bm{\kappa}\in Y^*}\{\lambda^{a}_{n}(\bm{\kappa})\in\mathbb{R}:1\leq n<\infty\}.
\end{equation*}
Here, $\{\lambda^{a}_{n}(\bm{\kappa}) \}$ are the eigenvalues of $\mathcal{L}^{a}(\bm{\kappa})$, which are piece-wisely smooth and $\Lambda^{*}$-periodic. As is well-known, the symmetry property of $\mathcal{L}^{a}$ leads to the existence of a conic intersection point, i.e. a Dirac point, between the lowest two bands at the vertex $\bm{K}:=\frac{2\pi}{3}(1,\sqrt{3})^{\top}$ of $Y^{*}$ (hence also at the other five vertices by the $C_{3v}$ symmetry). Moreover, the Floquet-Bloch eigenmodes at the Dirac point obey the rotation and parity-time symmetry, i.e., $\mathcal{R}$ and $\mathcal{V}\mathcal{C}$ with $\mathcal{C}$ being the complex conjugation. We summarize this local spectral behaviour as follows, which has been proved for generic elliptic operators $\mathcal{L}^{a}$ in \cite[Theorem 3]{lee2019elliptic}:
%$bm{K}:=\pi(1,\sqrt{3})^{\top},\quad \bm{K}^{\prime}:=F\bm{K}$
\begin{assumption}[Dirac point] \label{asmp_dirac_point}
The first two bands touch at $\bm{\kappa}=\bm{K}$, i.e.,
\begin{equation*}
\lambda^{a}_{1}(\bm{K})=\lambda^{a}_{2}(\bm{K})=:\lambda_*\in \mathbb{R}.
\end{equation*}
The two touching bands are locally conic near the intersection point $(\bm{K},\lambda_*)$ in the sense that
\begin{equation} \label{eq_conical_dispersion}
    \begin{aligned}
    \lambda_{1}^{a}(\bm{\kappa})-\lambda_{*}&=-\alpha_*|\bm{\kappa}-\bm{K}|+\mathcal{O}(|\bm{\kappa}-\bm{K}|^2), \\
    \lambda_{2}^{a}(\bm{\kappa})-\lambda_{*}&=\alpha_*|\bm{\kappa}-\bm{K}|+\mathcal{O}(|\bm{\kappa}-\bm{K}|^2),
    \end{aligned}
\end{equation}
for some $\alpha_*>0$. Moreover, $\ker (\mathcal{L}^a(\bm{K})-\lambda_*)=\text{span} \{u_{1,\bm{K}}^{a},u_{2,\bm{K}}^{a}\}$ with
\begin{equation} \label{eq_Dirac_eigenmode_symmetry}
\mathcal{R}u_{1,\bm{K}}^{a}=\tau u_{1,\bm{K}}^{a},\quad \mathcal{R}u_{2,\bm{K}}^{a}=\overline{\tau} u_{2,\bm{K}}^{a},\quad u_{2,\bm{K}}^{a}=\mathcal{V}\mathcal{C}u_{1,\bm{K}}^{a},
\end{equation}
where $\tau:=e^{i\frac{2\pi}{3}}$.
\end{assumption}
Moreover, we assume the so-called \textbf{spectral no-fold condition}, as introduced in \cite{fefferman2016honeycomb_edge,fefferman12jams}; namely, the frequency level $\lambda=\lambda_*$ intersect only with the first and second bands, and only at the six vertices of the Brillouin zone.\footnote{The spectral no-fold condition is essential for the study of interface modes bifurcating from the Dirac point. In fact, if the no-fold condition fails, the interface modes may turn into resonant modes and are no longer localized near the interface, as proved in \cite{qiu2026embedded}.}
\begin{assumption} \label{asmp_no_fold}
$\lambda_{1}^{a}(\bm{\kappa})\neq \lambda_*$, $\lambda_{2}^{a}(\bm{\kappa})\neq \lambda_*$ for any $\bm{\kappa}\in Y^*\backslash \{\bm{K},R\bm{K},R^2\bm{K},F\bm{K},RF\bm{K},R^2F\bm{K}\}$, and $\lambda_{n}^{a}(\bm{\kappa})\neq \lambda_*$ for any $n\geq 3$ and $\bm{\kappa}\in Y^*$.
\end{assumption}
Now, we consider \textbf{the effect of inversion-symmetry breaking}. This is achieved by adding an additional coefficient function $b(\bm{x})$ that anti-commutes with the operator $\mathcal{V}$. Specifically, we define
\begin{equation}
\label{eq_positive_perturbed_operator}
\mathcal{L}^{b}:H^2(\mathbb{R}^2) \subset L^2(\mathbb{R}^2)\to L^2(\mathbb{R}^2),\quad
\mathcal{L}^{b}u:= -\nabla\cdot b(\bm{x}) \nabla u,
\end{equation}
where $\delta$ represents the strength of the perturbation and $b=b(\bm{x})$ satisfies
\begin{assumption} \label{asmp_perturbed_coef}
$b\in C^{\infty}(\mathbb{R}^2;\mathbb{R})$,
\begin{equation*}
b(\bm{x}+\bm{e}_i)=b(\bm{x}),\quad \mathcal{O}b=b \quad \big(i\in\{1,2\},\, \mathcal{O}\in \{\mathcal{R},\mathcal{F}\},\, \forall \bm{x}
\in\mathbb{R}^2\big),
\end{equation*}
\begin{equation*}
\mathcal{V}b=-b,
\end{equation*}
and there exists $l>0$ such that $b(\bm{x})=0$ for $\bm{x}\in\mathbb{R}\times (-l,l)$.\footnote{This condition is imposed to guarantee there is no discontinuity in the coefficient of interface operator defined in \eqref{eq_straight_interface_operator}; as a consequence, the interface mode (generalized eigenfunction of the interface operator) is globally smooth. It is fulfilled in the case described by Figure \ref{fig_perturbed_structure_and_bands}, where the perturbation coefficient $b$ describes modification of material property of the dielectric rods.}
\end{assumption}

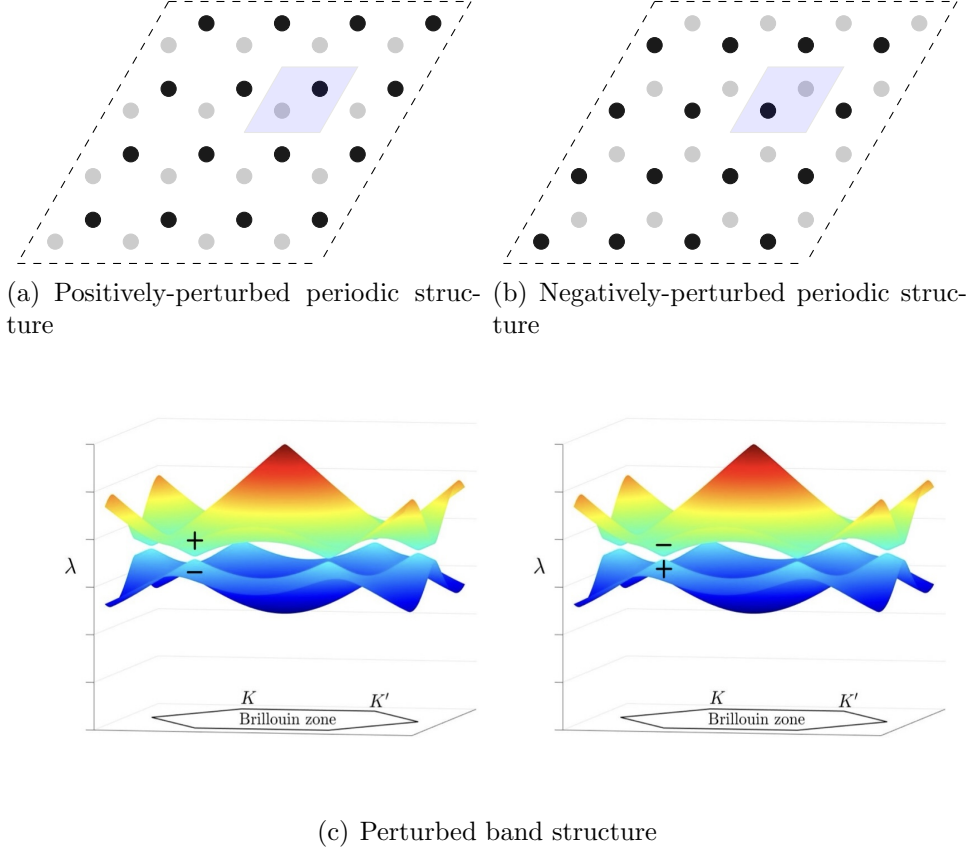
\begin{figure}
\centering
\subfigure[Positively-perturbed periodic structure]{
\label{fig_positive_perturbed_lattice}
\begin{tikzpicture}[scale=1]
%disks
\foreach \x in {-2,-1,0,1} { 
\foreach \y in {-2,-1,0,1} {
\draw[fill=black,opacity=0.2] ({1/2+\x+\y/2},{sqrt(3)/6+sqrt(3)*\y/2}) ellipse(0.1 and 0.1);
\draw[fill=black,opacity=0.9] ({1+\x+\y/2},{sqrt(3)/3+sqrt(3)*\y/2}) ellipse(0.1 and 0.1);
}
}
%boundary
\foreach \t in {-2,2} {
\draw[dashed] ({-2+\t/2},{sqrt(3)*\t/2})--({2+\t/2},{sqrt(3)*\t/2});
\draw[dashed] ({1+\t},{sqrt(3)})--({-1+\t},{-sqrt(3)});
}
%unit cell
\draw[fill=blue,opacity=0.1] (0,0)--(1,0)--({1/2+1},{sqrt(3)/2})--({1/2},{sqrt(3)/2})--(0,0);
\end{tikzpicture}
}
\subfigure[Negatively-perturbed periodic structure]{
\label{fig_negative_perturbed_lattice}
\begin{tikzpicture}[scale=1]
%disks
\foreach \x in {-2,-1,0,1} { 
\foreach \y in {-2,-1,0,1} {
\draw[fill=black,opacity=0.9] ({1/2+\x+\y/2},{sqrt(3)/6+sqrt(3)*\y/2}) ellipse(0.1 and 0.1);
\draw[fill=black,opacity=0.2] ({1+\x+\y/2},{sqrt(3)/3+sqrt(3)*\y/2}) ellipse(0.1 and 0.1);
}
}
%boundary
\foreach \t in {-2,2} {
\draw[dashed] ({-2+\t/2},{sqrt(3)*\t/2})--({2+\t/2},{sqrt(3)*\t/2});
\draw[dashed] ({1+\t},{sqrt(3)})--({-1+\t},{-sqrt(3)});
}
%unit cell
\draw[fill=blue,opacity=0.1] (0,0)--(1,0)--({1/2+1},{sqrt(3)/2})--({1/2},{sqrt(3)/2})--(0,0);
\end{tikzpicture}
}
\subfigure[Perturbed band structure]{
\label{fig_perturbed_band}
\includegraphics[height=6cm,width=14cm]{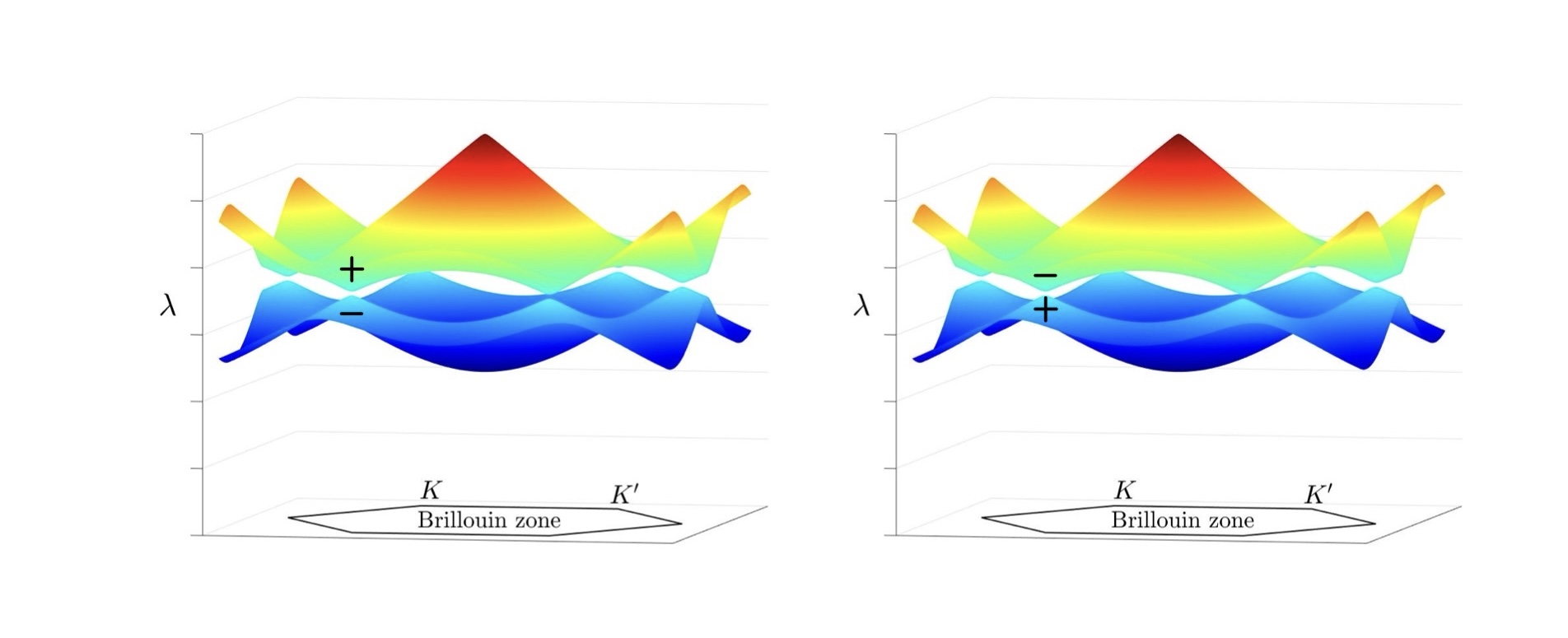}
}
\caption{Perturbed periodic structure. As shown in (a) and (b), both perturbations break the inversion symmetry, and hence, lift the spectral degeneracy at the Dirac point. Importantly, as indicated in (c), these two perturbations have distinct effects on the local eigenspace: near the Dirac point, the Floquet-Bloch eigenmode of the upper band associated with the structure (a) satisfies $\mathcal{R}u_{2,\bm{K}}^{a+b}=\tau u_{2,\bm{K}}^{a+b}$ (marked by `+'), while the lower band satisfies $\mathcal{R}u_{1,\bm{K}}^{a+b}=\overline{\tau} u_{1,\bm{K}}^{a+b}$. It is the opposite case for the structure (b).}
\label{fig_perturbed_structure_and_bands}
\end{figure}

Intuitively speaking, the effect of the symmetry-breaking operator $\mathcal{L}^{b}$ on the Floquet-Bloch eigenmodes at the Dirac point is analogous to applying an external field on electrons with opposite spins, whose energies shifts in opposite sign. More precisely, one can prove the following identities (cf., e.g., \cite[Proposition 7]{lee2019elliptic} or \cite[Lemma 7.6]{qiu2024square_lattice})
\begin{equation} \label{eq_perturbation_identity}
(\mathcal{L}^{b}u_{1,\bm{K}}^{a},u_{1,\bm{K}}^{a})_{L^2(Y)}=-(\mathcal{L}^{b}u_{2,\bm{K}}^{a},u_{2,\bm{K}}^{a})_{L^2(Y)},\quad
(\mathcal{L}^{b}u_{i,\bm{K}}^{a},u_{j,\bm{K}}^{a})_{L^2(Y)}=0\quad (i\neq j).
\end{equation}
Here, $Y:=\{s\bm{e}_1+t(\bm{e}_1+\bm{e}_2):\, s,t\in (0,1)\}$ denotes the unit cell of the lattice $\Lambda$; see the shadowed parallelogram in Figure \ref{fig_perturbed_structure_and_bands}. As a consequence of \eqref{eq_perturbation_identity}, the perturbation is expected to lift the degeneracy at the Dirac points and to open a band gap, which can be proved by a standard perturbation argument (cf., e.g., \cite[Theorem 7.1]{qiu2024square_lattice} or \cite[Section 5]{lee2019elliptic}):
\begin{proposition}
\label{prop_gap_open}
Suppose that Assumptions \ref{asmp_unperturbed_coef}-\ref{asmp_perturbed_coef} hold. Let $0<c_0<1$ be a constant close to one and assume that
\begin{equation*}
    t_*:=(\mathcal{L}^{b}u_{1,\bm{K}}^{a},u_{1,\bm{K}}^{a})_{L^2(Y)}\neq 0 .
\end{equation*}
Then, for $|\delta|$ being sufficiently small and nonzero, the operator $\mathcal{L}^{a}+\delta\cdot \mathcal{L}^{b}$
has a spectral gap $\mathcal{I}=(\lambda_*-c_0|t_*\delta|,\lambda_*+c_0|t_*\delta|)$ near $\lambda=\lambda_*$, that is, $\mathcal{I}\cap \text{Spec}(\mathcal{L}^{a}+\delta\cdot \mathcal{L}^{b})=\emptyset$.   
\end{proposition}
More importantly, as discussed following Figure \ref{fig_perturbed_band}, the perturbation operator $\mathcal{L}^{b}$ significantly changes the local eigenspace: for $\mathcal{L}^{a}+\delta\cdot \mathcal{L}^{b}$ and $\mathcal{L}^{a}-\delta\cdot \mathcal{L}^{b}$, their Floquet-Bloch eigenmodes near $\lambda=\lambda_*$ exchange parities. This is the so-called \textbf{band-inversion phenomenon} in the physics literature, a type of topological phase transition; see \cite[Section 5.1.1]{vanderbilt2018berry} for a detailed discussion. When such a phase transition occurs, one expects the existence of localized modes at the interface separating $\mathcal{L}^{a}+\delta\cdot \mathcal{L}^{b}$ and $\mathcal{L}^{a}-\delta\cdot \mathcal{L}^{b}$, as the fundamental mechanism underpinning the (photonic) valley Hall effect \cite{ma2016all_si_valley,noh2018observation_valley,vila2017observation_valley,ozawa19topological_photonics}.\footnote{The name of valley Hall effect is invented based on the fact that the parallel momenta of interface modes are near the local extrema of the bulk band (i.e., the two valleys).} To be more precise, we define
\begin{equation} \label{eq_straight_interface_operator}
\mathcal{L}^{E}:
H^2(\mathbb{R}^2)\subset L^2(\mathbb{R}^2)\to L^2(\mathbb{R}^2)
,\quad
\mathcal{L}^{E}u:= -\nabla\cdot a^{E}(\bm{x})\nabla u
\end{equation}
with
\begin{equation*}
a^{E}(\bm{x}):=\left\{
\begin{aligned}
&a(\bm{x})-\delta\cdot b(\bm{x}),\quad x_2<0, \\
&a(\bm{x})+\delta\cdot b(\bm{x}),\quad x_2>0.
\end{aligned}
\right.
\end{equation*}
Note that the coefficient function $a^E$ is smooth, as $b$ vanishes near the interface $E=\mathbb{R}\bm{e}_1$ (see Assumption \ref{asmp_perturbed_coef}). The interface modes refer to the in-gap generalized eigenfunction of the operator $\mathcal{L}^{E}$, which decays exponentially away from and is quasi-periodic along the interface $E$ (recall that $\mathcal{L}^{E}$ is a periodic operator along the direction of $E$). Specifically, we let $\mathcal{L}^{E}_{\kappa,\bm{e}_1}$ be the restriction of $\mathcal{L}^{E}$ to the space
\begin{equation} \label{eq_quasi_periodic_strip_space}
L_{\kappa,\bm{e}_1}^2:=\{u\in L^2_{loc}(\mathbb{R}^2):\, \|u\|_{L^2(S)}<\infty,\, u(\bm{x}+\bm{e}_1)=e^{i\kappa}u(\bm{x})\}.
\end{equation}
The strip $S:=\{s\bm{e}_1+t(\bm{e}_2+2\bm{e}_1):\, s\in (0,1),\, t\in \mathbb{R}\}$ is the underlying unit structure of $L_{\kappa,\bm{e}_1}^2$, as sketched in Figure \ref{fig_interface_model_and_band}, whose choice is not unique. With these notations, the appearance of interface modes is stated as follows:
\begin{theorem}[Straight-interface modes]
\label{thm_unbended_interface_mode}
Suppose that the conditions of Proposition \ref{prop_gap_open} hold. Then there exists $\delta_0,\kappa_0>0$ such that, for any $0<\delta<\delta_0$ and $|\kappa- \kappa^{\pm}|<\kappa_0 \delta$ ($\kappa^{\pm}:=\pm \bm{K}\cdot \bm{e}_1=\pm \frac{2\pi}{3}$), $\mathcal{L}^{E}_{\kappa,\bm{e}_1}$ has exactly one eigenvalue $\lambda^{E}_{\kappa}$ inside the bulk spectral gap $\mathcal{I}$. The associated eigenfunction $u^{E}_{\kappa}\in L_{\kappa,\bm{e}_1}^2$ decays exponentially as $|\bm{x}\cdot \bm{e}_2|\to \infty$. Moreover, there exists a constant $c_*=c_*(\alpha_*,t_*)\neq 0$ such that the interface eigenvalue $\lambda^{E}_{\kappa}$ admits the asymptotics
\begin{equation} \label{eq_interface_eigenvalue_asymptotics}
\lambda^{E}_{\kappa}=\lambda_*\pm c_*\cdot(\kappa- \kappa^{\pm})+o(\delta)
\end{equation}
for $|\kappa- \kappa^{\pm}|<\kappa_0 \delta$.
\end{theorem}

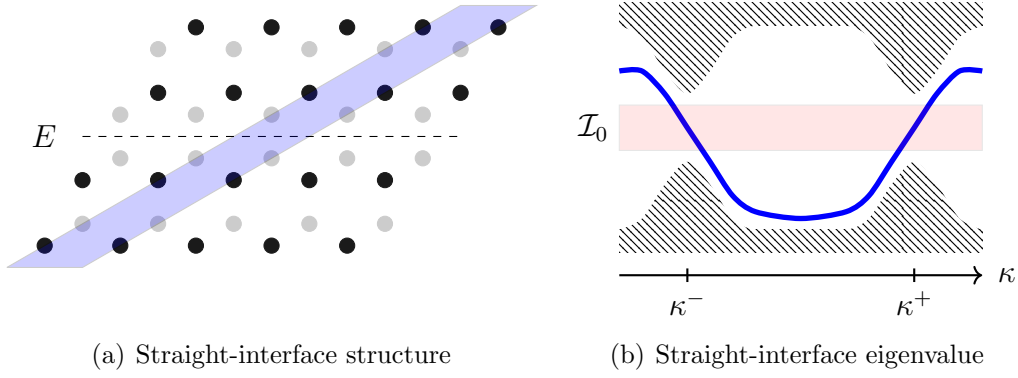
\begin{figure}
\centering
\subfigure[Straight-interface structure]{
\label{fig_straight-interface structure}
\begin{tikzpicture}[scale=1]
%disks
\foreach \x in {-2,-1,0,1,2} { 
\foreach \y in {0,1} {
\draw[fill=black,opacity=0.2] ({1/2+\x+\y/2},{sqrt(3)/6+sqrt(3)*\y/2}) ellipse(0.1 and 0.1);
\draw[fill=black,opacity=0.9] ({1+\x+\y/2},{sqrt(3)/3+sqrt(3)*\y/2}) ellipse(0.1 and 0.1);
}
}
\foreach \x in {-2,-1,0,1,2} { 
\foreach \y in {-2,-1} {
\draw[fill=black,opacity=0.9] ({1/2+\x+\y/2},{sqrt(3)/6+sqrt(3)*\y/2}) ellipse(0.1 and 0.1);
\draw[fill=black,opacity=0.2] ({1+\x+\y/2},{sqrt(3)/3+sqrt(3)*\y/2}) ellipse(0.1 and 0.1);
}
}
%boundary
\draw[dashed] (-2,0)--(3,0);
\node[left] at (-2.2,0) {$E$};
%unit strip
\draw[fill=blue,opacity=0.2] ({-3},{-sqrt(3)})--({-2},{-sqrt(3)})--({4},{sqrt(3)})--({3},{sqrt(3)})--({-3},{-sqrt(3)});
%just to fill
\path (-2,-2) rectangle (1,-2.5);
\end{tikzpicture}
}
\subfigure[Straight-interface eigenvalue]{
\label{fig_interface_eigenvalue}
\begin{tikzpicture}[scale=0.3]
%parallel momemtum
\draw[thick,->] (-8,0)--(8,0);
\draw[thick] (-5,0.3)--(-5,-0.3);
\draw[thick] (5,0.3)--(5,-0.3);
\node[below] at (-5,-0.2) {$\kappa^{-}$};
\node[below] at (5,-0.2) {$\kappa^{+}$};
\node[right] at (8.2,0) {$\kappa$};
%essential spectrum
\draw[white,line width=0pt, name path=one] plot [smooth] coordinates {(-8,11) (-7,10.5) (-6,9) (-5,8) (-4,9) (-3,10.5) (-2,11) (0,11) (2,11) (3,10.5) (4,9) (5,8) (6,9) (7,10.5) (8,11)};
\draw[white,line width=0pt, name path=two] plot [smooth] coordinates {(-8,2) (-7,2.5) (-6,4) (-5,5) (-4,4) (-3,2.5) (-2,2) (0,2) (2,2) (3,2.5) (4,4) (5,5) (6,4) (7,2.5) (8,2)};
\draw[white,line width=0pt, name path=three] (-8,11)--(8,11);
\draw[white,line width=0pt, name path=four] (-8,2)--(8,2);
\tikzfillbetween[
    of=one and three,split
  ] {pattern=north west lines};
\tikzfillbetween[
    of=two and four,split
  ] {pattern=north west lines};
\fill[pattern=north west lines] (-8,12) rectangle (8,11);
\fill[pattern=north west lines] (-8,2) rectangle (8,1);
%edge curve
\draw[blue,line width=2pt] plot [smooth] coordinates {(-8,9) (-7,9) (-6,8) (-5,6.5) (-4,5) (-3,3.5) (-2,2.8) (0,2.5) (2,2.8) (3,3.5) (4,5) (5,6.5) (6,8) (7,9) (8,9)};
\draw[fill=red,opacity=0.1] (-8,7.5) rectangle (8,5.5);
\node[left] at (-8,6.5) {$\mathcal{I}_0$};
\end{tikzpicture}
}
\caption{Underlying structure described by $\mathcal{L}^{E}$ and its spectrum. In (b), the shadowed area refers to the bulk spectrum, while the blue curve represents the in-gap interface eigenvalue, as stated in Theorem \ref{thm_unbended_interface_mode}. Note that Assumption \ref{asmp_detailed_edge_band} is satisfied in the case depicted in (b): first, the interface eigenvalue has non-vanishing derivative in the interval $\mathcal{I}_{0}$, and on the other hand, the interface eigenvalue is not absorbed into the bulk spectrum.}
\label{fig_interface_model_and_band}
\end{figure}

The proof of Theorem \ref{thm_unbended_interface_mode} is based on the impedance matching method along the interface $E$, following the lines as in \cite{qiu2024square_lattice,qiu2026waveguide_localized,li2024interface,ammari2026symmetry_protect}. Since this framework of studying interface modes is well-developed by the authors in many aspects, and it is not very important for studying the robustness, we will not repeat the proof in this paper and refer the reader to the aforementioned literature for a detailed discussion. In fact, as we will see, the only information of the straight-interface modes that is used to prove the robustness is the following:
\begin{itemize}
    \item[(i)] the number of modes (which equals two as shown in Theorem \ref{thm_unbended_interface_mode});
    \item[(ii)] the geometry of interface eigenvalue curves, including its slope and the isolation distance from the bulk spectrum (see Assumption \ref{asmp_detailed_edge_band}).
\end{itemize}
In other words, the proof of Theorem \ref{thm_unbended_interface_mode} is rather independent of the study of robustness.

The main focus of this paper is on the robustness of the interface modes introduced in Theorem \ref{thm_unbended_interface_mode} against a sharp bending of the interface by an angle of $\frac{2\pi}{3}$. Let us first characterize this bent-interface model. We still consider the divergence-form operator, whose coefficient function, as shown in Figure \ref{fig_bended_interface_structure}, equals $a-\delta b$ on the left of the bent interface $E^{bend}:=\mathbb{R}^{-}\bm{e}_1\cup \mathbb{R}^{-}(\bm{e}_1+\bm{e}_2)$ and equals $a+\delta b$ on the right side. To be precise, we define
\begin{equation*}
a^{bend}(\bm{x}):=\left\{
\begin{aligned}
&a(\bm{x})-\delta b(\bm{x}),\quad \text{$0>x_2>\sqrt{3}x_1$}, \\
&a(\bm{x})+\delta b(\bm{x}),\quad \text{otherwise,}
\end{aligned}
\right.
\end{equation*}
which is smooth across $E^{bend}$ by Assumption \ref{asmp_perturbed_coef}. Then the elliptic operator governing the bent-interface model is
\begin{equation} \label{eq_bended_interface_operator}
\mathcal{L}^{bend}:
H^2(\mathbb{R}^2)\subset L^2(\mathbb{R}^2)\to L^2(\mathbb{R}^2)
,\quad
\mathcal{L}^{bend}u:= -\nabla\cdot a^{bend}(\bm{x})\nabla u
\end{equation}
with
\begin{equation*}
\text{dom}(\mathcal{L}^{bend}):=\{u\in H^1(\mathbb{R}^2):\, -\nabla\cdot a^{bend}\nabla u\in {L^2(\mathbb{R}^2)}\} .
\end{equation*}
We note that this $\frac{2\pi}{3}-$bending is special: letting $\Omega_{L/R}$ be two half-planes separated by an auxiliary interface $\Gamma:=\mathbb{R}(2\bm{e}_1+\bm{e}_2)$, as depicted in Figure \ref{fig_bended_interface_structure},
\begin{equation*}
\Omega_{L}=\{\bm{x}\in \mathbb{R}^2:\, -\bm{x}\cdot \nu>0\},\quad
\Omega_{R}=\{\bm{x}\in \mathbb{R}^2:\, \bm{x}\cdot \nu>0\} ,\quad \nu:=\frac{1}{2}(1,-\sqrt{3})^{\top},
\end{equation*}
then, thanks to the rotation symmetry of the coefficient functions $a(\bm{x})$ and $b(\bm{x})$, it can be checked that $a^{bend}$ equals the coefficient of the straight-interface model in $\Omega_{L}$, and equals its rotation image by an angle of $\frac{2\pi}{3}$ in $\Omega_{R}$:
\begin{equation} \label{eq_medium_rotation_relation}
a^{bend}(\bm{x})=a^{E}(\bm{x}),\, \bm{x}\in \Omega_L,\quad \text{and}\quad a^{bend}(\bm{x})=a^{E}(R\bm{x}),\, \bm{x}\in \Omega_R.
\end{equation}
As we shall see in Section \ref{sec_bend_immunity}, this special rotation structure of $a^{bend}$ greatly simplifies the analysis of bent-interface modes.

We aim to show that the in-gap interface eigenvalues in Theorem \ref{thm_unbended_interface_mode} persist in the spectrum of $\mathcal{L}^{bend}$. To do that, we require more information on the interface eigenvalue associated with the straight interface, i.e., $\lambda^{E}(\kappa)$. First, as shown in \eqref{eq_interface_eigenvalue_asymptotics}, $\lambda^{E}(\kappa)$ is `approximately' linear near $\kappa=\kappa^{\pm}$.\footnote{If one can show that the $o(\delta)$ remainder, as a function of $\kappa$, attains an $o(1)$ derivative at $\kappa=\kappa^{\pm}$, then the linearity indeed holds. Unfortunately, our technique developed in \cite{li2024interface} cannot prove this point.} We will assume that it is indeed this case. On the other hand, since the operator $\mathcal{L}^{E}_{\kappa,\bm{e}_1}$ depends analytically on the quasi-momentum\footnote{See \cite[Section 3.2]{joly2016solutions} on a detailed explanation of this analytic dependence.} and $\lambda^{E}(\kappa^{\pm})$ is an isolated eigenvalue of $\mathcal{L}^{E}_{\kappa^{\pm},\bm{e}_1}$, the analytic perturbation theory indicates that $\lambda^{E}(\kappa)$ can be analytically extended to a neighborhood of $\kappa^{\pm}$ (cf. \cite[Section 3.2, Chapter 7]{kato2013perturbation}). The endpoints of the maximal interval of this extension are where $\lambda^{E}(\kappa)$ is absorbed into the essential spectrum of $\mathcal{L}^{E}_{\kappa,\bm{e}_1}$. We will assume that the maximal interval of extension is the whole period $[-\pi,\pi]$, i.e., no absorption appears. These assumptions are summarized as follows. Their validity and necessity (or unnecessity) are discussed at the end of this section. See also Figure \ref{fig_interface_eigenvalue} for an illustration of the assumed band structure.

\begin{figure}
\centering
\begin{tikzpicture}[scale=1]
%disks
%%positive bulk
\foreach \x in {-2,-1,0,1,2} { 
\foreach \y in {0,1} {
\draw[fill=black,opacity=0.2] ({1/2+\x+\y/2},{sqrt(3)/6+sqrt(3)*\y/2}) ellipse(0.1 and 0.1);
\draw[fill=black,opacity=0.9] ({1+\x+\y/2},{sqrt(3)/3+sqrt(3)*\y/2}) ellipse(0.1 and 0.1);
}
}
\foreach \x in {1,2} { 
\foreach \y in {-2,-1} {
\draw[fill=black,opacity=0.2] ({1/2+\x+\y/2},{sqrt(3)/6+sqrt(3)*\y/2}) ellipse(0.1 and 0.1);
\draw[fill=black,opacity=0.9] ({1+\x+\y/2},{sqrt(3)/3+sqrt(3)*\y/2}) ellipse(0.1 and 0.1);
}
}
%negative bulk
\foreach \x in {-2,-1,0} { 
\foreach \y in {-2,-1} {
\draw[fill=black,opacity=0.9] ({1/2+\x+\y/2},{sqrt(3)/6+sqrt(3)*\y/2}) ellipse(0.1 and 0.1);
\draw[fill=black,opacity=0.2] ({1+\x+\y/2},{sqrt(3)/3+sqrt(3)*\y/2}) ellipse(0.1 and 0.1);
}
}
%interface
\draw[dashed] (-2,0)--(1,0)--(0,{-sqrt(3)});
%imaginary matching interface
\draw[dash dot,very thick,blue] ({-2},{-sqrt(3)})--({4},{sqrt(3)});
\node[right,blue,scale=1] at ({4},{sqrt(3)}) {$\Gamma$};
\draw[very thick,red,->] ({1+3/2},{sqrt(3)/2})--({1+3/2+0.7*1/2},{sqrt(3)/2-0.7*sqrt(3)/2});
\node[red,below,scale=1] at ({1+3/2+0.7*1/2},{sqrt(3)/2-0.7*sqrt(3)/2}) {$\nu$};
\node[left,scale=1] at (-2,0) {$\Omega_{L}$};
\node[below,scale=1] at ({0},{-sqrt(3)}) {$\Omega_{R}$};
\end{tikzpicture}
\caption{The $\frac{2\pi}{3}-$bended interface model. The whole plane is splitted into two half-planes, i.e., $\Omega_{L}$ and $\Omega_{R}$, separated by an (imaginary) interface $\Gamma$. In $\Omega_{L}$, the underlying structure is same as the one described by $\mathcal{L}^{E}$, while the structure in $\Omega_{R}$ is obtained by a $\frac{2\pi}{3}-$rotation.}
\label{fig_bended_interface_structure}
\end{figure}
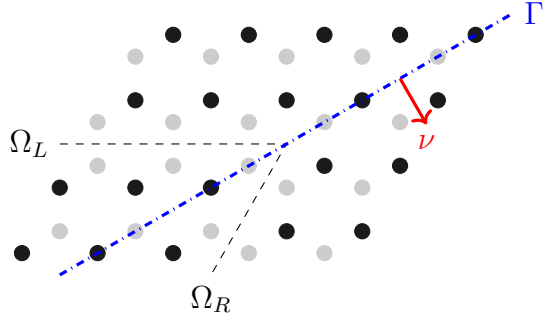

\begin{assumption} \label{asmp_detailed_edge_band}
\begin{itemize}
    \item[(i)] The interface eigenvalue has nonzero derivatives at $\kappa^{+}$:
    \begin{equation*}
    \partial_{\kappa}\lambda^{E}(\kappa^{+})\neq 0.
    \end{equation*}
    \item[(ii)] The minimal distance from $\lambda^{E}(\kappa)$ to the rest of spectrum is bounded from below, i.e.,
    \begin{equation*}
    \inf_{\kappa\in [-\pi,\pi]}\text{dist}\big(\lambda^{E}(\kappa),\text{Spec}(\mathcal{L}^{E}_{\kappa,\bm{e}_1})\backslash \{\lambda^{E}(\kappa)\} \big)=:d_*>0.
    \end{equation*}
\end{itemize}
\end{assumption}
Under Assumption \ref{asmp_detailed_edge_band}(i), there exists a neighborhood of $\lambda=\lambda_*$ in which the straight-interface eigenvalue $\lambda^{E}(\kappa)$ has non-vanishing derivatives; see Figure \ref{fig_interface_eigenvalue}. Our main result gives a complete characterization of the spectrum of the bent-interface operator $\mathcal{L}^{bend}$ within this range of frequencies.
\begin{theorem}[Bent-interface Modes] \label{thm_bent_interface_mode}
Suppose that the conditions of Theorem \ref{thm_unbended_interface_mode} and Assumption \ref{asmp_detailed_edge_band} hold. Let $\mathcal{I}_{0}\subset \mathcal{I}$ be a subinterval in which $\partial_{\kappa}\lambda^{E}$ is nonzero, i.e.,
\begin{equation*}
\partial_{\kappa}\lambda^{E}(\kappa)\neq 0,\quad \text{ for }\lambda^{E}(\kappa)\in \mathcal{I}_{0} .
\end{equation*}
Then there exists a finite set $\mathcal{I}_{0}^{exc}\subset \mathcal{I}_{0}$ that satisfies the following properties:
\begin{itemize}
    \item[(i)] Each $\lambda\in \mathcal{I}_{0}\backslash \mathcal{I}_{0}^{exc}$ is a generalized eigenvalue of $\mathcal{L}^{bend}$. Moreover, there are two linearly independent generalized eigenfunctions, denoted as $u^{bend}_{1}(\cdot;\lambda)$ and $u^{bend}_{2}(\cdot;\lambda)$, which admit the following far-field asymptotics:
    \begin{equation} \label{eq_bend_interface_mode_asymptotics}
    u^{bend}_{i}(\bm{x};\lambda)=\left\{
    \begin{aligned}
    &r_{i}^{L}u_{+}^{E}(\bm{x};\lambda)+\ell_{i}^{L}u_{-}^{E}(\bm{x};\lambda)+w^{evan,L}_{i}(\bm{x}),\quad \bm{x}\in \Omega_{L}, \\
    &r_{i}^{R}(\mathcal{R}u_{+}^{E})(\bm{x};\lambda)+\ell_{i}^{R}(\mathcal{R}u_{-}^{E})(\bm{x};\lambda)+w^{evan,R}_{i}(\bm{x}),\quad \bm{x}\in \Omega_{R},
    \end{aligned}
    \right.
    \end{equation}
    where $u_{\pm}^{E}(\cdot;\lambda)$ is the straight-interface mode (associated with $\mathcal{L}^{E}$) with frequency $\lambda\in\mathcal{I}_{0}$ and momentum $\kappa^{\pm}(\lambda)$ such that $\pm \partial_\kappa \lambda^{E}(\kappa^{\pm}(\lambda))>0$. The vector $(r_{i}^{L},\ell_{i}^{L},r_{i}^{R},\ell_{i}^{R})$ is nonzero (with an explicit formula presented later) and $w^{evan,L/R}_{i}(\bm{x})$ decays exponentially away from the corner in the sense that
    \begin{equation} \label{eq_bend_interface_mode_evan_decay_1}
\|w^{evan,L}_i\|_{H^1(S+n\bm{e}_1)}\leq e^{-c|n|},\quad \text{as $n\to-\infty$},
\end{equation}
\begin{equation} \label{eq_bend_interface_mode_evan_decay_2}
\|w^{evan,R}_i\|_{H^1(S+n(-\bm{e}_1-\bm{e}_2))}\leq e^{-c|n|},\quad \text{as $n\to \infty$},
\end{equation}
    for some $c>0$.
    \item[(ii)] The far-field profiles of $u^{bend}_{i}$ are explicitly given by
    \begin{equation} \label{eq_far_field_profile_1}
    (r_{1}^{L},\ell_{1}^{L},r_{1}^{R},\ell_{1}^{R})
    =(1,\beta^{o},-\beta\beta^{o},-\beta^{-1}),
    \end{equation}
    and
    \begin{equation} \label{eq_far_field_profile_2}
    (r_{2}^{L},\ell_{2}^{L},r_{2}^{R},\ell_{2}^{R})
    =(1,\beta^{e},\beta\beta^{e},\beta^{-1}).
    \end{equation}
    Here $\beta,\beta^{o},\beta^{e}\in \mathbb{C}$ are constants depending on $\lambda$ and satisfy $|\beta|=|\beta^{o}|=|\beta^{e}|=1$.
    \item[(iii)] For any $\lambda\in \mathcal{I}_{0}\backslash \mathcal{I}_{0}^{exc}$, $u^{bend}_{1}(\cdot;\lambda)$ and $u^{bend}_{2}(\cdot;\lambda)$ are the only generalized eigenfunctions of $\mathcal{L}^{bend}$ that satisfy the asymptotics \eqref{eq_bend_interface_mode_asymptotics}. In particular, corner-localized eigenfunctions with $r_{i}^{L}=\ell_{i}^{L}=r_{i}^{R}=\ell_{i}^{R}=0$ do not exist.
    \item[(iv)] For $\lambda\in \mathcal{I}_{0}^{exc}$, there are at most finitely many independent corner-localized eigenfunctions.
\end{itemize}
\end{theorem}

Let us briefly illustrate the above results. First, we look at the formula \eqref{eq_bend_interface_mode_asymptotics}. Physically, any wave in the left half-plane $\Omega_L$ should be decomposed as its propagating and evanescent parts, respectively. In particular, restricting the range of frequencies $\lambda\in \mathcal{I}_{0}$, there are exactly two propagating modes in $\Omega_L$, i.e., the right-going straight-interface mode $u_{+}^{E}(\cdot;\lambda)$ and the left-going mode $u_{-}^{E}(\cdot;\lambda)$. This is manifested by the first line of \eqref{eq_bend_interface_mode_asymptotics}. On the other hand, since the medium in $\Omega_{R}$ is obtained by rotating $\Omega_{L}$ by an angle of $\frac{2\pi}{3}$, it is expected that the propagating modes in $\Omega_{R}$ are the rotation image of their counterparts in $\Omega_{L}$, which is confirmed by the second line of \eqref{eq_bend_interface_mode_asymptotics}. By claims (i) and (ii), we know that, for almost every frequency, there are exactly two generalized eigenfunctions of $\mathcal{L}^{bend}$ in the bent-interface structure, which are delocalized from the corner of bending. Recalling that there are two interface modes when the interface is straight, our results clearly show that the number of interface modes persists when a sharp bending is applied to the interface. On the other hand, it is well-known in physics literature that the corner-localized modes pose a serious obstacle when using the valley-Hall system as a transport cavity in many application, and, to the best of our knowledge, there is no known mathematical result addressing their existence or non-existence. Our results, especially the claim (iii), give a definitive answer to this question: the corner-localized modes can only appear in a zero-measure set of frequencies. Moreover, as shown in claim (iv), the corner-localized modes are of at most finite multiplicity at the exceptional frequencies.

The property (ii) deserves an independent illustration because, as by which we can calculate explicitly the scattering matrix at the corner. In fact, according to \eqref{eq_far_field_profile_1}-\eqref{eq_far_field_profile_2}, we have the following relation between the amplitudes of propagating parts of the bent-interface modes
\begin{equation} \label{eq_scattering_matrix}
\begin{pmatrix}
\ell^{L} \\ r^{R}
\end{pmatrix}
=\mathfrak{S}
\begin{pmatrix}
r^{L} \\ \ell^{R}
\end{pmatrix},\quad
\mathfrak{S}:=
\begin{pmatrix}
\frac{\beta^{e}+\beta^{o}}{2} & \frac{\beta(\beta^{e}-\beta^{o})}{2} \\
\frac{\beta(\beta^{e}-\beta^{o})}{2} & \frac{\beta^2(\beta^{e}+\beta^{o})}{2}
\end{pmatrix} .
\end{equation}
The domain and range of the matrix $\mathfrak{S}$ have clear physical meaning: the vector $(r^{L}, \ell^{R})^{\top}$ represents the amplitude of in-coming waves to the corner (recall that $u_{+}^{E}$ and $\mathcal{R}u_{-}^{E}$ propagate toward to the corner; see Figure \ref{fig_bended_interface_structure_matching}), while $(\ell^{L}, r^{R})^{\top}$ stands for the out-coming parts. Moreover, it is direct to check that the scattering $\mathfrak{S}$ is unitary, confirming that there is no loss in the scattering process caused by the bending.

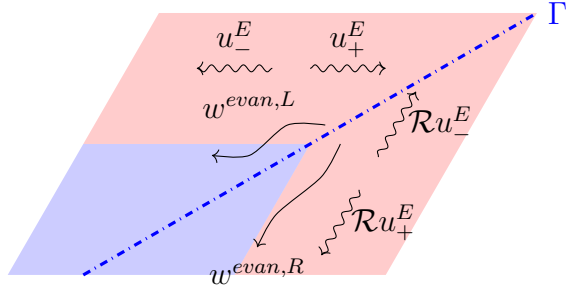
\begin{figure}
\centering
\begin{tikzpicture}[scale=1]
\tikzset{
    evan/.style   = {very thick, dashed, -{Stealth[length=3.2mm,width=2.8mm]}, draw=gray!70},
    prop/.style   = {decorate, decoration={snake, amplitude=1.3pt, segment length=6pt}},
  }
%bulks
\path [fill=red,opacity=0.2] ({-2},{0})--({1},{0})--({0},{-sqrt(3)})--({2},{-sqrt(3)})--({4},{sqrt(3)})--({-1},{sqrt(3)})--({-2},{0});
\path [fill=blue,opacity=0.2] ({-2},{0})--({1},{0})--({0},{-sqrt(3)})--({-3},{-sqrt(3)})--({-2},{0});
%imaginary matching interface
\draw[dash dot,very thick,blue] ({-2},{-sqrt(3)})--({4},{sqrt(3)});
\node[right,blue,scale=1] at ({4},{sqrt(3)}) {$\Gamma$};
%waves
%%wave at left
\draw[prop,->] (1,1)--(2,1);
\node[above,scale=1] at ({1.5},{1}) {$u_{+}^{E}$};
\draw[prop,->] (0.5,1)--(-0.5,1);
\node[above,scale=1] at ({0},{1}) {$u_{-}^{E}$};
\draw[->] plot [smooth] coordinates {(1.2,0.25) (0.7,0.25) (0.2,-0.15) (-0.3,-0.15)};
\node[above,scale=1] at ({0.2},{0.2}) {$w^{evan,L}$};
%%wave at right
\draw[prop,->] ({2-0.2*1/2},{-0.2*sqrt(3)/2})--({2+0.8*1/2},{0.8*sqrt(3)/2});
\draw[prop,->] ({2-0.7*1/2},{-0.7*sqrt(3)/2})--({2-1.7*1/2},{-1.7*sqrt(3)/2});
\draw[->] plot [smooth] coordinates {({1.4},{0}) ({1.4-0.5*1/2},{0-0.5*sqrt(3)/2}) ({0.8-0.5*1/2},{-0.5-0.5*sqrt(3)/2}) ({0.8-1*1/2},{-0.5-1*sqrt(3)/2})};
\node[below,scale=1] at ({0.8-1*1/2},{-0.5-1*sqrt(3)/2}) {$w^{evan,R}$};
\node[right,scale=1] at ({2+0.3*1/2},{0.3*sqrt(3)/2}) {$\mathcal{R}u_{-}^{E}$};
\node[right,scale=1] at ({2-1.2*1/2},{-1.2*sqrt(3)/2}) {$\mathcal{R}u_{+}^{E}$};
\end{tikzpicture}
\caption{Wave propagation in the bent-interface structure. The wavy lines represent propagating waves in the medium, while the curvy solid arrows refer to the evanescent waves localized near the corner.}
\label{fig_bended_interface_structure_matching}
\end{figure}

\begin{remark} \label{rmk_non_vanishing_slope}
We briefly remark on how to relax Assumption \ref{asmp_detailed_edge_band}(i). A straightforward extension is to allow for multiple branches of the interface eigenvalue (i.e., more than two), as shown in Figure \ref{fig_vanish_slope}, in which case the analysis of this paper can be readily applied without essential change. However, one should be careful when trying to include a frequency at which the interface eigenvalue reaches a critical point (e.g., the two local maxima in Figure \ref{fig_vanish_slope}). The reason is that our analysis of bent-interface modes is based on the out-going Green function $G^{E}$ defined by the limiting absorption principle (see Section \ref{sec_LA_out_green}), which is not applicable at the critical points of $\lambda^E(\kappa)$ (cf. \cite[Theorem 6]{joly2016solutions}). However, since $\lambda^E(\kappa)$ is (locally) analytic, the set of critical frequencies is discrete. It means that Theorem \ref{thm_bent_interface_mode} still holds if we only assume $\lambda^E(\kappa)$ is analytic within the region of interest $\mathcal{I}_{0}$, instead of Assumption \ref{asmp_detailed_edge_band}(i), for which one just needs to slightly enlarge the exceptional set of frequencies $\mathcal{I}_{0}^{exc}$ to include the critical frequencies (which are still finite). However, we maintain Assumption \ref{asmp_detailed_edge_band}(i) because it is a typical case for the band structure of the valley Hall effect.
\end{remark}

\begin{remark} \label{rmk_no_absorption}
Assumption \ref{asmp_detailed_edge_band}(ii) is imposed to simplify the analysis of the Green function $G^E$ associated with $\mathcal{L}^{E}$, especially its decay rate along the direction of the interface. In fact, by Assumption \ref{asmp_detailed_edge_band}(ii), the along-interface decay of $G^E$ is carried out using the standard method to analyze the oscillating integral of a periodic analytic function, following the lines of \cite[Section 5]{joly2016solutions}; see Section \ref{sec_LA_out_green_def} for details. In that case, the evanescent part of $G^E$ decays exponentially (cf. Proposition \ref{prop_parallel_decay_Green_function}). As a consequence, the corner-localized modes are expected to exhibit exponential decay along the interface. In contrast, if the interface eigenvalue is absorbed into the bulk spectrum, as shown in Figure \ref{fig_absorbed_spectrum}, the estimate of the off-diagonal decay rate of $G^{E}$ is much more complicated. We do not pursue this general case in this paper and leave this interesting extension for future study.
\end{remark}

\begin{figure}
\centering
\subfigure[]{
\label{fig_vanish_slope}
\begin{tikzpicture}[scale=0.3]
%parallel momemtum
\draw[thick,->] (-8,0)--(8,0);
\draw[thick] (-5,0.3)--(-5,-0.3);
\draw[thick] (5,0.3)--(5,-0.3);
\node[below] at (-5,-0.2) {$\kappa^{-}$};
\node[below] at (5,-0.2) {$\kappa^{+}$};
\node[right] at (8.2,0) {$\kappa$};
%essential spectrum
\draw[white,line width=0pt, name path=one] plot [smooth] coordinates {(-8,11) (-7,10.5) (-6,9) (-5,8) (-4,9) (-3,10.5) (-2,11) (0,11) (2,11) (3,10.5) (4,9) (5,8) (6,9) (7,10.5) (8,11)};
\draw[white,line width=0pt, name path=two] plot [smooth] coordinates {(-8,2) (-7,2.5) (-6,4) (-5,5) (-4,4) (-3,2.5) (-2,2) (0,2) (2,2) (3,2.5) (4,4) (5,5) (6,4) (7,2.5) (8,2)};
\draw[white,line width=0pt, name path=three] (-8,11)--(8,11);
\draw[white,line width=0pt, name path=four] (-8,2)--(8,2);
\tikzfillbetween[
    of=one and three,split
  ] {pattern=north west lines};
\tikzfillbetween[
    of=two and four,split
  ] {pattern=north west lines};
\fill[pattern=north west lines] (-8,12) rectangle (8,11);
\fill[pattern=north west lines] (-8,2) rectangle (8,1);
%edge curve
\draw[blue,line width=2pt] plot [smooth] coordinates {(-8,4) (-7,4.5) (-6,6) (-5,6.5) (-4,6) (-3,4.5) (-2,3) (0,2.5) (2,3) (3,4.5) (4,6) (5,6.5) (6,6) (7,4.5) (8,4)};
\draw[fill=red,opacity=0.1] (-8,6) rectangle (8,5.2);
\node[left] at (-8,5.5) {$\mathcal{I}_0$};
\end{tikzpicture}
}
\subfigure[]{
\label{fig_absorbed_spectrum}
\begin{tikzpicture}[scale=0.3]
%parallel momemtum
\draw[thick,->] (-8,0)--(8,0);
\draw[thick] (-5,0.3)--(-5,-0.3);
\draw[thick] (5,0.3)--(5,-0.3);
\node[below] at (-5,-0.2) {$\kappa^{-}$};
\node[below] at (5,-0.2) {$\kappa^{+}$};
\node[right] at (8.2,0) {$\kappa$};
%essential spectrum
\draw[white,line width=0pt, name path=one] plot [smooth] coordinates {(-8,11) (-7,10.5) (-6,9) (-5,8) (-4,9) (-3,10.5) (-2,11) (0,11) (2,11) (3,10.5) (4,9) (5,8) (6,9) (7,10.5) (8,11)};
\draw[white,line width=0pt, name path=two] plot [smooth] coordinates {(-8,2) (-7,2.5) (-6,4) (-5,5) (-4,4) (-3,2.5) (-2,2) (0,2) (2,2) (3,2.5) (4,4) (5,5) (6,4) (7,2.5) (8,2)};
\draw[white,line width=0pt, name path=three] (-8,11)--(8,11);
\draw[white,line width=0pt, name path=four] (-8,2)--(8,2);
\tikzfillbetween[
    of=one and three,split
  ] {pattern=north west lines};
\tikzfillbetween[
    of=two and four,split
  ] {pattern=north west lines};
\fill[pattern=north west lines] (-8,12) rectangle (8,11);
\fill[pattern=north west lines] (-8,2) rectangle (8,1);
%edge curve
\draw[blue,line width=2pt] plot [smooth] coordinates {(-6,9) (-5,6.5) (-4,4)};
\draw[blue,line width=2pt] plot [smooth] coordinates {(6,9) (5,6.5) (4,4)};
\draw[fill=red,opacity=0.1] (-8,7.5) rectangle (8,5.5);
\node[left] at (-8,6.5) {$\mathcal{I}_0$};
\end{tikzpicture}
}
\caption{Possible extension of Assumption \ref{asmp_detailed_edge_band}.}
\label{fig_discussion}
\end{figure}
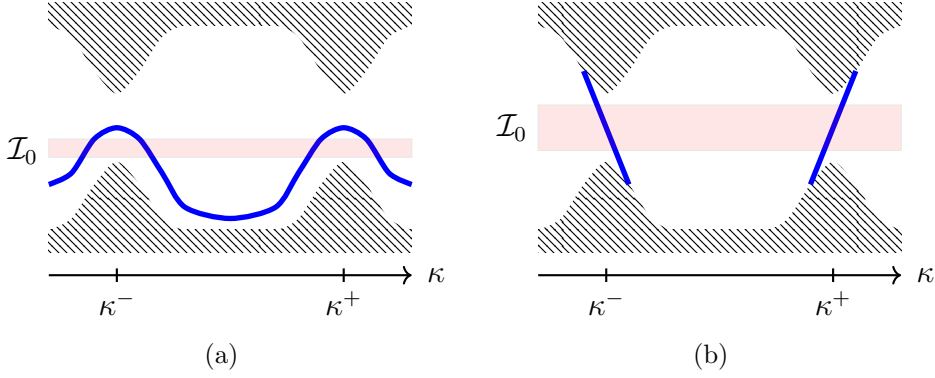

\section{Bending Immunity of Interface Modes: Main Structure of Proof} \label{sec_bend_immunity}

In this section, we present the main structure of the proof of Theorem \ref{thm_bent_interface_mode}, and leave the details to the sequel. The proof consists of two steps. In the first step, we establish a boundary integral formulation for studying the bent-interface modes. Specifically, if $u^{bend}$ takes the form of \eqref{eq_bend_interface_mode_asymptotics} with undetermined coefficients $(r^{L},\ell^{L},r^{R},\ell^{R})$ and evanescent parts $(w^{evan,L},w^{evan,R})$, then $u^{bend}$ is actually a (generalized) eigenfunction of $\mathcal{L}^{bend}$ if and only if its value and conormal derivative are continuous across the auxiliary interface $\Gamma$. By matching these data, we obtain an integral equation on $\Gamma$, whose solutions are exactly the unknown coefficients and evanescent waves. In the second step, we apply the analytic Fredholm theorem to the boundary integral operator, which leads to the conclusion that the integral equation is solvable for almost every frequency in the interval $\mathcal{I}_{0}$. In particular, as we shall see, the coefficients $(r^{L},\ell^{L},r^{R},\ell^{R})$ that solve the integral equation are nontrivial, which proves the claims (i) and (iii) in Theorem \ref{thm_bent_interface_mode}. Remarkably, we will show that these coefficients are explicitly computed from the straight-interface modes $u_{\pm}^{E}(\cdot;\lambda)$, which leads to the conclusion of property (ii). The claim (iv) is a byproduct of the Fredholm property of the boundary integral operator.

\subsection{Layer-Potential Formulation} \label{sec_LP_formulation}

Suppose that $u^{bend}\in H^2_{loc}(\mathbb{R}^2)$ solves $(\mathcal{L}^{bend}-\lambda)u^{bend}=0$ ($\lambda\in\mathcal{I}_{0}$) and takes the form of \eqref{eq_bend_interface_mode_asymptotics}, i.e.,
\begin{equation} \label{eq_bend_immunity_proof_1}
    u^{bend}(\bm{x})=\left\{
    \begin{aligned}
    &r^{L}u_{+}^{E}(\bm{x})+\ell^{L}u_{-}^{E}(\bm{x})+w^{evan,L}(\bm{x}),\quad \bm{x}\in \Omega_{L}, \\
    &r^{R}(\mathcal{R}u_{+}^{E})(\bm{x})+\ell^{R}(\mathcal{R}u_{-}^{E})(\bm{x})+w^{evan,R}(\bm{x}),\quad \bm{x}\in \Omega_{R},
    \end{aligned}
    \right.
\end{equation}
with unknown coefficients $(r^{L},\ell^{L},r^{R},\ell^{R})\in\mathbb{C}^{4}$ and evanescent parts $(w^{evan,L},w^{evan,R})$ satisfying
\begin{equation} \label{eq_bend_immunity_proof_3}
\|w^{evan,L}\|_{H^1(S+n\bm{e}_1)}\leq e^{-c|n|},\quad \text{for $n<0$},
\end{equation}
\begin{equation} \label{eq_bend_immunity_proof_4}
\|w^{evan,R}\|_{H^1(S+n(-\bm{e}_1-\bm{e}_2))}\leq e^{-c|n|},\quad \text{for $n\geq 0$},
\end{equation}
for some $c>0$. Here and afterwards, we omit the $\lambda$-dependence of the straight-interface modes $u_{\pm}^{E}(\cdot;\lambda)$ if no confusion arises. We first write the evanescent wave $w^{evan,L}$ in a layer-potential form. To this end, we introduce the out-going Green function associated with $\mathcal{L}^{E}$, defined by the limiting absorption principle, whose detailed properties are presented in Section \ref{sec_LA_out_green_def}:
\begin{equation*}
G^{E,out}(\bm{x},\bm{y};\lambda)=\lim_{\epsilon\to 0^+}G^{E}(\bm{x},\bm{y};\lambda+i\epsilon),\quad \lambda\in \mathcal{I}_{0},
\end{equation*}
where $G^{E}(\cdot,\cdot;\lambda+i\epsilon)$ is the kernel of the resolvent $(\mathcal{L}^{E}-(\lambda+i\epsilon))^{-1}$. As shown in Section \ref{sec_LA_out_green_def}, $G^{E,out}$ satisfies 
\begin{itemize}
    \item[(i)] $(\mathcal{L}^{E}-\lambda)G^{E,out}(\bm{x},\bm{y};\lambda)=\delta(\bm{x}-\bm{y})$;
    \item[(ii)] the out-going radiation condition in the sense that, as $x_1\to\infty$, $G^{E,out}(\bm{x},\bm{y};\lambda)$ converges exponentially to $u_{+}^{E}(\bm{x})\overline{u_{+}^{E}(\bm{y})}$, which is the contribution of the straight-interface mode with positive group velocity,
\begin{equation} \label{eq_bend_immunity_proof_6}
\lim_{x_1\to\infty}\Big[G^{E,out}(\bm{x},\bm{y};\lambda)-\frac{i}{|\partial_{\kappa}\lambda^{E}(\kappa^+(\lambda))|}u_{+}^{E}(\bm{x})\overline{u_{+}^{E}(\bm{y})} \Big]=0 .
\end{equation}
Similarly,
\begin{equation} \label{eq_bend_immunity_proof_7}
\lim_{x_1\to-\infty}\Big[G^{E,out}(\bm{x},\bm{y};\lambda)-\frac{i}{|\partial_{\kappa}\lambda^{E}(\kappa^-(\lambda))|}u_{-}^{E}(\bm{x})\overline{u_{-}^{E}(\bm{y})} \Big]=0 .
\end{equation}
\end{itemize}
Hence, applying the Green identity and the fact that $w^{evan,L}$ decays as $x_1\to -\infty$, we prove in Proposition \ref{prop_green_formula_left_evan} that $w^{evan,L}$ admits the following layer-potential expression.
\begin{proposition} \label{prop_layer_potential_left_evan}
For $\bm{x}\in \Omega_{L}$,
\begin{equation} \label{eq_layer_potential_left_evan_1}
w^{evan,L}(\bm{x})=-\mathcal{D}(\lambda;G^{E,out})\big[\gamma^{-}w^{evan,L} \big]+\mathcal{S}(\lambda;G^{E,out})\big[\partial_{\nu,a^{E}}^{-}w^{evan,L} \big] .
\end{equation}
Here, $\gamma^{\pm}$ and $\partial_{\nu,a^{E}}^{\pm}$ denote the trace and conormal derivatives on $\Gamma$, approaching from the right/left side
\begin{equation*}
\gamma^{\pm}f(\bm{x}):=\lim_{t\to 0^{\pm}}f(\bm{x}+t\bm{\nu}),\quad
\partial_{\nu,a^{E}}^{\pm}f(\bm{x}):=\lim_{t\to 0^{\pm}}\nu\cdot a^{E}(\bm{x}+t\bm{\nu})(\nabla f)(\bm{x}+t\bm{\nu}).
\end{equation*}
The single-layer (SL) and double-layer (DL) potentials are defined as
\begin{equation} \label{eq_layer_potential_left_evan_2}
\begin{aligned}
&\mathcal{S}(\lambda;G^{E,out})\big[\varphi\big](\bm{x}):=\int_{\Gamma}G^{E,out}(\bm{x},\bm{y};\lambda)\varphi(\bm{y})d\sigma({\bm{y}}),\quad \varphi\in H^{-\frac{1}{2}}(\Gamma),\, \bm{x}\in \mathbb{R}^2\backslash \Gamma, \\
&\mathcal{D}(\lambda;G^{E,out})\big[\phi\big](\bm{x}):=\int_{\Gamma}\partial_{\nu,a^{E},\bm{y}}G^{E,out}(\bm{x},\bm{y};\lambda)\phi(\bm{y})d\sigma({\bm{y}}),\quad \phi\in H^{\frac{1}{2}}(\Gamma),\, \bm{x}\in \mathbb{R}^2\backslash \Gamma , \\
\end{aligned}
\end{equation}
where $\partial_{\nu,a^{E},\bm{y}}$ refers to the conormal derivative with respect to the variable $\bm{y}$ (we omit the superscript when the right/left traces or the conormal derivatives coincide).
\end{proposition}
The regularity, e.g., boundedness, of the layer potential operators is presented in Section \ref{sec_LP_operators_bound_analytic_symmetry}. Importantly, by the relation \eqref{eq_medium_rotation_relation} between coefficient functions linked by rotation symmetry, the out-going Green function in $\Omega_{R}$ is obtained by a covariant transformation on $G^{E,out}$, and $w^{evan,L}$ admits a similar layer-potential expression as \eqref{eq_layer_potential_left_evan_1}.
\begin{proposition} \label{prop_layer_potential_right_evan}
For $\bm{x}\in \Omega_{R}$,
\begin{equation} \label{eq_layer_potential_right_evan_1}
w^{evan,R}(\bm{x})=\mathcal{D}(\lambda;G^{R^{-1}E,out})\big[\gamma^{+} w^{evan,R} \big]-\mathcal{S}(\lambda;G^{R^{-1}E,out})\big[\partial_{\nu,a^{E}}^{+}w^{evan,R}\big],
\end{equation}
where $\mathcal{D}(\lambda;G^{R^{-1}E,out}),\mathcal{S}(\lambda;G^{R^{-1}E,out})$ are defined similarly to \eqref{eq_layer_potential_left_evan_2} by replacing $G^{E,out}$ with $G^{R^{-1}E,out}$ given by
\begin{equation} \label{eq_green_function_rotation_relation}
G^{R^{-1}E,out}(\bm{x},\bm{y};\lambda):=G^{E,out}(R\bm{x},R\bm{y};\lambda).
\end{equation}
\end{proposition}
Here, the superscript $R^{-1}E$ indicates that the underlying structure in $\Omega_{R}$ is periodic along $R^{-1}E=\mathbb{R}(\bm{e}_1+\bm{e}_2)$. In order to match the left and right components of \eqref{eq_bend_immunity_proof_1} over $\Gamma$, we recall the following jump relations of layer-potential operators (see Proposition \ref{prop_boundary_operators_analytic_bounded_jump}):
\begin{equation} \label{eq_layer_potential_right_evan_2}
\begin{aligned}
\gamma^{\pm} \mathcal{S}(\lambda;G)[\varphi] &=S(\lambda;G)[\varphi], \\
\gamma^{\pm} \mathcal{D}(\lambda;G)[\phi] &=\big(\pm\frac{1}{2}+K(\lambda;G)\big)[\phi], \\
\partial_{\nu,a^{E}}^{\pm} \mathcal{S}(\lambda;G)[\varphi] &=\big(\mp\frac{1}{2}+K^{*}(\lambda;G)\big)[\varphi], \\
\partial_{\nu,a^{E}}^{\pm}\mathcal{D}(\lambda;G)[\phi] &=N(\lambda;G)[\phi]
\end{aligned}
\end{equation}
for $G\in\{G^{E,out},G^{R^{-1}E,out}\}$. Here $S:H^{-\frac{1}{2}}(\Gamma)\to H^{\frac{1}{2}}(\Gamma)$ (trace of the SL operator), $N:H^{\frac{1}{2}}(\Gamma)\to H^{-\frac{1}{2}}(\Gamma)$ (conormal derivative of the DL operator), $K^{*}:H^{-\frac{1}{2}}(\Gamma)\to H^{-\frac{1}{2}}(\Gamma)$ and $K:H^{\frac{1}{2}}(\Gamma)\to H^{\frac{1}{2}}(\Gamma)$ (Neumann-Poincaré operators) are all bounded, which follows from the standard layer-potential theory of elliptic operators. To make $u^{bend}$ actually an $H^{2}_{loc}$ eigenfunction of $\mathcal{L}^{bend}$, its trace and conormal derivative must be continuous across the imaginary interface $\Gamma$:
\begin{equation*}
\begin{pmatrix}
\gamma^{-} u^{bend} \\ \partial_{\nu,a^{E}}^{-} u^{bend}
\end{pmatrix}
=
\begin{pmatrix}
\gamma^{+} u^{bend} \\ \partial_{\nu,a^{E}}^{+} u^{bend}
\end{pmatrix}.
\end{equation*}
In particular, using Propositions \ref{prop_layer_potential_left_evan}-\ref{prop_layer_potential_right_evan}, \eqref{eq_bend_immunity_proof_1} and \eqref{eq_layer_potential_right_evan_2}, the above identity is rewritten as the following boundary integral equation:
\begin{equation} \label{eq_bend_immunity_proof_5}
\begin{aligned}
&r^{L}\begin{pmatrix}\gamma u_{+}^{E} \\  \partial_{\nu,a^{E}} u_{+}^{E} \end{pmatrix}
+\ell^{L}\begin{pmatrix}\gamma u_{-}^{E} \\  \partial_{\nu,a^{E}}u_{-}^{E} \end{pmatrix}    \\
&\quad +\frac{1}{2} \begin{pmatrix} \phi^{L} \\ \varphi^{L} \end{pmatrix}
+\begin{pmatrix}
-K(\lambda;G^{E,out}) & S(\lambda;G^{E,out}) \\ -N(\lambda;G^{E,out}) & K^{*}(\lambda;G^{E,out})
\end{pmatrix}
\begin{pmatrix} \phi^{L} \\ \varphi^{L} \end{pmatrix} \\
&=r^{R}\begin{pmatrix}\gamma (\mathcal{R}u_{+}^{E}) \\  \partial_{\nu,a^{E}}(\mathcal{R}u_{+}^{E}) \end{pmatrix}
+\ell^{R}\begin{pmatrix}\gamma (\mathcal{R}u_{-}^{E}) \\  \partial_{\nu,a^{E}}(\mathcal{R}u_{-}^{E}) \end{pmatrix} \\
&\quad +\frac{1}{2} \begin{pmatrix} \phi^{R} \\ \varphi^{R} \end{pmatrix}
+\begin{pmatrix}
K(\lambda;G^{R^{-1}E,out}) & -S(\lambda;G^{R^{-1}E,out}) \\ N(\lambda;G^{R^{-1}E,out}) & -K^{*}(\lambda;G^{R^{-1}E,out})
\end{pmatrix}
\begin{pmatrix} \phi^{R} \\ \varphi^{R} \end{pmatrix}
\end{aligned}
\end{equation}
and
\begin{equation} \label{eq_bend_immunity_proof_8}
\begin{aligned}
\begin{pmatrix} \phi^{L} \\ \varphi^{L} \end{pmatrix}
&=\frac{1}{2} \begin{pmatrix} \phi^{L} \\ \varphi^{L} \end{pmatrix}
+\begin{pmatrix}
-K(\lambda;G^{E,out}) & S(\lambda;G^{E,out}) \\ -N(\lambda;G^{E,out}) & K^{*}(\lambda;G^{E,out})
\end{pmatrix}
\begin{pmatrix} \phi^{L} \\ \varphi^{L} \end{pmatrix}     \\
\begin{pmatrix} \phi^{R} \\ \varphi^{R} \end{pmatrix}
&=\frac{1}{2} \begin{pmatrix} \phi^{R} \\ \varphi^{R} \end{pmatrix}
+\begin{pmatrix}
K(\lambda;G^{R^{-1}E,out}) & -S(\lambda;G^{R^{-1}E,out}) \\ N(\lambda;G^{R^{-1}E,out}) & -K^{*}(\lambda;G^{R^{-1}E,out})
\end{pmatrix}
\begin{pmatrix} \phi^{R} \\ \varphi^{R} \end{pmatrix}
\end{aligned}
\end{equation}
with
\begin{equation*}
\begin{pmatrix} \phi^{L} \\ \varphi^{L} \end{pmatrix}:=
\begin{pmatrix}
\gamma^{-} w^{evan,L} \\ \partial_{\nu,a^{E}}^{-}w^{evan,L}
\end{pmatrix},\quad
\begin{pmatrix} \phi^{R} \\ \varphi^{R} \end{pmatrix}:=
\begin{pmatrix}
\gamma^{+}w^{evan,R} \\ \partial_{\nu,a^{E}}^{+}w^{evan,R} 
\end{pmatrix}.
\end{equation*}
Moreover, since the evanescent parts decay away from $\Gamma$, the asymptotics \eqref{eq_bend_immunity_proof_6}-\eqref{eq_bend_immunity_proof_7}, the layer-potential expressions \eqref{eq_layer_potential_left_evan_1} and \eqref{eq_layer_potential_right_evan_1} lead to the following equation, which means the boundary data of $w^{evan,L/R}$ must be decoupled from the out-going interface modes
\begin{equation} \label{eq_bend_immunity_proof_9}
\int_{\Gamma}\varphi^{L}\cdot \overline{\gamma u_{-}^{E}}-\phi^{L} \cdot \overline{\partial_{\nu,a^{E}}u_{-}^{E}}=0,\quad
\int_{\Gamma}\varphi^{R}\cdot \overline{\gamma \mathcal{R}u_{+}^{E}}-\phi^{R} \cdot \overline{\partial_{\nu,a^{E}}(\mathcal{R}u_{+}^{E})}=0 .
\end{equation}
Conversely, if \eqref{eq_bend_immunity_proof_5}-\eqref{eq_bend_immunity_proof_9} are satisfied, one can check $u^{bend}$ defined by \eqref{eq_bend_immunity_proof_1} is an $H^2_{loc}$ function, satisfying $(\mathcal{L}^{bend}-\lambda)u^{bend}=0$ and the asymptotics \eqref{eq_bend_immunity_proof_3}-\eqref{eq_bend_immunity_proof_4}. Thus, we conclude that the following result holds. 
\begin{proposition} \label{prop_layer_potential_formulation}
The function $u^{bend}$ defined in \eqref{eq_bend_immunity_proof_1} satisfies (i) $u^{bend}\in H^2_{loc}(\mathbb{R}^2)$, (ii) $(\mathcal{L}^{bend}-\lambda)u^{bend}=0$, and (iii) the asymptotics \eqref{eq_bend_immunity_proof_3}-\eqref{eq_bend_immunity_proof_4} if and only if the equations \eqref{eq_bend_immunity_proof_5}-\eqref{eq_bend_immunity_proof_9} are satisfied.
\end{proposition}
Next, we simplify equations \eqref{eq_bend_immunity_proof_5}-\eqref{eq_bend_immunity_proof_9} by exploiting the symmetry of our system. Let us look at the operator $S(\lambda;G^{R^{-1}E,out})$. By the identity \eqref{eq_green_function_rotation_relation}, we have
\begin{equation*}
S(\lambda;G^{R^{-1}E,out})[\varphi](\bm{x})=\int_{\Gamma}G^{E,out}(R\bm{x},R\bm{y};\lambda)\varphi(\bm{y})d\sigma({\bm{y}})\quad (\bm{x}\in\Gamma) .
\end{equation*}
Importantly, for $\bm{x}\in \Gamma$, its rotation image $R\bm{x}$ coincides with the reflection $F\bm{x}$. This means that
\begin{equation*}
S(\lambda;G^{R^{-1}E,out})[\varphi](\bm{x})=\int_{\Gamma}G^{E,out}(F\bm{x},F\bm{y};\lambda)\varphi(\bm{y})d\sigma({\bm{y}})\quad (\bm{x}\in\Gamma) .
\end{equation*}
Finally, since $\mathcal{L}^{E}$ is reflectional symmetric, i.e., $\mathcal{L}^{E}\mathcal{F}=\mathcal{F}\mathcal{L}^{E}$, the Green function $G^{E,out}$ is covariant under reflection (see Proposition \ref{prop_reflection_covariance}):
\begin{equation*}
G^{E,out}(F\bm{x},F\bm{y};\lambda)=G^{E,out}(\bm{x},\bm{y};\lambda) .
\end{equation*}
In conclusion,
\begin{equation} \label{eq_bend_immunity_proof_10}
S(\lambda;G^{R^{-1}E,out})=S(\lambda;G^{E,out}).
\end{equation}
Similarly, for the other operators, it holds that
\begin{equation} \label{eq_bend_immunity_proof_11}
\begin{aligned}
K(\lambda;G^{R^{-1}E,out})&=-K(\lambda;G^{E,out}),\\
K^{*}(\lambda;G^{R^{-1}E,out})&=-K^{*}(\lambda;G^{E,out}),\\ N(\lambda;G^{R^{-1}E,out})&=N(\lambda;G^{E,out}),
\end{aligned}
\end{equation}
where the negative sign comes from the conormal derivative; see Proposition \ref{prop_boundary_integral_operators_relation}. On the other hand, the trace and conormal derivative of straight-interface modes are also linked by symmetry: as proved in Proposition \ref{prop_trace_symmetry}, it holds that
\begin{equation} \label{eq_bend_immunity_proof_12}
\begin{aligned}
&\gamma(\mathcal{R}u_{-}^{E})=\beta \gamma u_{+}^{E},\quad \partial_{\nu,a^{E}}(\mathcal{R}u_{-}^{E})=-\beta \partial_{\nu,a^{E}}u_{+}^{E}  ,\\
&\gamma(\mathcal{R}u_{+}^{E})=\beta^{-1} \gamma u_{-}^{E} (\bm{x}),\quad \partial_{\nu,a^{E}}(\mathcal{R}u_{+}^{E})=-\beta^{-1} \partial_{\nu,a^{E}}u_{-}^{E}
\end{aligned}
\end{equation}
for some $\beta\in \mathbb{C}$ with $|\beta|=1$. In summary, with \eqref{eq_bend_immunity_proof_10}-\eqref{eq_bend_immunity_proof_11}, we derive from Proposition \ref{prop_layer_potential_formulation} the following formulation of bent-interface modes. 
\begin{corollary}[Layer-potential formulation of bent-interface modes] \label{corol_layer_potential_formulation}
The function $u^{bend}$ defined in \eqref{eq_bend_immunity_proof_1} satisfies (i) $u^{bend}\in H^2_{loc}(\mathbb{R}^2)$, (ii) $(\mathcal{L}^{bend}-\lambda)u^{bend}=0$, and (iii) the asymptotics \eqref{eq_bend_immunity_proof_3}-\eqref{eq_bend_immunity_proof_4} if and only if
\begin{equation} \label{eq_layer_potential_formulation_1}
\begin{aligned}
&r^{L}\begin{pmatrix}\gamma u_{+}^{E} \\  \partial_{\nu,a^{E}}u_{+}^{E} \end{pmatrix}
+\ell^{L}\begin{pmatrix}\gamma u_{-}^{E} \\  \partial_{\nu,a^{E}}u_{-}^{E} \end{pmatrix}    
+\begin{pmatrix} \phi^{L} \\ \varphi^{L} \end{pmatrix}    \\
&=\beta^{-1} r^{R}\begin{pmatrix} \gamma u_{-}^{E} \\  -\partial_{\nu,a^{E}}u_{-}^{E} \end{pmatrix}
+\beta\ell^{R}\begin{pmatrix} \gamma u_{+}^{E} \\  -\partial_{\nu,a^{E}}u_{+}^{E} \end{pmatrix} 
+ \begin{pmatrix} \phi^{R} \\ \varphi^{R} \end{pmatrix},
\end{aligned}
\end{equation}

\begin{equation} \label{eq_layer_potential_formulation_2}
\begin{aligned}
\begin{pmatrix} \phi^{L} \\ \varphi^{L} \end{pmatrix}
&=\frac{1}{2} \begin{pmatrix} \phi^{L} \\ \varphi^{L} \end{pmatrix}
+\begin{pmatrix}
-K(\lambda) & S(\lambda) \\ -N(\lambda) & K^{*}(\lambda)
\end{pmatrix}
\begin{pmatrix} \phi^{L} \\ \varphi^{L} \end{pmatrix}   ,  \\
\begin{pmatrix} \phi^{R} \\ \varphi^{R} \end{pmatrix}
&=\frac{1}{2} \begin{pmatrix} \phi^{R} \\ \varphi^{R} \end{pmatrix}
+\begin{pmatrix}
-K(\lambda) & -S(\lambda) \\ N(\lambda) & K^{*}(\lambda)
\end{pmatrix}
\begin{pmatrix} \phi^{R} \\ \varphi^{R} \end{pmatrix} ,
\end{aligned}
\end{equation}
and 
\begin{equation} \label{eq_layer_potential_formulation_3}
\int_{\Gamma}\varphi^{L}\cdot \overline{\gamma u_{-}^{E}}-\phi^{L} \cdot \overline{\partial_{\nu,a^{E}}u_{-}^{E}}=0,\quad
\int_{\Gamma}\varphi^{R}\cdot \overline{\gamma u_{-}^{E}}+\phi^{R} \cdot \overline{\partial_{\nu,a^{E}}u_{-}^{E}}=0 ,
\end{equation}
where
\begin{equation*}
\begin{pmatrix} \phi^{L} \\ \varphi^{L} \end{pmatrix}:=
\begin{pmatrix}
\gamma^{-} w^{evan,L} \\ \partial_{\nu,a^{E}}^{-} w^{evan,L}
\end{pmatrix},\quad
\begin{pmatrix} \phi^{R} \\ \varphi^{R} \end{pmatrix}:=
\begin{pmatrix}
\gamma^{+} w^{evan,R} \\ \partial_{\nu,a^{E}}^{+} w^{evan,R}  
\end{pmatrix},
\end{equation*}
$O(\lambda):=O(\lambda;G^{E,out})$ for $O\in \{S,N,K,K^{*}\}$ and the constant $\beta\in \mathbb{C}$ is such that $|\beta|=1$.
\end{corollary}

\subsection{Solution of the Boundary Integral Equation} \label{sec_LP_solution}

In this section, we solve \eqref{eq_layer_potential_formulation_1}-\eqref{eq_layer_potential_formulation_3} by applying the analytic Fredholm theory to layer-potential operators. Before doing this, we first reduce the number of unknown variables, i.e., $(r^{\sigma},\ell^{\sigma})\in\mathbb{C}^{2}$ and $(\phi^{\sigma},\varphi^{\sigma})\in H^{\frac{1}{2}}(\Gamma)\times H^{-\frac{1}{2}}(\Gamma)$ with $\sigma\in \{L,R\}$, using the symmetry of our bent-interface model. As checked directly, the elliptic operator $\mathcal{L}^{bend}$ is reflectional symmetric about the imaginary interface $\Gamma$\footnote{One just verifies the coefficient function $a^{bend}(FR\bm{x})=a^{bend}(\bm{x})$ by its definition.} (see Figure \ref{fig_bended_interface_structure})
\begin{equation*}
\mathcal{L}^{bend}\mathcal{F}_{\Gamma}=\mathcal{F}_{\Gamma}\mathcal{L}^{bend},\quad \mathcal{F}_{\Gamma}:=\mathcal{R}\mathcal{F} .
\end{equation*}
As a consequence, any bent-interface mode $u^{bend}$ is classified as an even or odd function with respect to the reflection $\mathcal{F}_{\Gamma}$. Here, we show the details of solving the odd mode, whereas the even mode is treated similarly. Note that the reflection relation between straight-interface modes in Proposition \ref{prop_trace_symmetry} implies that
\begin{equation} \label{eq_solution_integral_equation_1}
(\mathcal{F}_{\Gamma}u_{\kappa^{\pm}}^{E})(\bm{x})=\beta^{\mp 1}(\mathcal{R}u_{\kappa^{\mp}}^{E})(\bm{x}),\quad \forall \bm{x}\in\mathbb{R}^2 .
\end{equation}
Hence, recalling the asymptotics \eqref{eq_bend_immunity_proof_3}-\eqref{eq_bend_immunity_proof_4}, one can calculate that
\begin{equation*}
\begin{aligned}
&\big\|u^{bend}(\cdot)-r^{L}u_{+}^{E}(\cdot)-\ell^{L}u_{-}^{E}(\cdot) \big\|_{L^2(S+n\bm{e}_1)} \to 0, \\
&\big\|u^{bend}(F_{\Gamma}\cdot)-\beta^{-1}r^{R}u_{-}^{E}(\cdot)-\beta\ell^{R}u_{+}^{E}(\cdot) \big\|_{L^2(S+n\bm{e}_1)} \to 0
\end{aligned}
\end{equation*}
as $n\to -\infty$ with $F_{\Gamma}:=F\cdot R$ (one should distinguish this identity, which acts on spatial variables, from its counterpart $\mathcal{F}_{\Gamma}=\mathcal{R}\mathcal{F}$ regarding the induced operator). Hence, by the orthogonality of the straight-interface modes in each strip $S+n\bm{e}_1$, we must have the following to satisfy $\mathcal{F}_{\Gamma}u^{bend}=-u^{bend}$:
\begin{equation} \label{eq_solution_integral_equation_2}
r^{L}=-\beta\ell^{R},\quad \ell^{L}=-\beta^{-1}r^{R}.
\end{equation}
Next, with \eqref{eq_solution_integral_equation_1} and \eqref{eq_solution_integral_equation_2}, one sees that the propagation part of $u^{bend}$ is already $\mathcal{F}_{\Gamma}$-odd. Thus, the evanescent part should also be odd, which leads to
\begin{equation} \label{eq_solution_integral_equation_3}
\begin{pmatrix} \phi^{R} \\ \varphi^{R} \end{pmatrix}=\begin{pmatrix} -\phi^{L} \\ \varphi^{L} \end{pmatrix} .
\end{equation}
In conclusion, \eqref{eq_solution_integral_equation_2} and \eqref{eq_solution_integral_equation_3} reduce the number of unknown variables by half: taking \eqref{eq_solution_integral_equation_2} and \eqref{eq_solution_integral_equation_3} into \eqref{eq_layer_potential_formulation_1}-\eqref{eq_layer_potential_formulation_3} produces a set of equations of $(r^{L},\ell^{L})\in\mathbb{C}^{2}$ and $(\phi^{L},\varphi^{L})\in H^{\frac{1}{2}}(\Gamma)\times H^{-\frac{1}{2}}(\Gamma)$; afterward, solving these equations and applying the odd extension gives the full mode $u^{bend}$.
\begin{corollary}[Layer-potential formulation of bent-interface modes with odd parity] \label{corol_layer_potential_formulation_odd}
The function $u^{bend}$ defined in \eqref{eq_bend_immunity_proof_1} satisfies (i) $u^{bend}\in H^2_{loc}(\mathbb{R}^2)$, (ii) $(\mathcal{L}^{bend}-\lambda)u^{bend}=0$, (iii) the asymptotics \eqref{eq_bend_immunity_proof_3}-\eqref{eq_bend_immunity_proof_4}, and (iv) $\mathcal{F}_{\Gamma}u^{bend}=-u^{bend}$ if and only if
\begin{equation} \label{eq_layer_potential_formulation_odd_1}
\begin{aligned}
r^{L}\begin{pmatrix}\gamma u_{+}^{E} \\  0 \end{pmatrix}
+\ell^{L}\begin{pmatrix}\gamma u_{-}^{E} \\  0 \end{pmatrix}    
+\begin{pmatrix} \phi^{L} \\ 0 \end{pmatrix}
= 0 ,
\end{aligned}
\end{equation}

\begin{equation} \label{eq_layer_potential_formulation_odd_2}
\begin{aligned}
\begin{pmatrix} \phi^{L} \\ \varphi^{L} \end{pmatrix}
=\frac{1}{2} \begin{pmatrix} \phi^{L} \\ \varphi^{L} \end{pmatrix}
+\begin{pmatrix}
-K(\lambda) & S(\lambda) \\ -N(\lambda) & K^{*}(\lambda)
\end{pmatrix}
\begin{pmatrix} \phi^{L} \\ \varphi^{L} \end{pmatrix} ,
\end{aligned}
\end{equation}

\begin{equation} \label{eq_layer_potential_formulation_odd_3}
\int_{\Gamma}\varphi^{L}\cdot \overline{\gamma u_{-}^{E}}-\phi^{L} \cdot \overline{\partial_{\nu,a^{E}}u_{-}^{E}}=0,
\end{equation}
and $u^{bend}(\bm{x})=-u^{bend}(F_{\Gamma}\bm{x})$ for $\bm{x}\in\Omega_{R}$.
\end{corollary}
To solve \eqref{eq_layer_potential_formulation_odd_1}-\eqref{eq_layer_potential_formulation_odd_3}, we first solve $(\phi^{L},\varphi^{L})$ in terms of $(r^{L},\ell^{L})$ from \eqref{eq_layer_potential_formulation_odd_1}-\eqref{eq_layer_potential_formulation_odd_2}, with which \eqref{eq_layer_potential_formulation_odd_3} becomes a linear equation of $(r^{L},\ell^{L})$. We will show that it admits a unique solution, by which we conclude that there exists a unique $\mathcal{F}_{\Gamma}$-odd bent-interface mode.

Let us look at the first row of \eqref{eq_layer_potential_formulation_odd_2}, which reads
\begin{equation} \label{eq_solution_integral_equation_4}
\big(\frac{1}{2}+K(\lambda) \big)\phi^{L}=S(\lambda)\varphi^{L}.
\end{equation}
Note that, when $S(\lambda)$ is invertible, $\varphi_{L}$ is uniquely determined by $\phi^{L}$. That is where the analytic Fredholm theorem is applied. As will be shown in Section \ref{sec_LP_operators}, the operator $S(\lambda)$ analytically continues to a neighborhood of $\mathcal{I}_{0}$, invertible in the upper-half plane and is a Fredholm operator with zero index for $\lambda\in \mathcal{I}_{0}$.
\begin{theorem} \label{thm_analytic_Fredholm_SL}
There exists a complex neighborhood $\mathcal{N} \supset \mathcal{I}_{0}$ such that
\begin{itemize}
    \item[(i)] $S(\lambda)$ extends analytically to $\mathcal{N}$, i.e., there exists an analytic operator-valued map $\lambda\ni \mathcal{N} \mapsto \tilde{S}(\lambda)\in \mathcal{B}(H^{-\frac{1}{2}}(\Gamma),H^{\frac{1}{2}}(\Gamma))$ which satisfies $\tilde{S}(\lambda)=S(\lambda)$ for $\lambda\in \mathcal{N}\cap \mathbb{R}$;
    \item[(ii)] For $\lambda\in \mathcal{N}\cap\mathbb{C}^{+}$, $\tilde{S}(\lambda)$ is invertible;
    \item[(iii)] For all $\lambda\in \mathcal{N}\cap \mathbb{R}$, $S(\lambda)$ is a Fredholm operator with zero index.
\end{itemize}
\end{theorem}
As we will show in Section \ref{sec_LA_out_green_analytic}, the analytic extension (i) is constructed by a contour integral approach, as in our previous paper \cite{qiu2026embedded}. The invertibility (ii) is standard, which follows from the well-posedness of half-space Dirichlet problem for non-real spectral parameter. However, we emphasize that one should be careful about the boundedness in (i) and the Fredholm property (iii). Although it is well-known that the trace of the single-layer potential on a \textbf{compact manifold} is a bounded Fredholm operator \cite{ammari2018mathematical_method}, it may not be true when the underlying manifold is infinite, as in the case of Theorem \ref{thm_analytic_Fredholm_SL} (recall that $\Gamma\simeq \mathbb{R}$). The possible failure of the boundedness lies in the fact that there may be extended waves along the direction of $\Gamma$, which indeed emerges if one considers the free Laplacian, for which the out-going Green operator $(-\Delta-(\lambda+i0^+))^{-1}$ contains a plane-wave projection. As a consequence, the image $S\varphi$ may be a non-integrable function on $\Gamma$. Remarkably, this does not happen in our case: the projection part in $(\mathcal{L}^{E}-(\lambda+i0^+))^{-1}$ is associated with the straight-interface modes, \textbf{which are exponentially localized in the direction of $\Gamma$}. On the other hand, the problem of Fredholmness lies in the possible infinite dimensionality of the kernel space $\ker S(\lambda)$. Nevertheless, this possibility is again ruled out: we observe that there are no half-space Dirichlet eigenfunctions (of $\mathcal{L}^{bend}\big|_{\Omega_{R}}$) supported at the infinity of $\Gamma$, where the elliptic operator is spectrally gapped (recall that $\mathcal{L}^{bend}$ tends to the gapped operator $\mathcal{L}^{a}\pm\delta\mathcal{L}^{b}$ as $\pm \bm{x}\cdot (2\bm{e}_1+\bm{e}_2)\to\infty$; see Figure \ref{fig_bended_interface_structure}). This implies that all half-space Dirichlet eigenfunctions are essentially supported near the center of $\Gamma$, and hence the finite dimensionality holds: $\dim \ker S(\lambda)<\infty$. See Section \ref{sec_LP_operators} for a detailed proof of Theorem \ref{thm_analytic_Fredholm_SL}.

By Theorem \ref{thm_analytic_Fredholm_SL}, the analytical Fredholm theorem states that $S(\lambda)$ is invertible for all $\lambda\in\mathcal{I}_{0}$, except for a finite set $\mathcal{I}_{0}^{exc,odd}$ (cf. \cite{ramm1986singularities}). Hence, for $\lambda\in \mathcal{I}_{0}\backslash \mathcal{I}_{0}^{exc,odd}$, we solve from \eqref{eq_solution_integral_equation_4} that
\begin{equation} \label{eq_solution_integral_equation_5}
\varphi^{L}=S^{-1}(\lambda)\big(\frac{1}{2}+K(\lambda) \big)\phi^{L}.
\end{equation}
Using \eqref{eq_solution_integral_equation_5} and the Calderón identity, one can check that the second row of \eqref{eq_layer_potential_formulation_odd_2} is also satisfied; see Proposition \ref{prop_calderon_identity}. Hence, by \eqref{eq_layer_potential_formulation_odd_1} and \eqref{eq_solution_integral_equation_5}, we have solved $(\phi^{L},\varphi^{L})$ in terms of $(r^{L},\ell^{L})$:
\begin{equation} \label{eq_solution_integral_equation_6}
\begin{aligned}
\phi^{L}&=-r^{L}\gamma u_{+}^{E}-\ell^{L}\gamma u_{-}^{E} , \\
\varphi^{L}&=-r^{L}S^{-1}(\lambda)\big(\frac{1}{2}+K(\lambda) \big)\gamma u_{+}^{E}-\ell^{L}S^{-1}(\lambda)\big(\frac{1}{2}+K(\lambda) \big)\gamma u_{-}^{E} .
\end{aligned}
\end{equation}
Substituting \eqref{eq_solution_integral_equation_6} into \eqref{eq_layer_potential_formulation_odd_3}, we obtain a linear equation of $(r^{L},\ell^{L})$
\begin{equation} \label{eq_solution_integral_equation_7}
\begin{aligned}
&r^{L}\Big[ \int_{\Gamma}\Big(S^{-1}(\lambda)\big(\frac{1}{2}+K(\lambda) \big)\gamma u_{+}^{E} \Big)\cdot \overline{\gamma u_{-}^{E}}-\gamma u_{+}^{E} \cdot \overline{\partial_{\nu,a^{E}}u_{-}^{E}} \Big] \\
&+\ell^{L}\Big[ \int_{\Gamma}\Big( S^{-1}(\lambda) \big(\frac{1}{2}+K(\lambda) \big)\gamma u_{-}^{E} \Big)\cdot \overline{\gamma u_{-}^{E}}-\gamma u_{-}^{E} \cdot \overline{\partial_{\nu,a^{E}}u_{-}^{E}} \Big] =0.
\end{aligned}
\end{equation}
Remarkably, as will be proved, \textbf{the coefficient of $\ell^{L}$ is exactly the group velocity of the (left-going) straight-interface mode $u_{-}^{E}$}, which is nonzero (recall Assumption \ref{asmp_detailed_edge_band}(i))
\begin{equation} \label{eq_solution_integral_equation_8}
    \int_{\Gamma}\Big( S^{-1}(\lambda) \big(\frac{1}{2}+K(\lambda) \big)\gamma u_{-}^{E} \Big)\cdot \overline{\gamma u_{-}^{E}}-\gamma u_{-}^{E} \cdot \overline{\partial_{\nu,a^{E}}u_{-}^{E}}=i\partial_{\kappa}\lambda^E(\kappa^{-}(\lambda))\neq 0,
\end{equation}
for all $\lambda\in\mathcal{I}_{0}\backslash \mathcal{I}_{0}^{exc,odd}$; see Remark \ref{rmk_d2n_map}. By \eqref{eq_solution_integral_equation_8}, we see that \eqref{eq_solution_integral_equation_7} admits a unique nontrivial solution. In conclusion, we have shown that for all $\lambda\in \mathcal{I}_{0}\backslash \mathcal{I}_{0}^{exc,odd}$, there exists a unique $\mathcal{F}_{\Gamma}$-odd bent-interface mode $u^{bend}_1(\cdot;\lambda)$, whose propagating parts are non-zero thanks to \eqref{eq_solution_integral_equation_8}. Similarly, for all $\lambda\in \mathcal{I}_{0}$ except for a finite set, one can prove that there is a unique $\mathcal{F}_{\Gamma}$-even mode $u^{bend}_2(\cdot;\lambda)$ (see Remark \ref{rmk_d2n_map} for a sketch of the proof). This proves claims (i) and (ii) of Theorem \ref{thm_bent_interface_mode}. On the other hand, for $\lambda\in \mathcal{I}_{0}^{exc,odd}$, the boundary data of any $\mathcal{F}_{\Gamma}$-odd corner-localized mode satisfy the conditions of Corollary \ref{corol_layer_potential_formulation_odd} with $r^{L}=\ell^{L}=0$, which implies
that \begin{equation} \label{eq_solution_integral_equation_9}
S(\lambda)\varphi_{L}=0 .
\end{equation}
However, as stated in Theorem \ref{thm_analytic_Fredholm_SL}, $S(\lambda)$ is Fredholm with zero index for all $\lambda\in\mathcal{I}_{0}$, which implies that there are only finitely many independent solutions to \eqref{eq_solution_integral_equation_9}. This concludes that there are at most finitely many $\mathcal{F}_{\Gamma}$-odd corner-localized modes for $\lambda\in \mathcal{I}_{0}^{exc,odd}$. The number of $\mathcal{F}_{\Gamma}$-even corner-localized modes is discussed similarly, which then proves the claim (iv) of Theorem \ref{thm_bent_interface_mode}. Before turning to the proof of property (ii), it is instructive to pause and comment on the arguments we have derived in this section.

\begin{remark} \label{rmk_d2n_map}
We briefly illustrate the derivation of the remarkable identity \eqref{eq_solution_integral_equation_8}, which consists of two steps. First, as proved in Proposition \ref{prop_projection_identity}, it holds that
\begin{equation} \label{eq_d2n_map_1}
S^{-1}(\lambda)\big(\frac{1}{2}+K(\lambda) \big)\gamma u_{-}^{E} = \partial_{\nu,a^{E}}u_{-}^{E} .
\end{equation}
We can understand \eqref{eq_d2n_map_1} as follows. For $\lambda\notin\mathbb{R}$, the (analytic extension of) operator at the left side is exactly the Dirichlet-to-Neumann map associated with the operator $\mathcal{L}^{bend}-\lambda$ on the left half-plane $\Omega_{L}$.\footnote{See \cite[Exercise 7.7]{mclean2000strongly}.} Moreover, when considering $\lambda\in\mathbb{R}$ so that $S^{-1}(\lambda)$ is well-defined, the operator in \eqref{eq_d2n_map_1} still plays the role of Dirichlet-to-Neumann map, but only selecting wave that satisfies the out-going radiation condition: for the unique out-going wave in $\Omega_{L}$, i.e., the left-propagating straight-interface mode $u_{-}^{E}$, it maps the trace $\gamma u_{-}^{E}$ to the conormal derivative $\partial_{\nu,a^{E}}u_{-}^{E}$, which explains the validity of \eqref{eq_d2n_map_1}. Hence, we have seen that
\begin{equation} \label{eq_d2n_map_2}
\begin{aligned}
&\int_{\Gamma}\Big( S^{-1}(\lambda) \big(\frac{1}{2}+K(\lambda) \big)\gamma u_{-}^{E} \Big)\cdot \overline{\gamma u_{-}^{E}}-u_{-}^{E} \cdot \overline{\partial_{\nu,a^{E}}u_{-}^{E}} \\
&=\int_{\Gamma}\partial_{\nu,a^{E}}u_{-}^{E}\cdot \overline{\gamma u_{-}^{E}}-\gamma u_{-}^{E} \cdot \overline{\partial_{\nu,a^{E}}u_{-}^{E}} .
\end{aligned}
\end{equation}
Moreover, the sesquilinear product at the right side is well-understood: it equals the energy flux, or equivalently, the group velocity (since we have dropped all units), of the propagating mode $u_{-}^{E}$:
\begin{equation} \label{eq_d2n_map_3}
\int_{\Gamma}\partial_{\nu,a^{E}}u_{-}^{E}\cdot \overline{\gamma u_{-}^{E}}-\gamma u_{-}^{E} \cdot \overline{\partial_{\nu,a^{E}}u_{-}^{E}} = i\partial_{\kappa}\lambda^E(\kappa^{-}(\lambda)) ;
\end{equation}
see Proposition \ref{prop_energy_flux}. This concludes the proof of \eqref{eq_solution_integral_equation_8}.
\end{remark}

\begin{remark} \label{rmk_even_mode}
The analysis in this section can be carried out for studying $\mathcal{F}_{\Gamma}$-even bent-interface modes, with only slight modification. First, we solve the following equations, compared with those in Corollary \ref{corol_layer_potential_formulation_odd}, and then apply the even extension:
\begin{equation} \label{eq_layer_potential_formulation_even_1}
\begin{aligned}
r^{L}\begin{pmatrix}  0 \\  \partial_{\nu,a^{E}}u_{+}^{E} \end{pmatrix}
+\ell^{L}\begin{pmatrix} 0 \\  \partial_{\nu,a^{E}}u_{-}^{E} \end{pmatrix}    
+\begin{pmatrix} 0 \\ \varphi^{L} \end{pmatrix}
= 0 ,
\end{aligned}
\end{equation}

\begin{equation} \label{eq_layer_potential_formulation_even_2}
\begin{aligned}
\begin{pmatrix} \phi^{L} \\ \varphi^{L} \end{pmatrix}
=\frac{1}{2} \begin{pmatrix} \phi^{L} \\ \varphi^{L} \end{pmatrix}
+\begin{pmatrix}
-K(\lambda) & S(\lambda) \\ -N(\lambda) & K^{*}(\lambda)
\end{pmatrix}
\begin{pmatrix} \phi^{L} \\ \varphi^{L} \end{pmatrix} , 
\end{aligned}
\end{equation}
and
\begin{equation} \label{eq_layer_potential_formulation_even_3}
\int_{\Gamma}\varphi^{L}\cdot \overline{\gamma u_{-}^{E}}-\phi^{L} \cdot \overline{\partial_{\nu,a^{E}}u_{-}^{E}}=0.
\end{equation}
In this step, we solve $\phi^L$ in terms of $\varphi^L$ from the second row of \eqref{eq_layer_potential_formulation_even_2}:
\begin{equation} \label{eq_layer_potential_formulation_even_4}
\phi^{L}=N^{-1}(\lambda) \big(-\frac{1}{2}+K^{*}(\lambda) \big)\varphi^{L}.
\end{equation}
Note that the almost-everywhere invertibility of $N(\lambda)$ (conormal derivative of the DL operator) is proved similarly as in Theorem \ref{thm_analytic_Fredholm_SL}.\footnote{The key insight lies in the finite-dimensionality of Neumann eigenfunctions, similar to the Dirichlet eigenfunctions discussed in Section \ref{sec_LP_fredholm}.} Substituting \eqref{eq_layer_potential_formulation_even_1} and \eqref{eq_layer_potential_formulation_even_4} into \eqref{eq_layer_potential_formulation_even_3}, we obtain the following linear equation of $(r^{L},\ell^{L})$:
\begin{equation} \label{eq_layer_potential_formulation_even_5}
\begin{aligned}
&r^{L}\Big[\int_{\Gamma}\partial_{\nu,a^{E}}u_{+}^{E} \cdot \overline{\gamma u_{-}^{E}}-\Big( N^{-1}(\lambda) \big(-\frac{1}{2}+K^{*}(\lambda) \big) \gamma u_{+}^{E} \Big) \cdot \overline{\partial_{\nu,a^{E}}u_{-}^{E}} \Big] \\
&+\ell^{L}\Big[ \int_{\Gamma}\partial_{\nu,a^{E}}u_{-}^{E} \cdot \overline{\gamma u_{-}^{E}}-\Big( N^{-1}(\lambda) \big(-\frac{1}{2}+K^{*}(\lambda) \big) \gamma u_{-}^{E} \Big)\cdot \overline{\partial_{\nu,a^{E}}u_{-}^{E}} \Big] =0.
\end{aligned}
\end{equation}
Equation \eqref{eq_layer_potential_formulation_even_5} possesses a unique nontrivial solution, guaranteed by the following identity:
\begin{equation*}
\Big[ \int_{\Gamma}\partial_{\nu,a^{E}}u_{-}^{E} \cdot \overline{\gamma u_{-}^{E}}-\Big( N^{-1}(\lambda) \big(-\frac{1}{2}+K^{*}(\lambda) \big) \gamma u_{-}^{E} \Big)\cdot \overline{\partial_{\nu,a^{E}}u_{-}^{E}} \Big] = i\partial_{\kappa}\lambda^E(\kappa^{-}(\lambda)) ,
\end{equation*}
which is also proved by Propositions \ref{prop_energy_flux} and \ref{prop_projection_identity}. In conclusion, for all $\lambda\in \mathcal{I}_{0}$ except for a finite set, there is a unique $\mathcal{F}_{\Gamma}$-even mode. Moreover, using the Fredholmness of $N(\lambda)$, there are at most finitely many $\mathcal{F}_{\Gamma}$-even corner-localized modes within the exceptional set of frequencies. The details are left to the reader.
\end{remark}

\subsection{Calculation of the Far-field Profiles}

In this section, we calculate the far-field profiles of the bent-interface modes, i.e., property (ii) of Theorem \ref{thm_bent_interface_mode}. Let $u_{1}^{bend}$ be the bent-interface mode with odd parity, which is obtained in the previous section. Recalling that its far-field profile in the left-half plane $\Omega_L$ is characterized as
\begin{equation} \label{eq_far_field_profile_proof_1}
\big\|u^{bend}_1(\cdot)-r^{L}_1 u_{+}^{E}(\cdot)-\ell^{L}_1 u_{-}^{E}(\cdot) \big\|_{L^2(S-n\bm{e}_1)} \to 0
\end{equation}
as $n\to \infty$. On the other hand, the profile of $u^{bend}_1$ in $\Omega_{R}$ is obtained by applying the odd extension about the auxiliary interface $\Gamma$. Instead of solving the amplitudes $r^{L}_1$ and $\ell^{L}_1$ from equation \eqref{eq_solution_integral_equation_7}, which is hard as it is difficult to determine the first coefficient in \eqref{eq_solution_integral_equation_7}, we now apply the energy-flux identity to show
\begin{equation} \label{eq_far_field_profile_proof_2}
|r^{L}_1|=|\ell^{L}_1| .
\end{equation}
This, together with the fact $|r^{L}_1|+|\ell^{L}_1|\neq 0$ since we have shown that \eqref{eq_solution_integral_equation_7} admits a unique solution, concludes the proof of \eqref{eq_far_field_profile_1}. To see \eqref{eq_far_field_profile_proof_2}, we apply the Gauss-Green identity and obtain
\begin{equation} \label{eq_far_field_profile_proof_3}
\begin{aligned}
&\int_{\Gamma}\big(\nu\cdot a^{E}\nabla u^{bend}_1\big)\cdot\overline{u^{bend}_1}-u^{bend}_1\cdot \overline{\big(\nu\cdot a^{E}\nabla u^{bend}_1\big)} \\
&\quad -\int_{\Gamma_{-N}}\big(\nu\cdot a^{E}\nabla u^{bend}_1\big)\cdot\overline{u^{bend}_1}-u^{bend}_1\cdot \overline{\big(\nu\cdot a^{E}\nabla u^{bend}_1\big)} \\
&=\int_{\bigcup_{-N\leq n\leq -1}S_{n}} (\mathcal{L}^{E}-\lambda)u^{bend}_1\cdot\overline{u^{bend}_1}-u^{bend}_1\cdot(\mathcal{L}^{E}-\lambda)\overline{u^{bend}_1} \\
&=0,
\end{aligned}
\end{equation}
where $S_{n}:=S+n\bm{e}_1$, i.e. translation of the unit strip along $E=\mathbb{R}\bm{e}_1$, and $\Gamma_{n}:=\Gamma+n\bm{e}_1$. The boundary integrals on the left side are exactly the energy flux on the faces $\Gamma$ and $\Gamma_{-N}$ carried by $u^{bend}_1$, which has already appeared in \eqref{eq_d2n_map_2} and will be discussed in more detail in Section \ref{sec_prelim}. On the one hand, since $u^{bend}_1$ vanishes on $\Gamma$ as it is $\mathcal{F}_{\Gamma}$-odd, we have
\begin{equation}  \label{eq_far_field_profile_proof_4}
\int_{\Gamma}\big(\nu\cdot a^{E}\nabla u^{bend}_1\big)\cdot\overline{u^{bend}_1}-u^{bend}_1\cdot \overline{\big(\nu\cdot a^{E}\nabla u^{bend}_1\big)}=0.
\end{equation}
On the other hand, as $N\to\infty$, the evanescent part of $u^{bend}_1$ vanishes and only the propagating parts contribute to the energy flux:
\begin{equation} \label{eq_far_field_profile_proof_5}
\begin{aligned}
&\lim_{N\to\infty}\int_{\Gamma_{-N}}\big(\nu\cdot a^{E}\nabla u^{bend}_1\big)\cdot\overline{u^{bend}_1}-u^{bend}_1\cdot \overline{\big(\nu\cdot a^{E}\nabla u^{bend}_1\big)} \\
&=\lim_{N\to\infty}
\Big[ |r_{1}^{L}|^2 \int_{\Gamma_{-N}}\big(\nu\cdot a^{E}\nabla u_{+}^{E}\big)\cdot\overline{u_{+}^{E}}-u_{+}^{E}\cdot \overline{\big(\nu\cdot a^{E}\nabla u_{+}^{E}\big)} \\
&\quad\quad\quad\quad +|\ell_{1}^{L}|^2 \int_{\Gamma_{-N}}\big(\nu\cdot a^{E}\nabla u_{-}^{E}\big)\cdot\overline{u_{-}^{E}}-u_{-}^{E}\cdot \overline{\big(\nu\cdot a^{E}\nabla u_{-}^{E}\big)} \\
&\quad\quad\quad\quad +r_{1}^{L}\overline{\ell_{1}^{L}} \int_{\Gamma_{-N}}\big(\nu\cdot a^{E}\nabla u_{+}^{E}\big)\cdot\overline{u_{-}^{E}}-u_{+}^{E}\cdot \overline{\big(\nu\cdot a^{E}\nabla u_{-}^{E}\big)} \\
&\quad\quad\quad\quad +\ell_{1}^{L}\overline{r_{1}^{L}} \int_{\Gamma_{-N}}\big(\nu\cdot a^{E}\nabla u_{-}^{E}\big)\cdot\overline{u_{+}^{E}}-u_{-}^{E}\cdot \overline{\big(\nu\cdot a^{E}\nabla u_{+}^{E}\big)}
\Big].
\end{aligned}
\end{equation}
Within these four integrals, as shown in Proposition \ref{prop_energy_flux}, the energy fluxes carried by $u_{+}^{E}$ and $u_{-}^{E}$ are nonzero, equaling to the group velocity of straight-interface modes, and have opposite signs:
\begin{equation*}
\begin{aligned}
&\int_{\Gamma_{-N}}\big(\nu\cdot a^{E}\nabla u_{+}^{E}\big)\cdot\overline{u_{+}^{E}}-u_{+}^{E}\cdot \overline{\big(\nu\cdot a^{E}\nabla u_{+}^{E}\big)} \\
&=i\partial_{\kappa}\lambda^{E}\big(\kappa^{+}(\lambda)\big) \\
&=-i\partial_{\kappa}\lambda^{E}\big(\kappa^{+}(\lambda)\big) \\
&=\int_{\Gamma_{-N}}\big(\nu\cdot a^{E}\nabla u_{-}^{E}\big)\cdot\overline{u_{-}^{E}}-u_{-}^{E}\cdot \overline{\big(\nu\cdot a^{E}\nabla u_{-}^{E}\big)}
\end{aligned}
\end{equation*}
while the crossing terms vanish, i.e.,
\begin{equation*}
\begin{aligned}
&\int_{\Gamma_{-N}}\big(\nu\cdot a^{E}\nabla u_{+}^{E}\big)\cdot\overline{u_{-}^{E}}-u_{+}^{E}\cdot \overline{\big(\nu\cdot a^{E}\nabla u_{-}^{E}\big)} \\
&=\int_{\Gamma_{-N}}\big(\nu\cdot a^{E}\nabla u_{-}^{E}\big)\cdot\overline{u_{+}^{E}}-u_{-}^{E}\cdot \overline{\big(\nu\cdot a^{E}\nabla u_{+}^{E}\big)} \\
&=0.
\end{aligned}
\end{equation*}
Hence, by \eqref{eq_far_field_profile_proof_3}-\eqref{eq_far_field_profile_proof_5}, we conclude that
\begin{equation*}
|r_1^{L}|^2-|\ell_1^{L}|^2=0 ,
\end{equation*}
which leads to \eqref{eq_far_field_profile_proof_2}. The calculation of the far-field profile of $\mathcal{F}_{\Gamma}$-even mode, i.e., the identity \eqref{eq_far_field_profile_2}, is similar and omitted here.

\section{Further Discussions} \label{sec_further_discussion}

As already pointed out in Section \ref{sec_relation_previous}, the $\frac{2\pi}{3}$ bending is deemed as the most promising setup to realize interface modes that are immune against bending in the physics literature. Despite the various existing informal explanations on this phenomenon, we have shown clearly the reason from a mathematical point of view in Section \ref{sec_LP_formulation}-\ref{sec_LP_solution}: after the $\frac{2\pi}{3}$ bending, the bent-interface model is reflectional symmetric about the corner, i.e.,
\begin{equation*}
\mathcal{L}^{bend}\mathcal{F}_{\Gamma}=\mathcal{F}_{\Gamma}\mathcal{L}^{bend}.
\end{equation*}
As a consequence, the Green function and the propagating wave in the right plane $\Omega_{R}$ are obtained by the $\frac{2\pi}{3}$ rotation of their counterparts in $\Omega_{L}$. This surprising relation not only greatly simplifies the analysis, and more importantly, it is the genuine reason that guarantees the corner scattering matrix is of full rank, which leads to the bending-immunity of interface modes.

Nevertheless, we need to point out that the valley-Hall interface modes are also experimentally shown to be robust against other angles of bending, e.g. $\frac{\pi}{3}$ \cite{xue2021device_bend_immune_3}, in which case the coefficient function should be

\begin{equation} \label{eq_further_discussion_1}
a^{bend}(\bm{x})=\left\{
\begin{aligned}
&a(\bm{x})-\delta b(\bm{x}),\quad \text{$x_2<\min\{0,\sqrt{3}x_1\}$}, \\
&a(\bm{x})+\delta b(\bm{x}),\quad \text{otherwise.}
\end{aligned}
\right.
\end{equation}
See Figure \ref{fig_60_degree_bending} for an illustration. Unfortunately, the analysis in Section \ref{sec_bend_immunity} does not apply directly for that case: when \eqref{eq_further_discussion_1} holds, it is impossible to divide $\mathbb{R}^2$ into two half planes, $\Omega_{L}$ and $\Omega_{R}$, where $a^{bend}\cdot\mathbbm{1}_{\Omega_{R}}$ is obtained by the $\frac{\pi}{3}$-rotation of $a^{bend}\cdot\mathbbm{1}_{\Omega_{L}}$. Due to this reason, the extension of Theorem \ref{thm_bent_interface_mode} to general bending angles is not a routine technical modification of the present proof but require new ideas. We leave this interesting problem for future study.

On the other hand, we point out that the analysis of bending-immunity in this paper only applies for a single bending. It will be very interesting to extend our method to a multiple-bending interface, as shown in Figure \ref{fig_two_bend}. In that case, the interaction of evanescent waves between adjacent interfaces must be taken into account. More interestingly, if one can show that the interface modes are robust against three consecutive bending, as depicted in Figure \ref{fig_circuit}, then it follows that the valley-Hall interface modes are long-lived in a closed circuit. As the existence of robust boundary-localized modes in closed domains has only been proved for systems with a nontrivial topological character \cite{qiu2025bec_finite}, a similar result for topologically trivial systems will provide a strong theoretical basis for using valley-Hall material in many applications, e.g., chiral quantum optics \cite{Hafezi2020chiral_optic}. We also leave this challenging problem for future study.

\begin{figure}
\centering
\subfigure[$\pi/3$-bending interface]{
\label{fig_60_degree_bending}
\begin{tikzpicture}[scale=1]
%bulks
\path [fill=red,opacity=0.2] ({-2},{0})--({1},{0})--({2},{-sqrt(3)})--({4},{sqrt(3)})--({-1},{sqrt(3)})--({-2},{0});
\path [fill=blue,opacity=0.2] ({-2},{0})--({1},{0})--({2},{-sqrt(3)})--({-3},{-sqrt(3)})--({-2},{0});
\end{tikzpicture}
}
\subfigure[consecutive-bending structure]{
\label{fig_two_bend}
\begin{tikzpicture}[scale=1.3]
\path [fill=blue,opacity=0.2] ({-2},{0})--({0},{0})--({-1},{-sqrt(3)})--({1},{-sqrt(3)})--({1},{-sqrt(3)})--({1},{-sqrt(3)-1/2})--({-2},{-sqrt(3)-1/2})--({-2},{0});
\path [fill=red,opacity=0.2] ({-2},{0})--({0},{0})--({-1},{-sqrt(3)})--({1},{-sqrt(3)})--({1},{-sqrt(3)})--({1},{1/2})--({-2},{1/2})--({-2},{0});
\end{tikzpicture}
}
\subfigure[closed circuit formed by three bending]{
\label{fig_circuit}
\begin{tikzpicture}[scale=1.3]
\path [fill=blue,opacity=0.2] ({-2},{0})--({0},{0})--({-1},{-sqrt(3)})--({-2},{0});
\path [fill=red,opacity=0.2] ({-2-1/2},{1/2})--({1/2},{1/2})--({1/2},{-sqrt(3)-1/2})--({-2-1/2},{-sqrt(3)-1/2})--({-2-1/2},{1/2});
\end{tikzpicture}
}
\caption{Various bending-interface structures.}
\label{fig_multiple_bending}
\end{figure}
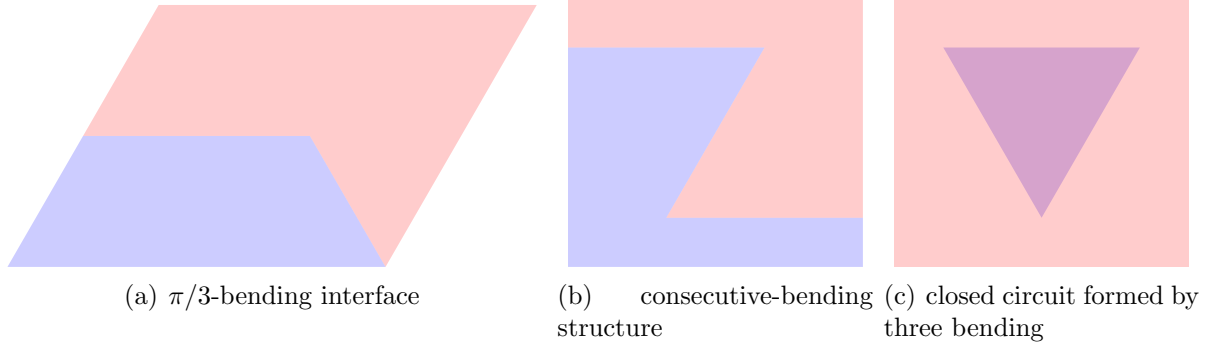

\section{Some Properties of Straight-Interface Modes} \label{sec_prelim}

In this section, we list some properties of the straight-interface modes $u_{\pm}^{E}$, focusing on their data on the imaginary interface $\Gamma$. Specifically, we will link these data with the group velocities of $u_{\pm}^{E}$ (e.g., identity \eqref{eq_d2n_map_3}), and show how these data are correlated by the symmetry of $u_{\pm}^{E}$ (i.e., \eqref{eq_bend_immunity_proof_12}).

We achieve the first purpose by introducing the energy-flux sesquilinear form associated with the ($\bm{e}_1$-)periodic operator $\mathcal{L}^{E}$. To this end, we define the following space:
\begin{equation*}
H^{n}_{y}:=\big\{u\in H^{n}_{loc}(\mathbb{R}^2):\, \sup_{k\in\mathbb{Z}}\|u\|_{H^{n}(S_{k})}<\infty \big\}, \quad n\in\mathbb{N},
\end{equation*}
The functions in $H^{n}_{y}$ are $H^n$-bounded in the $y$-direction, as manifested by the subscript. In particular, the straight-interface modes $u_{\pm}^{E}\in H^2_{y}$. With these notions, the energy-flux form on $\Gamma_n$ is defined as
\begin{equation} \label{eq_def_energy_flux}
\mathfrak{a}(u,v;n):=\int_{\Gamma_n}\big(\nu\cdot a^{E}\nabla u\big)\cdot\overline{v}-u\cdot \overline{\big(\nu\cdot a^{E}\nabla v\big)},\quad u,v\in H^2_{y} .
\end{equation}
Note that $\mathfrak{a}(\cdot,\cdot;n)$ is well-defined, by the trace theorem, and is sesquilinear in its variable. Importantly, it orthogonalizes the straight-interface modes in the following sense.
\begin{proposition} \label{prop_energy_flux}
Let $\lambda\in \mathcal{I}_{0}$. For $\sigma,\sigma^{\prime}\in \{+,-\}$, $\mathfrak{a}\big(u_{\sigma}^{E}(\cdot;\lambda),u_{\sigma^{\prime}}^{E}(\cdot;\lambda);n\big)$ is independent of $n$, that is,
\begin{equation} \label{eq_energy_flux_1}
\mathfrak{a}\big(u_{\sigma}^{E}(\cdot;\lambda),u_{\sigma^{\prime}}^{E}(\cdot;\lambda);n\big)
\equiv \mathfrak{a}\big(u_{\sigma}^{E}(\cdot;\lambda),u_{\sigma^{\prime}}^{E}(\cdot;\lambda);0\big)
=:\mathfrak{a}\big(u_{\sigma}^{E}(\cdot;\lambda),u_{\sigma^{\prime}}^{E}(\cdot;\lambda)\big) .
\end{equation}
When $\sigma=\sigma^{\prime}$, $\mathfrak{a}\big(u_{\sigma}^{E}(\cdot;\lambda),u_{\sigma^{\prime}}^{E}(\cdot;\lambda)\big)$ is proportional to the slope of the interface eigenvalue $\lambda^{E}(\kappa)$, measured at $\kappa=\kappa^{\sigma}(\lambda)$, i.e.,
\begin{equation} \label{eq_energy_flux_2}
\mathfrak{a}\big(u_{\sigma}^{E}(\cdot;\lambda),u_{\sigma^{\prime}}^{E}(\cdot;\lambda)\big)=i\partial_{\kappa}\lambda^{E}\big(\kappa^{\sigma}(\lambda)\big) .
\end{equation}
When $\sigma\neq \sigma^{\prime}$, it holds that
\begin{equation} \label{eq_energy_flux_3}
\mathfrak{a}\big(u_{\sigma}^{E}(\cdot;\lambda),u_{\sigma^{\prime}}^{E}(\cdot;\lambda)\big)=0.
\end{equation}
\end{proposition}
The proof follows the same lines as \cite[Theorem 3]{joly2016solutions}. We illustrate the results in Proposition \ref{prop_energy_flux} as follows. Physically, \eqref{eq_energy_flux_2} justifies the name of the sesquilinear form $\mathfrak{a}(\cdot,\cdot)$: when applying it to straight-interface modes, its value is proportional to the energy flux carried by these modes, since the latter quantity equals the group velocity as we have dropped all units. Moreover, \eqref{eq_energy_flux_1} states that the energy flux has the same value when evaluated on different cross sections of the interface $E$; this follows simply from the fact that $u_{\pm}^{E}$ are (generalized) eigenmodes of $\mathcal{L}^{E}$, i.e., 
$$
(\mathcal{L}^{E}-\lambda)u_{\pm}^{E}=0
$$ 
and no sources appear. Finally, by \eqref{eq_energy_flux_3}, we see that different eigenmodes are decoupled under the action of $\mathfrak{a}(\cdot,\cdot)$.

The following symmetry relations will also be frequently applied.
\begin{proposition} \label{prop_trace_symmetry}
For any $\lambda\in \mathcal{I}_{0}$, there exists $\beta=\beta(\lambda)\in\mathbb{C}$ with $|\beta|=1$ such that
\begin{equation} \label{eq_trace_symmetry_1}
(\mathcal{F}u_{-}^{E})(\bm{x})=\beta u_{+}^{E}(\bm{x}),
\end{equation}
for all $\bm{x}\in\mathbb{R}^2$. Moreover, on the auxiliary interface $\Gamma$, it holds that
\begin{equation} \label{eq_trace_symmetry_2}
\begin{aligned}
&\gamma (\mathcal{R}u_{-}^{E})=\beta \gamma u_{+}^{E} ,\quad \partial_{\nu,a^{E}}(\mathcal{R}u_{-}^{E})=-\beta \partial_{\nu,a^{E}}u_{+}^{E} , \\
&\gamma (\mathcal{R}u_{+}^{E})=\beta^{-1}\gamma u_{-}^{E},\quad \partial_{\nu,a^{E}}(\mathcal{R}u_{+}^{E})=-\beta^{-1} \partial_{\nu,a^{E}}u_{-}^{E} .
\end{aligned}
\end{equation}
\end{proposition}
\begin{proof}
The identity \eqref{eq_trace_symmetry_1} follows directly from the reflection symmetry of the straight-interface operator, i.e., $[\mathcal{L}^{E},\mathcal{F}]=0$. With \eqref{eq_trace_symmetry_1}, the first identity in \eqref{eq_trace_symmetry_2} is checked as
\begin{equation*}
(\mathcal{R}u_{-}^{E})(\bm{x})=u_{-}^{E}(R\bm{x})\overset{(i)}{=}u_{-}^{E}(F\bm{x})=\beta u_{+}^{E}(\bm{x}) \quad (\forall \bm{x}\in \Gamma),
\end{equation*}
where (i) is derived by noticing that $R\bm{x}=F\bm{x}$ when $\bm{x}\in \Gamma$, i.e., the rotation of $\Gamma$ coincides with its reflection image about the $y$-axis; see Figure \ref{fig_bended_interface_structure}. Similarly, the conormal derivative is computed as follows
\begin{equation*}
\begin{aligned}
&\partial_{\nu,a^{E}}(\mathcal{R}u_{-}^{E})(\bm{x})=\nu\cdot a^{E}(\bm{x})\nabla \big(u_{-}^{E}(R\bm{x}) \big)
\overset{(i)}{=}R\nu\cdot a^{E}(\bm{x})(\nabla u_{-}^{E})(R\bm{x}) \\
&=R\nu\cdot a^{E}(\bm{x})(\nabla u_{-}^{E})(F\bm{x})
\overset{(ii)}{=} F^{\top}R\nu\cdot a^{E}(\bm{x})\nabla \big(u_{-}^{E}(F\bm{x}) \big) \\
&\overset{(iii)}{=}-\beta \nu \cdot a^{E}(\bm{x}) \nabla \big(u_{+}^{E}(\bm{x}) \big) = -\beta \partial_{\nu,a^{E}}u_{+}^{E} (\bm{x}) .
\end{aligned}
\end{equation*}
Here, the equalities (i)-(ii) follow from the formula $\nabla(f(M \bm{x}))=M^{T}(\nabla f)(M\bm{x})$ for any orthogonal matrix $M$, and (iii) is derived by noticing the symmetry relation $F_{\Gamma}=F^{\top}R=FR$ (recall that $F_{\Gamma}$ is the reflection about $\Gamma$), and applying \eqref{eq_trace_symmetry_1}. The other identities in \eqref{eq_trace_symmetry_2} are checked similarly.
\end{proof}

\section{Limiting Absorption Principle and Out-going Green Operator/Function}
\label{sec_LA_out_green}

In this section, we introduce the out-going Green operator/function associated with the straight-interface operator $\mathcal{L}^{E}$ following the lines of \cite{joly2016solutions}. These notions, as shown in Section \ref{sec_bend_immunity}, serve as the fundamental building block of our layer-potential framework.

\subsection{Definition, Far-field Asymptotics and Symmetry}
\label{sec_LA_out_green_def}

For the elliptic operator $\mathcal{L}^{E}$, its Green function is typically defined as the integral kernel associated with the resolvent $(\mathcal{L}^{E}-\lambda)^{-1}$. However, it is apparently ill-defined within the range of frequencies corresponding to the interface eigenvalue, i.e., $\lambda\in\mathcal{I}_{0}$. Nonetheless, thanks to our assumption that the interface eigenvalue has a non-vanishing derivative for $\lambda\in\mathcal{I}_{0}$, it is possible to define a right inverse $\mathcal{G}^{E,out}$ of $\mathcal{L}^{E}-\lambda$ by choosing an appropriate radiation condition at infinity. Specifically, the operator $\mathcal{G}^{E,out}$ is defined via the limiting-absorption principle. To outline its construction, we first note that, by the Floquet transform, the resolvent $\big(\mathcal{L}^{E}-(\lambda+i\epsilon)\big)^{-1}$ ($\epsilon>0$) can be written as
\begin{equation} \label{eq_out_going_green_1}
\begin{aligned}
\big[\big(\mathcal{L}^{E}-(\lambda+i\epsilon)\big)^{-1}u\big](\bm{x})
&=\frac{1}{2\pi}\int_{-\pi}^{\pi}d\kappa \big[\big(\mathcal{L}^{E}_{\kappa,\bm{e}_1}\mathcal{Q}_{\kappa}-(\lambda+i\epsilon)\big)^{-1}\mathcal{Q}_{\kappa}\hat{u}_{\kappa}\big](\tilde{\bm{x}})e^{i\kappa (\bm{x}-\tilde{\bm{x}})\cdot\bm{e}_1} \\
&\quad + \frac{1}{2\pi}\int_{-\pi}^{\pi}d\kappa \frac{\big(\mathcal{P}_{\kappa}\hat{u}_{\kappa}\big)(\tilde{\bm{x}})}{\lambda^{E}(\kappa)-(\lambda+i\epsilon)}e^{i\kappa (\bm{x}-\tilde{\bm{x}})\cdot\bm{e}_1} ,
\end{aligned}
\end{equation}
for $\bm{x}\in\mathbb{R}^2$ and $u\in H^{-1}(\mathbb{R}^2)$ with compact support along $E$ (specified later). Here, $\mathcal{L}^{E}_{\kappa,\bm{e}_1}=\mathcal{L}^{E}\big|_{L^2_{\kappa,\bm{e}_1}}$ (see \eqref{eq_quasi_periodic_strip_space}), $\mathcal{P}_{\kappa}$ is the spectral projection to the straight-interface mode, 
\begin{equation*}
\mathcal{P}_{\kappa}:H^{-1}(S)\to H^{1}(S),\quad \mathcal{P}_{\kappa}\hat{u}:=\big\langle\hat{u},u_{\kappa}^{E} \big\rangle_{S}u_{\kappa}^{E},
\end{equation*}
($\langle \cdot,\cdot \rangle_{S}$ denotes the dual product induced by $L^2$-product; see Section \ref{sec_notation}) $\mathcal{Q}_{\kappa}:=\mathbbm{1}-\mathcal{P}_{\kappa}$, $\hat{u}_{\kappa}$ is the Floquet transform of $u$ along $\bm{e}_1$
\begin{equation*}
\hat{u}_{\kappa}(\bm{x}):=\sum_{n\in\mathbb{Z}}u(\bm{x}+n\bm{e}_1)e^{-i\kappa n} ,
\end{equation*}
and $\tilde{\bm{x}}:=\bm{x}\text{ mod }\bm{e}_1$ is the translation image of $\bm{x}$ in the unit strip $S$. Since the first integral in \eqref{eq_out_going_green_1} is regular, as $\mathcal{I}_{0}\cap \text{Spec}(\mathcal{L}^{E}_{\kappa,\bm{e}_1}\mathcal{Q}_{\kappa})=\emptyset$, the ill-definedness of $(\mathcal{L}^{E}-\lambda)^{-1}$ is manifested by the singularity of the integrand $1/(\lambda^{E}-(\lambda+i\epsilon))$ when $\epsilon=0$. However, thanks to the assumption $\partial_{\kappa}\lambda^{E}\neq 0$, the limit of \eqref{eq_out_going_green_1} as $\epsilon\to 0^+$ exists, which contains a principal-value integral and a projection term
\begin{equation} \label{eq_out_going_green_2}
\begin{aligned}
\lim_{\epsilon\to 0^+}\big[\big(\mathcal{L}^{E}-(\lambda+i\epsilon)\big)^{-1}u\big](\bm{x})
&=\frac{1}{2\pi}\int_{-\pi}^{\pi}d\kappa \big[\big(\mathcal{L}^{E}_{\kappa,\bm{e}_1}\mathcal{Q}_{\kappa}-\lambda\big)^{-1}\mathcal{Q}_{\kappa}\hat{u}_{\kappa}\big](\tilde{\bm{x}})e^{i\kappa (\bm{x}-\tilde{\bm{x}})\cdot\bm{e}_1} \\
&\quad + \frac{1}{2\pi}\text{p.v.}\int_{-\pi}^{\pi}d\kappa \frac{\big(\mathcal{P}_{\kappa}\hat{u}_{\kappa}\big)(\tilde{\bm{x}})}{\lambda^{E}(\kappa)-\lambda}e^{i\kappa (\bm{x}-\tilde{\bm{x}})\cdot\bm{e}_1} \\
&\quad + \sum_{\sigma\in\{\pm\}}\frac{i}{2|\partial_{\kappa}\lambda^{E}\big(\kappa^{\sigma}(\lambda)\big)|}\big(\mathcal{P}_{\kappa^{\sigma}(\lambda)}\hat{u}_{\kappa^{\sigma}(\lambda)}\big)(\tilde{\bm{x}})e^{i\kappa^{\sigma}(\lambda) (\bm{x}-\tilde{\bm{x}})\cdot\bm{e}_1} \\
&=: \big[\mathcal{G}^{E,out}(\lambda) u\big](\bm{x}).
\end{aligned}
\end{equation}
Here, p.v. indicates that the $\kappa$-integral is understood in the sense of the Cauchy principal-value integral near $\kappa=\kappa^{\pm}(\lambda)$. The operator $\mathcal{G}^{E,out}(\lambda)$ defined in \eqref{eq_out_going_green_2} is referred to as \textbf{the out-going Green operator}. As manifested by the name, it satisfies the out-going radiation condition in the sense that $\big[\mathcal{G}^{E,out}(\lambda) u\big](\bm{x})$ is proportional to the right- (left-) going straight-interface mode $u_{+}^{E}$ ($u_{-}^{E}$) as $x_1\to\infty$ ($x_1\to -\infty$). These properties are proved following the lines of \cite[Theorem 6 and 7]{joly2016solutions}:
\begin{proposition}[Out-going Green's Operator] \label{prop_out_going_Green_basic}
Let $\lambda\in \mathcal{I}_{0}$. Define the space
\begin{equation*}
    H^{-1}_{x-comp}:=\big\{u\in H^{-1}(\mathbb{R}^2):\, \text{supp }(u)\subset \overline{\cup_{n\leq j\leq m}S_{j}} \quad \text{for some $n,m\in\mathbb{Z}$} \big\}.
\end{equation*}
Then the operator $\mathcal{G}^{E,out}(\lambda)$ in \eqref{eq_out_going_green_2} is well-defined, mapping $H^{-1}_{x-comp}$ to $H^1_y$, and is bounded:
\begin{equation} \label{eq_out_going_Green_operator_bound}
\big\|\mathcal{G}^{E,out}(\lambda)u\big\|_{H^1(S_n)}\leq C\|u\|_{H^{-1}(\mathbb{R}^2)},\quad \forall n\in\mathbb{Z},
\end{equation}
with $C=C(\lambda)>0$. Moreover, for any $u\in H^{-1}_{x-comp}$, it holds that
\begin{equation*}
(\mathcal{L}^{E}-\lambda)\mathcal{G}^{E,out}(\lambda)u=u.
\end{equation*}
The convergence \eqref{eq_out_going_green_2} is uniform in norm,
\begin{equation*}
\lim_{\epsilon\to 0^{+}}\sup_{\substack{u\in H^{-1}_{x-comp} \\ \|u\|_{H^{-1}(\mathbb{R}^2)}=1}} \Big\|\big(\mathcal{L}^{E}-(\lambda+i\epsilon)\big)^{-1}u- \mathcal{G}^{E,out}(\lambda)u \Big\|_{H^1(S_n)} =0.
\end{equation*}
for any $n\in\mathbb{Z}$. Moreover, supposing that $\text{supp }(u)\subset S_m$ ($m\in\mathbb{Z}$), it holds that
\begin{equation} \label{eq_out_going_Green_asymptotics_positive}
\big[\mathcal{G}^{E,out}(\lambda)u\big](\bm{x})=w_{u}^{+}(\bm{x}) + \frac{i}{|\partial_{\kappa}\lambda^{E}\big(\kappa^{+}(\lambda)\big)|}\big(\mathcal{P}_{\kappa^{+}(\lambda)}\hat{u}_{\kappa^{+}(\lambda)}\big)(\tilde{\bm{x}})e^{i\kappa^{+}(\lambda) (\bm{x}-\tilde{\bm{x}})\cdot\bm{e}_1} ,
\end{equation}
and 
\begin{equation} \label{eq_out_going_Green_asymptotics_negative}
\big[\mathcal{G}^{E,out}(\lambda)u\big](\bm{x})=w_{u}^{-}(\bm{x}) + \frac{i}{|\partial_{\kappa}\lambda^{E}\big(\kappa^{-}(\lambda)\big)|}\big(\mathcal{P}_{\kappa^{-}(\lambda)}\hat{u}_{\kappa^{-}(\lambda)}\big)(\tilde{\bm{x}})e^{i\kappa^{-}(\lambda) (\bm{x}-\tilde{\bm{x}})\cdot\bm{e}_1} ,
\end{equation}
where the functions $w_{u}^{\pm}$ satisfy
\begin{equation} \label{eq_out_going_Green_asymptotics_evan_part}
\|w_{u}^{\pm} \|_{H^1(S_n)}\leq c_1e^{- c_2 |n-m|}\|u\|_{H^{-1}(\mathbb{R}^2)} \quad \text{for $\pm (n-m)\geq 0$,}
\end{equation}
with $c_i=c_i(\lambda)>0$.
\end{proposition}
We point out that, compared to the setup of \cite{joly2016solutions}, where the periodic elliptic operator $\mathcal{A}$ is defined on a waveguide $\mathbb{R}\times [0,1]$ (with Neumann boundary conditions on the boundary), our operator $\mathcal{L}^{E}$ is defined on the whole plane $\mathbb{R}^2$. This difference does not cause any difficulty in the boundedness of $\mathcal{G}^{E,out}(\lambda)$ because 
\begin{itemize}
    \item[(i)]  we have posed a transverse boundedness for functions in $\text{dom}(\mathcal{G}^{E,out}(\lambda))$ (achieved by using the global $H^{-1}$ functions);
    \item[(ii)] the straight-interface mode $u_{\kappa}^{E}$, as well as the reduced resolvent, are bounded transversely in the sense that
    \begin{equation*}
    \|u_{\kappa}^{E}\|_{H^1(S)}<\infty,\quad \|\big(\mathcal{L}^{E}_{\kappa,\bm{e}_1}\mathcal{Q}_{\kappa}-\lambda\big)^{-1}\mathcal{Q}_{\kappa}\|_{\mathcal{B}(H^{-1}(S),H^1(S))}<\infty.
    \end{equation*}
\end{itemize}
Both estimates follow directly from the fact that the (quasi-periodic) interface operator has only a discrete spectrum in the bulk spectral gap $\mathcal{I}$.

However, we should be careful when establishing the far-field asymptotics \eqref{eq_out_going_Green_asymptotics_positive}-\eqref{eq_out_going_Green_asymptotics_negative} following the lines of \cite[Theorem 7]{joly2016solutions}, which is based on the analytic extension of Bloch eigenvalues. The latter point is guaranteed automatically in the setup of \cite{joly2016solutions}: since the operator $\mathcal{A}_{\kappa,\bm{e}_1}$ (i.e., restriction of $\mathcal{A}$ to functions with quasi-periodic boundary condition along $\bm{e}_1$) admits a compact resolvent as the unit cell $[0,1]\times [0,1]$ is compact, $\mathcal{A}_{\kappa,\bm{e}_1}$ has a purely discrete spectrum. Then, since $\mathcal{A}_{\kappa,\bm{e}_1}$ depends analytically on $\kappa$, all its eigenvalues extend along the whole period $[-\pi,\pi]$ (cf. \cite[Theorem 3.9, Chapter 7]{kato2013perturbation}). Nevertheless, in our case, the operator $\mathcal{L}^{E}_{\kappa,\bm{e}_1}$ does not have a compact resolvent. As a consequence, the interface eigenvalue $\lambda^{E}(\kappa)$, which belongs to the discrete spectrum of $\mathcal{L}^{E}_{\kappa,\bm{e}_1}$, extends analytically \textbf{until it enters the essential spectrum $\text{Spec}_{ess}\big(\mathcal{L}^{E}_{\kappa,\bm{e}_1}\big)$}. This problem is resolved by Assumption \ref{asmp_detailed_edge_band}(ii), which guarantees that $\lambda^{E}(\kappa)$, as well as the projections $\mathcal{P}_{\kappa}$ and $\mathcal{Q}_{\kappa}$, can be analytically extended to a complex neighborhood $\mathcal{U}\supset [-\pi,\pi]$; see \cite[Chapter 7, Section 3.2]{kato2013perturbation}. Consequently, the proof of \cite[Theorem 7]{joly2016solutions} carries directly as follows. First, using the residue formula, one can show that\footnote{See the second equality in \eqref{eq_analytic_extension_proof_3} for an illustration.} 
\begin{equation} \label{eq_out_going_green_3}
\begin{aligned}
&\frac{1}{2\pi}\text{p.v.}\int_{-\pi}^{\pi}d\kappa \frac{\big(\mathcal{P}_{\kappa}\hat{u}_{\kappa}\big)(\tilde{\bm{x}})}{\lambda^{E}(\kappa)-\lambda}e^{i\kappa (\bm{x}-\tilde{\bm{x}})\cdot\bm{e}_1}  + \sum_{\sigma\in\{\pm\}}\frac{i}{2|\partial_{\kappa}\lambda^{E}\big(\kappa^{\sigma}(\lambda)\big)|}\big(\mathcal{P}_{\kappa^{\sigma}(\lambda)}\hat{u}_{\kappa^{\sigma}(\lambda)}\big)(\tilde{\bm{x}})e^{i\kappa^{\sigma}(\lambda) (\bm{x}-\tilde{\bm{x}})\cdot\bm{e}_1} \\
&=\frac{1}{2\pi}\int_{C_1}d\kappa \frac{\big(\mathcal{P}_{\kappa}\hat{u}_{\kappa}\big)(\tilde{\bm{x}})}{\lambda^{E}(\kappa)-\lambda}e^{i\kappa (\bm{x}-\tilde{\bm{x}})\cdot\bm{e}_1},
\end{aligned}
\end{equation}
where the integral contour $C_{1}=C_1(\eta)\subset \mathcal{U}$ is defined as
\begin{equation} \label{eq_C1_contour}
\begin{aligned}
C_1:=&\big[-\pi,\kappa^{-}(\lambda)-\eta\big]\cup \big[\kappa^{-}(\lambda)+\eta,\kappa^{+}(\lambda)-\eta\big]\cup \big[\kappa^{+}(\lambda)+\eta,\pi\big]\cup C^{+}_{\lambda,\eta} \cup C^{-}_{\lambda,\eta} \\
&\text{with}\quad C^{+}_{\lambda,\eta}:=\big\{\kappa^{-}(\lambda)+\eta e^{i\theta}:\pi\geq \theta\geq 0\big\},\quad C^{-}_{\lambda,\eta}:=\big\{\kappa^{+}(\lambda)+\eta e^{i\theta}:\pi\leq \theta\leq 2\pi\big\} ,
\end{aligned}
\end{equation}
which is shown in Figure \ref{fig_integral_contours}. Similarly,
\begin{equation} \label{eq_out_going_green_4}
\begin{aligned}
&\frac{1}{2\pi}\int_{-\pi}^{\pi}d\kappa \big[\big(\mathcal{L}^{E}_{\kappa,\bm{e}_1}\mathcal{Q}_{\kappa}-\lambda\big)^{-1}\mathcal{Q}_{\kappa}\hat{u}_{\kappa}\big](\tilde{\bm{x}})e^{i\kappa (\bm{x}-\tilde{\bm{x}})\cdot\bm{e}_1} 
=\frac{1}{2\pi}\int_{C_1}d\kappa \big[\big(\mathcal{L}^{E}_{\kappa,\bm{e}_1}\mathcal{Q}_{\kappa}-\lambda\big)^{-1}\mathcal{Q}_{\kappa}\hat{u}_{\kappa}\big](\tilde{\bm{x}})e^{i\kappa (\bm{x}-\tilde{\bm{x}})\cdot\bm{e}_1},
\end{aligned}
\end{equation}
which holds because the integrand is analytic in $\mathcal{U}$. Second, we deform the contour $C_{1}$ to $C_{2}$, as sketched in Figure \ref{fig_integral_contours}. As $C_1\cup C_2$ encloses a simple pole of the integrand $1/(\lambda^E(\kappa)-\lambda)$ located at $\kappa=\kappa^{+}(\lambda)$ (the black cross in Figure \ref{fig_integral_contours}), applying the residue formula again leads to
\begin{equation} \label{eq_out_going_green_5}
\begin{aligned}
\frac{1}{2\pi}\int_{C_1}d\kappa \frac{\big(\mathcal{P}_{\kappa}\hat{u}_{\kappa}\big)(\tilde{\bm{x}})}{\lambda^{E}(\kappa)-\lambda}e^{i\kappa (\bm{x}-\tilde{\bm{x}})\cdot\bm{e}_1}
&=\frac{1}{2\pi}\int_{C_2}d\kappa \frac{\big(\mathcal{P}_{\kappa}\hat{u}_{\kappa}\big)(\tilde{\bm{x}})}{\lambda^{E}(\kappa)-\lambda}e^{i\kappa (\bm{x}-\tilde{\bm{x}})\cdot\bm{e}_1} \\
&\quad - \frac{i}{|\partial_{\kappa}\lambda^{E}\big(\kappa^{+}(\lambda)\big)|}\big(\mathcal{P}_{\kappa^{+}(\lambda)}\hat{u}_{\kappa^{+}(\lambda)}\big)(\tilde{\bm{x}})e^{i\kappa^{+}(\lambda) (\bm{x}-\tilde{\bm{x}})\cdot\bm{e}_1} ,\\
\frac{1}{2\pi}\int_{C_1}d\kappa \big[\big(\mathcal{L}^{E}_{\kappa,\bm{e}_1}\mathcal{Q}_{\kappa}-\lambda\big)^{-1}\mathcal{Q}_{\kappa}\hat{u}_{\kappa}\big](\tilde{\bm{x}})e^{i\kappa (\bm{x}-\tilde{\bm{x}})\cdot\bm{e}_1}
&=\frac{1}{2\pi}\int_{C_2}d\kappa \big[\big(\mathcal{L}^{E}_{\kappa,\bm{e}_1}\mathcal{Q}_{\kappa}-\lambda\big)^{-1}\mathcal{Q}_{\kappa}\hat{u}_{\kappa}\big](\tilde{\bm{x}})e^{i\kappa (\bm{x}-\tilde{\bm{x}})\cdot\bm{e}_1} .
\end{aligned}
\end{equation}
Finally, noting that the integrals over the vertical segments of $C_{2}$ cancel by the periodicity of the integrand, and the integral over the horizontal segment decays exponentially (since $\Im \kappa >0$), we conclude the proof of \eqref{eq_out_going_Green_asymptotics_positive} by \eqref{eq_out_going_green_2} and \eqref{eq_out_going_green_3}-\eqref{eq_out_going_green_5}. The proof of \eqref{eq_out_going_Green_asymptotics_negative} is similar. We refer the reader to \cite{joly2016solutions} for more details.

\begin{figure}
\centering
\begin{tikzpicture}[scale=0.6]
%C_1 curve
\tikzset{->-/.style={decoration={markings,mark=at position #1 with 
{\arrow{latex}}},postaction={decorate}}};

\draw[->] (-11,0)--(11,0);
\draw[->] (0,-3.5)--(0,3.5);
\node[right] at (11.2,0) {$\Re \kappa$};
\node[above] at (0,3.7) {$\Im \kappa$};
\draw[thick] (-10,-0.1)--(-10,0.1);
\node[below] at (-10,-0.29) {$-\pi$};
\draw[thick] (10,-0.1)--(10,0.1);
\node[below] at (10,-0.29) {$\pi$};
\draw[thick] (-5.25,-0.1)--(-5.25,0.1);
\node[below,scale=0.8] at (-5.25,-0.2) {$\kappa^{-}(\lambda)$};
\draw[thick] (5.25,-0.1)--(5.25,0.1);
\node[above,scale=0.8] at (5.25,0.2) {$\kappa^{+}(\lambda)$};
\draw[very thick] (5.1,0.15)--(5.4,-0.15);
\draw[very thick] (5.1,-0.15)--(5.4,0.15);
\draw[very thick] (-5.4,0.15)--(-5.1,-0.15);
\draw[very thick] (-5.4,-0.15)--(-5.1,0.15);

\draw[->-=0.5,very thick,red,domain=180:0] plot ({-5.25+1.5*cos(\x)},{1.5*sin(\x)});
\draw[->-=0.5,very thick,red,domain=180:360] plot ({5.25+1.5*cos(\x)},{1.5*sin(\x)});
\draw[->,dashed] ({-5.25},{0})--({-5.25+1.5*cos(135)},{1.5*sin(135)});
\draw[->,dashed] ({5.25},{0})--({5.25+1.5*cos(315)},{1.5*sin(315)});
\node[right,scale=0.8] at ({-5.25+1*cos(135)},{1*sin(135)}) {$\eta$};
\node[left,scale=0.8] at ({5+1*cos(315)},{1*sin(315)}) {$\eta$};
\draw[->-=0.5,very thick,red] (-10,0)--(-6.75,0);
\draw[->-=0.5,very thick,red] (-3.75,0)--(3.75,0);
\draw[->-=0.5,very thick,red] (6.75,0)--(10,0);
\node[above,scale=1,red] at (-3.2,0.5) {$C_{1}$};

%analytic domain
\path [fill=blue,opacity=0.1] ({-11},{-3})--({11},{-3})--({11},{3})--({-11},{3})--({-11},{-3});
%C2 curve
\draw[->-=0.5,very thick,blue] (-10,0)--(-10,2.5)--(10,2.5)--(10,0);
\node[right,scale=1,blue] at (-10,1.5) {$C_{2}$};
\end{tikzpicture}
\caption{Integral contours used in the proof. The shadowed area is the domain $\mathcal{U}$ in which $\lambda^E(\kappa),\mathcal{P}_{\kappa},\mathcal{Q}_{\kappa}$ are analytic. The red/blue curve refers to the integral contour $C_1/C_2$, respectively.}
\label{fig_integral_contours}
\end{figure}
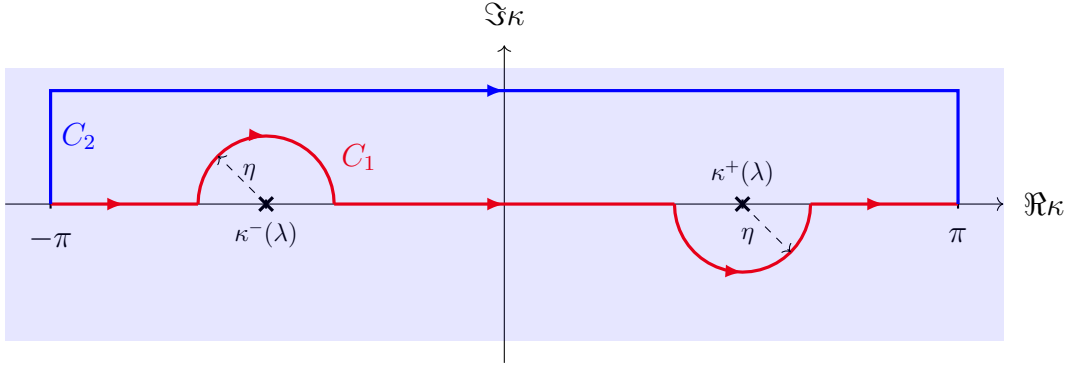

We note that the asymptotics \eqref{eq_out_going_Green_asymptotics_positive}-\eqref{eq_out_going_Green_asymptotics_negative} imply similar properties of the out-going Green function $G^{E}(\bm{x},\bm{y};\lambda)$, i.e., the integral kernel of $\mathcal{G}^{E,out}(\lambda)$, by the De-Giorgi estimate. For example, let us fix $\bm{x},\bm{y}\in \mathbb{R}^2$ so that $x_1-y_1>1$. Without loss of generality, we assume $\bm{y}\in S_{0}$, $\bm{x}\in S_{n}$, and $B(\bm{y},1/3)\subset S_{0}=S$, $B(\bm{x},1/3)\subset S_{n}$, where $B(\bm{x},r)$ is the open ball centered at $\bm{x}$ with radius $r>0$. By \eqref{eq_out_going_Green_asymptotics_positive}, it holds that
\begin{equation} \label{eq_out_going_green_6}
\begin{aligned}
&\Big\|\mathcal{G}^{E,out}(\lambda)(g\mathbbm{1}_{B(\bm{y},1/3)}) - \frac{i}{|\partial_{\kappa}\lambda^{E}\big(\kappa^{+}(\lambda)\big)|}u^{E}_{\kappa^{+}(\lambda)}\big(g\mathbbm{1}_{B(\bm{y},1/6)},u^{E}_{\kappa^{+}(\lambda)} \big)_{L^2(\mathbb{R}^2)} \Big\|_{L^2(B(\bm{x},1/3))} \\
&\leq c_1 e^{-c_2 n}\|g\|_{L^{2}(\mathbb{R}^2)}
\end{aligned}
\end{equation}
for any $g\in L^{\infty}(\mathbb{R}^2)$. This means that if we define the function
\begin{equation*}
\begin{aligned}
u_{\bm{y},g}(\cdot)
&:=\big[\mathcal{G}^{E,out}(\lambda)(g\mathbbm{1}_{B(\bm{y},1/3)})\big](\cdot) - \frac{i}{|\partial_{\kappa}\lambda^{E}\big(\kappa^{+}(\lambda)\big)|}u^{E}_{\kappa^{+}(\lambda)}(\cdot)\big(g\mathbbm{1}_{B(\bm{y},1/6)},u^{E}_{\kappa^{+}(\lambda)} \big)_{L^2(\mathbb{R}^2)} \\
&=\int_{B(\bm{y},1/3)}d\bm{z}\Big[G^{E,out}(\cdot,\bm{z};\lambda)-\frac{i}{|\partial_{\kappa}\lambda^{E}\big(\kappa^{+}(\lambda)\big)|}u^{E}_{\kappa^{+}(\lambda)}(\cdot)\overline{u^{E}_{\kappa^{+}(\lambda)}(\bm{z})} \Big]g(\bm{z}),
\end{aligned}
\end{equation*}
then $u_{\bm{y},g}$ satisfies $\mathcal{L}^{E}u_{\bm{y},g}=\lambda u_{\bm{y},g}$ in the ball $B(\bm{x},1/3)$, with $\|u_{\bm{y},g}\|_{L^2(B(\bm{x},1/3))}$ satisfying the estimate \eqref{eq_out_going_green_6}. Hence, using the De-Giorgi estimate for the elliptic equations, the $L^2$ estimate is lifted to $L^{\infty}$, which implies that
\begin{equation} \label{eq_out_going_green_7}
\begin{aligned}
\Big|\int_{B(\bm{y},1/3)}d\bm{z}\Big[G^{E,out}(\bm{x},\bm{z};\lambda)-\frac{i}{|\partial_{\kappa}\lambda^{E}\big(\kappa^{+}(\lambda)\big)|}u^{E}_{\kappa^{+}(\lambda)}(\bm{x}) \overline{u^{E}_{\kappa^{+}(\lambda)}(\bm{z})} \Big]g(\bm{z})\Big|
&\leq \|u_{\bm{y},g}\|_{L^{\infty}(B(\bm{x},1/6))} \\
&\leq C e^{-c_2 n}\|g\|_{L^{2}(\mathbb{R}^2)} .
\end{aligned}
\end{equation}
Repeating the above argument for the function 
$$
u_{\bm{x}}(\cdot):=G^{E,out}(\bm{x},\cdot;\lambda)-\frac{i}{|\partial_{\kappa}\lambda^{E}(\kappa^{+}(\lambda))|}u^{E}_{\kappa^{+}(\lambda)}(\bm{x})\overline{u^{E}_{\kappa^{+}(\lambda)}(\cdot)},
$$
whose local $L^2$ estimate is obtained by taking $g=\overline{u_{\bm{x}}}$ in \eqref{eq_out_going_green_7}, we obtain the far-field asymptotics of the out-going Green function
\begin{equation*}
\Big| G^{E,out}(\bm{x},\bm{y};\lambda)-\frac{i}{|\partial_{\kappa}\lambda^{E}\big(\kappa^{+}(\lambda)\big)|}u^{E}_{\kappa^+(\lambda)}(\bm{x}) \overline{u^{E}_{\kappa^{+}(\lambda)}(\bm{y})} \Big|\leq C_1e^{-C_2 n} .
\end{equation*}
We refer the reader to \cite[Proposition 3.1]{qiu2025bec_finite} for a detailed implementation of this argument. The decay estimates of derivatives $\nabla_{\bm{x}}\nabla_{\bm{y}} G^{E,out}$ are also obtained following the same lines. In conclusion, the following result holds. 
\begin{proposition}[Far-field asymptotics of $G^{E,out}$ along the interface $E$]
\label{prop_parallel_decay_Green_function}
Let $\lambda\in\mathcal{I}_{0}$. There exists $C_i=C_i(\lambda)>0$ such that, when $\pm(x_1-y_1)>1$, we have
\begin{equation*}
\begin{aligned}
&\Big| \partial_{\bm{x}}^{\alpha}\partial_{\bm{y}}^{\beta}\Big[G^{E,out}(\bm{x},\bm{y};\lambda)-\frac{i}{|\partial_{\kappa}\lambda^{E}\big(\kappa^{\pm}(\lambda)\big)|}u^{E}_{\pm}(\bm{x})\overline{u^{E}_{\pm}(\bm{y})} \Big]\Big| \leq C_1e^{-C_2 |x_1-y_1|}
\end{aligned}
\end{equation*}
for any multi-indices $\alpha,\beta\in \mathbb{N}^2$ with $0\leq |\alpha|,|\beta|\leq 1$.
\end{proposition}

We will also need a decay estimate of $G^{E,out}$ along the transverse direction.
\begin{proposition}[Far-field asymptotics of $G^{E,out}$ transverse to the interface $E$]
\label{prop_perpendicular_decay_Green_function}
Let $\lambda\in\mathcal{I}_{0}$. There exists $C_i=C_i(\lambda)>0$ such that, when $|x_2-y_2|>1$,
\begin{equation*}
\begin{aligned}
&\big| \partial_{\bm{x}}^{\alpha}\partial_{\bm{y}}^{\beta}G^{E,out}(\bm{x},\bm{y};\lambda)\big| \leq C_1(1+|x_1-y_1|)e^{-C_2 |x_2-y_2|}
\end{aligned}
\end{equation*}
for any multi-indices $\alpha,\beta\in \mathbb{N}^2$ with $0\leq |\alpha|,|\beta|\leq 1$.
\end{proposition}

\begin{proof}
We only prove the estimate for $\alpha=\beta=0$, while the derivatives of the Green function are estimated similarly, following the lines of \cite[Proposition 3.1, Step 3]{qiu2025bec_finite}. It is sufficient to establish a similar estimate as \eqref{eq_out_going_green_6}, but regarding the transverse decay, for each of the three parts in the definition \eqref{eq_out_going_green_2} of $G^{E,out}$: for fixed $\bm{x},\bm{y}\in\mathbb{R}^2$ with $|x_2-y_2|>1$, $\bm{y}\in S_0=S$ and $\bm{x}\in S_n$ for some $n\in\mathbb{Z}$,
\begin{equation}  \label{eq_out_going_green_8}
\sup_{\kappa\in [-\pi,\pi]}\big\|  \mathcal{P}_{\kappa}(g\mathbbm{1}_{B(\bm{y},1/3)})_{\kappa}^{\wedge}\big\|_{L^{2}(B(\bm{x},1/3))}
\leq C_1e^{-C_2|x_2-y_2|},
\end{equation}
\begin{equation}  \label{eq_out_going_green_9}
\sup_{\kappa\in [-\pi,\pi]}\big\|  \big(\mathcal{L}^{E}_{\kappa,\bm{e}_1}\mathcal{Q}_{\kappa}-\lambda\big)^{-1}\mathcal{Q}_{\kappa}(g\mathbbm{1}_{B(\bm{y},1/3)})_{\kappa}^{\wedge}\big\|_{L^{2}(B(\bm{x},1/3))}
\leq C_1e^{-C_2|x_2-y_2|},
\end{equation}
and
\begin{equation}  \label{eq_out_going_green_10}
\big\|  \text{p.v.}\int_{-\pi}^{\pi}\frac{d\kappa}{\lambda^{E}(\kappa)-\lambda}\mathcal{P}_{\kappa}(g\mathbbm{1}_{B(\bm{y},1/3)})_{\kappa}^{\wedge}\big\|_{L^{2}(B(\bm{x},1/3))}
\leq C_1(1+|x_1-y_1|)e^{-C_2|x_2-y_2|}
\end{equation}
At this time, the estimates cannot be derived from the bound \eqref{eq_out_going_Green_asymptotics_positive}-\eqref{eq_out_going_Green_asymptotics_negative}. We will prove them using the Combes-Thomas estimate.

{\color{blue}Step 1.} The estimate \eqref{eq_out_going_green_8} follows directly from the Combes-Thomas estimate of resolvents and Assumption \ref{asmp_detailed_edge_band}(ii). Indeed, by Assumption \ref{asmp_detailed_edge_band}(ii), the complex circle $C_{\lambda^{E}(\kappa),d_*/2}:=\partial B(\lambda^{E}(\kappa),d_*/2)\subset \mathbb{C}$, enclosing the interface eigenvalue $\lambda^{E}(\kappa)$ with radius $d_*/2$, satisfies
\begin{equation} \label{eq_out_going_green_11}
d\big(C_{\lambda^{E}(\kappa),d_*/2},\text{Spec}(\mathcal{L}^E_{\kappa,\bm{e}_1}) \big)=\frac{d_*}{2}>0 .
\end{equation}
Hence, by the Riesz theorem
\begin{equation} \label{eq_out_going_green_14}
\mathcal{P}_{\kappa}=-\frac{1}{2\pi i}\oint_{C_{\lambda^{E}(\kappa),d_*/2}}(\mathcal{L}^E_{\kappa,\bm{e}_1}-z)^{-1}dz ,
\end{equation}
the Combes-Thomas estimate (cf., e.g. \cite[Lemma 12]{figotin1996localization})
\begin{equation} \label{eq_out_going_green_16}
\big\|\mathbbm{1}_{B(\tilde{\bm{x}},1/3)}(\mathcal{L}^E_{\kappa,\bm{e}_1}-z)^{-1}(g\mathbbm{1}_{B(\bm{y},1/3)})_{\kappa}^{\wedge} \big\|_{L^2(S)}
\leq C_1e^{-C_2|x_2-y_2|} \quad (z\in C_{\lambda^{E}(\kappa),d_*/2})
\end{equation}
with uniformly bounded (from above and below) $C_i>0$ thanks to \eqref{eq_out_going_green_11}, and noting that the Floquet transform preserves the local $L^2$-norm, i.e., 
\begin{equation*}
\big\|  \mathcal{P}_{\kappa}(g\mathbbm{1}_{B(\bm{y},1/3)})_{\kappa}^{\wedge}\big\|_{L^{2}(B(\bm{x},1/3))}
=\big\|e^{i\kappa (\bm{x}-\bm{\tilde{x}})\cdot\bm{e}_1}  \mathcal{P}_{\kappa}(g\mathbbm{1}_{B(\bm{y},1/3)})_{\kappa}^{\wedge}\big\|_{L^{2}(B(\tilde{\bm{x}},1/3))}
=\big\|\mathcal{P}_{\kappa}(g\mathbbm{1}_{B(\bm{y},1/3)})_{\kappa}^{\wedge}\big\|_{L^{2}(B(\tilde{\bm{x}},1/3))},
\end{equation*}
one concludes the proof of \eqref{eq_out_going_green_8}.

{\color{blue}Step 2.} The estimate \eqref{eq_out_going_green_9} is obtained similarly as in Step 1. For $\kappa\in [-\pi,\pi]$, there are two cases:
\begin{center}
%Case 1: $\lambda^{E}(\kappa)=\lambda$, 
Case 1: $|\lambda^{E}(\kappa)-\lambda|<\frac{d_*}{4}$, and Case 2: $|\lambda^{E}(\kappa)-\lambda|\geq \frac{d_*}{4}$.
\end{center}
For Case 1, we still take the complex circle $C_{\lambda^{E}(\kappa),d_*/2}$, introduced in Step 1. Recalling that the reduced resolvent $\big(\mathcal{L}^{E}_{\kappa,\bm{e}_1}\mathcal{Q}_{\kappa}-z\big)^{-1}\mathcal{Q}_{\kappa}$ is the holomorphic part of $(\mathcal{L}^{E}_{\kappa,\bm{e}_1}-z)^{-1}$ for $z\in B(\lambda^{E}(\kappa),d_*/2)$, the residue formula gives
\begin{equation} \label{eq_out_going_green_12}
\big(\mathcal{L}^{E}_{\kappa,\bm{e}_1}\mathcal{Q}_{\kappa}-\lambda\big)^{-1}\mathcal{Q}_{\kappa}=\frac{1}{2\pi i}\oint_{C_{\lambda^{E}(\kappa),d_*/2}}\frac{(\mathcal{L}^E_{\kappa,\bm{e}_1}-z)^{-1}}{z-\lambda}dz ,
\end{equation}
which holds by noting the Laurent expansion
\begin{equation*}
\frac{(\mathcal{L}^E_{\kappa,\bm{e}_1}-z)^{-1}}{z-\lambda}
=\frac{\mathcal{P}_{\kappa}}{(\lambda^{E}(\kappa)-z)(z-\lambda)}
+\frac{\big(\mathcal{L}^{E}_{\kappa,\bm{e}_1}\mathcal{Q}_{\kappa}-z\big)^{-1}\mathcal{Q}_{\kappa}}{z-\lambda},\quad z\in \overline{B(\lambda^{E}(\kappa),d_*/2)}.
\end{equation*}
Then the estimate \eqref{eq_out_going_green_9} is proved by applying the Combes-Thomas estimate to \eqref{eq_out_going_green_12} and noting that the denominator is bounded from below, i.e., $|z-\lambda|>\frac{d_*}{4}>0$. For Case 2, we replace the integral contour in \eqref{eq_out_going_green_12} by $C_{\lambda,d_*/8}$, which is the circle centered at $\lambda$ with radius $d_*/8$. By doing this, one can check that the identity \eqref{eq_out_going_green_12} is modified as
\begin{equation*}
\big(\mathcal{L}^{E}_{\kappa,\bm{e}_1}\mathcal{Q}_{\kappa}-\lambda\big)^{-1}\mathcal{Q}_{\kappa}=-\frac{\mathcal{P}_{\kappa}}{\lambda^{E}-\lambda}+\frac{1}{2\pi i}\oint_{C_{\lambda,d_*/8}}\frac{(\mathcal{L}^E_{\kappa,\bm{e}_1}-z)^{-1}}{z-\lambda}dz .
\end{equation*}
Both terms decay exponentially as $|x_2-y_2|\to\infty$ because
\begin{itemize}
    \item[(i)] $\mathcal{P}_{\kappa}$ decays exponentially, due to \eqref{eq_out_going_green_8}, and $|\lambda^{E}(\kappa)-\lambda|\geq \frac{d_*}{4}$ by our assumption;
    \item[(ii)] For all $z\in C_{\kappa}$, $(\mathcal{L}^E_{\kappa,\bm{e}_1}-z)^{-1}$ decays exponentially with a uniform decay rate, according to the Combes-Thomas estimate and the fact that $\text{dist}(z,\text{Spec}(\lambda^{E}_{\kappa,\bm{e}_1}))\geq \frac{d_*}{8}$. On the other hand, the denominator $|z-\lambda|\geq \frac{d_*}{8}$ for $z\in C_{\lambda,d_*/8}$.
\end{itemize}
This completes the proof of \eqref{eq_out_going_green_9}.

{\color{blue}Step 3.} Finally, we prove \eqref{eq_out_going_green_10}. Using a standard estimate of the Cauchy principal-value integral, it follows that
\begin{equation} \label{eq_out_going_green_13}
\begin{aligned}
&\big\|  \text{p.v.}\int_{-\pi}^{\pi}\frac{d\kappa}{\lambda^{E}(\kappa)-\lambda}\mathcal{P}_{\kappa}(g\mathbbm{1}_{B(\bm{y},1/3)})_{\kappa}^{\wedge}\big\|_{L^{2}(B(\bm{x},1/3))} \\
&=\big\|  \text{p.v.}\int_{-\pi}^{\pi}\frac{d\kappa}{\lambda^{E}(\kappa)-\lambda}e^{i\kappa (\bm{x}-\bm{\tilde{x}})\cdot\bm{e}_1}\mathcal{P}_{\kappa}(g\mathbbm{1}_{B(\bm{y},1/3)})_{\kappa}^{\wedge}\big\|_{L^{2}(B(\tilde{\bm{x}},1/3))} \\
&\leq C\Big[\sup_{\kappa\in[-\pi,\pi]}\big\| e^{i\kappa (\bm{x}-\bm{\tilde{x}})\cdot\bm{e}_1} \mathcal{P}_{\kappa}(g\mathbbm{1}_{B(\bm{y},1/3)})_{\kappa}^{\wedge}\big\|_{L^{2}(B(\tilde{\bm{x}},1/3))} \\
&\quad\quad + \sup_{\kappa\in[-\pi,\pi]}\big\| \partial_{\kappa}\big(e^{i\kappa (\bm{x}-\bm{\tilde{x}})\cdot\bm{e}_1} \mathcal{P}_{\kappa}(g\mathbbm{1}_{B(\bm{y},1/3)})_{\kappa}^{\wedge}\big)\big\|_{L^{2}(B(\tilde{\bm{x}},1/3))}  \Big] \\
&=C\Big[\sup_{\kappa\in[-\pi,\pi]}\big\| e^{i\kappa (\bm{x}-\bm{\tilde{x}})\cdot\bm{e}_1} \mathcal{P}_{\kappa}(g\mathbbm{1}_{B(\bm{y},1/3)})_{\kappa}^{\wedge}\big\|_{L^{2}(B(\tilde{\bm{x}},1/3))} \\
&\quad\quad + \sup_{\kappa\in[-\pi,\pi]}\big\| e^{i\kappa (\bm{x}-\bm{\tilde{x}})\cdot\bm{e}_1}\mathcal{P}_{\kappa}\partial_{\kappa}(g\mathbbm{1}_{B(\bm{y},1/3)})_{\kappa}^{\wedge}\big\|_{L^{2}(B(\tilde{\bm{x}},1/3))} \\
&\quad\quad + \sup_{\kappa\in[-\pi,\pi]}\big\|i(\bm{x}-\bm{\tilde{x}})\cdot\bm{e}_1 e^{i\kappa (\bm{x}-\bm{\tilde{x}})\cdot\bm{e}_1} \mathcal{P}_{\kappa}(g\mathbbm{1}_{B(\bm{y},1/3)})_{\kappa}^{\wedge}\big\|_{L^{2}(B(\tilde{\bm{x}},1/3))} \\
&\quad\quad + \sup_{\kappa\in[-\pi,\pi]}\big\| e^{i\kappa (\bm{x}-\bm{\tilde{x}})\cdot\bm{e}_1}(\partial_{\kappa} \mathcal{P}_{\kappa})(g\mathbbm{1}_{B(\bm{y},1/3)})_{\kappa}^{\wedge}\big\|_{L^{2}(B(\tilde{\bm{x}},1/3))} \Big] \\
&\leq C\Big[\sup_{\kappa\in[-\pi,\pi]}\big\| \mathcal{P}_{\kappa}(g\mathbbm{1}_{B(\bm{y},1/3)})_{\kappa}^{\wedge}\big\|_{L^{2}(B(\tilde{\bm{x}},1/3))}  + \sup_{\kappa\in[-\pi,\pi]}\big\| \mathcal{P}_{\kappa}\partial_{\kappa}(g\mathbbm{1}_{B(\bm{y},1/3)})_{\kappa}^{\wedge}\big\|_{L^{2}(B(\tilde{\bm{x}},1/3))} \\
&\quad\quad + |x_1-y_1|\sup_{\kappa\in[-\pi,\pi]}\big\| \mathcal{P}_{\kappa}(g\mathbbm{1}_{B(\bm{y},1/3)})_{\kappa}^{\wedge}\big\|_{L^{2}(B(\tilde{\bm{x}},1/3))} + \sup_{\kappa\in[-\pi,\pi]}\big\| (\partial_{\kappa} \mathcal{P}_{\kappa})(g\mathbbm{1}_{B(\bm{y},1/3)})_{\kappa}^{\wedge}\big\|_{L^{2}(B(\tilde{\bm{x}},1/3))} \Big],
\end{aligned}
\end{equation}
where $C=C(\lambda)>0$ depends only on $|\partial_{\kappa}\lambda^{E}(\kappa^{\pm}(\lambda))|$. The proof of Step 1 applies to the first three terms in \eqref{eq_out_going_green_13}, which produces an estimate in the form of \eqref{eq_out_going_green_10}. For the last term, we recall \eqref{eq_out_going_green_14} and apply the resolvent identity
\begin{equation} \label{eq_out_going_green_15}
\partial_{\kappa}\mathcal{P}_{\kappa}=-\frac{1}{2\pi i}\oint_{C_{\lambda^{E}(\kappa),d_*/2}}(\mathcal{L}^E_{\kappa,\bm{e}_1}-z)^{-1}\partial_{\kappa}\mathcal{L}^E_{\kappa,\bm{e}_1}(\mathcal{L}^E_{\kappa,\bm{e}_1}-z)^{-1}dz .
\end{equation}
Recalling that $\partial_{\kappa}\mathcal{L}^E_{\kappa,\bm{e}_1}$ is a first-order differential operator and \eqref{eq_out_going_green_16} also holds for $\nabla (\mathcal{L}^E_{\kappa,\bm{e}_1}-z)^{-1}$, one can then check that $\partial_{\kappa}\mathcal{P}_{\kappa}$ admits a transverse decay as in \eqref{eq_out_going_green_13} by writing the composition in \eqref{eq_out_going_green_15} as a double integral. Note that the singularity in the Green functions is integrable. A detailed calculation follows the lines in \cite[Eq. (6.11)]{qiu2025bec_finite}. Then we arrive at \eqref{eq_out_going_green_10}.
\end{proof}

Finally, we note that, by the reflection symmetry $\mathcal{L}^{E}\mathcal{F}=\mathcal{F}\mathcal{L}^{E}$, the Green function $G^{E}(\bm{x},\bm{y};\lambda+i\epsilon)$ with $\epsilon>0$ (the integral kernel of $(\mathcal{L}^{E}-(\lambda+i\epsilon))^{-1}$) is covariant
\begin{equation*}
G^{E}(F\bm{x},F\bm{y};\lambda+i\epsilon)=G^{E}(\bm{x},\bm{y};\lambda+i\epsilon).
\end{equation*}
Hence, letting $\epsilon\to 0$, we know that the covariance also holds for the out-going Green function.
\begin{proposition} \label{prop_reflection_covariance}
For all $\lambda\in\mathcal{I}_{0}$, it holds that
\begin{equation*}
G^{E,out}(F\bm{x},F\bm{y};\lambda)=G^{E,out}(\bm{x},\bm{y};\lambda).
\end{equation*}
\end{proposition}

\subsection{Analytic Extension}
\label{sec_LA_out_green_analytic}

In this section, we construct an analytic extension of the Green operator $\mathcal{G}^{E,out}(\lambda)$ following the idea of \cite[Section 4.1]{qiu2026embedded}. This is achieved using the integral contour $C_1=C_1(\eta)$ introduced in Figure \ref{fig_integral_contours}. Specifically, we fix $\lambda\in\mathcal{I}_{0}$ and select  $\eta=\eta(\lambda)\in (0,1)$ such that \eqref{eq_out_going_green_3} and \eqref{eq_out_going_green_4} hold, i.e.,
\begin{equation*}
\begin{aligned}
\big[\mathcal{G}^{E,out}(\lambda) u\big](\bm{x})
&=\frac{1}{2\pi}\int_{-\pi}^{\pi}d\kappa \big[\big(\mathcal{L}^{E}_{\kappa,\bm{e}_1}\mathcal{Q}_{\kappa}-\lambda\big)^{-1}\mathcal{Q}_{\kappa}\hat{u}_{\kappa}\big](\tilde{\bm{x}})e^{i\kappa (\bm{x}-\tilde{\bm{x}})\cdot\bm{e}_1} \\
&\quad + \frac{1}{2\pi}\int_{C_1}d\kappa \frac{\big(\mathcal{P}_{\kappa}\hat{u}_{\kappa}\big)(\tilde{\bm{x}})}{\lambda^{E}(\kappa)-\lambda}e^{i\kappa (\bm{x}-\tilde{\bm{x}})\cdot\bm{e}_1} .
\end{aligned}
\end{equation*}
By the continuity of $\lambda^E(\kappa)$ and the fact that $C_1$ is compact, we can select an open ball $\mathcal{N}(\lambda)\ni \lambda$ such that
\begin{equation*}
z\neq \lambda^{E}(\kappa) \quad \text{for any $\kappa\in C_1$.}
\end{equation*}
Hence, the fraction $1/(\lambda^{E}(\kappa)-z)$ is analytic for $z\in \mathcal{N}(\lambda)$ and is continuous in $\kappa$. Moreover, restricting $\mathcal{N}(\lambda)\cap\mathbb{R}\subset \mathcal{I}_{0}$, we also have $z\mapsto\big(\mathcal{L}^{E}_{\kappa,\bm{e}_1}\mathcal{Q}_{\kappa}-z\big)^{-1}\mathcal{Q}_{\kappa}$ in $\mathcal{N}(\lambda)$. This implies that if we define
\begin{equation} \label{eq_analytic_extension_proof_1}
\begin{aligned}
z\ni\mathcal{N}(\lambda)\mapsto \big[\tilde{\mathcal{G}}^{E,out}(z) u\big](\bm{x})
&:=\frac{1}{2\pi}\int_{-\pi}^{\pi}d\kappa \big[\big(\mathcal{L}^{E}_{\kappa,\bm{e}_1}\mathcal{Q}_{\kappa}-z\big)^{-1}\mathcal{Q}_{\kappa}\hat{u}_{\kappa}\big](\tilde{\bm{x}})e^{i\kappa (\bm{x}-\tilde{\bm{x}})\cdot\bm{e}_1} \\
&\quad + \frac{1}{2\pi}\int_{C_1}d\kappa \frac{\big(\mathcal{P}_{\kappa}\hat{u}_{\kappa}\big)(\tilde{\bm{x}})}{\lambda^{E}(\kappa)-z}e^{i\kappa (\bm{x}-\tilde{\bm{x}})\cdot\bm{e}_1}  ,
\end{aligned}
\end{equation}
then $\tilde{\mathcal{G}}^{E,out}(z)$ is analytic and coincides with $\mathcal{G}^{E,out}(\lambda)$ for $z=\lambda$. In fact, a stronger claim is true: by controlling the size of $\eta$, we can show that
\begin{equation} \label{eq_analytic_extension_proof_4}
\tilde{\mathcal{G}}^{E,out}(z)=\mathcal{G}^{E,out}(z)\quad \text{for any $z\in \mathcal{N}(\lambda)\cap\mathbb{R}$}.
\end{equation}
To see this, let us choose $\eta$ so small that
\begin{equation} \label{eq_analytic_extension_proof_2}
\begin{aligned}
\Im \lambda^{E}(\kappa)<0,\quad\text{for } \kappa\in\text{int}(C_{\lambda,\eta}^{+})\cup \text{int}(C_{\lambda,\eta}^{-}),
\end{aligned}
\end{equation}
where $\text{int}(C_{\lambda,\eta}^{\pm})$ denotes the semi-disk enclosed by $C_{\lambda,\eta}^{\pm}$ and the real line (recall the definition \eqref{eq_C1_contour} of $C_{\lambda,\eta}^{\pm}$). This is possible due to the fact that $\pm\partial\lambda^{E}(\kappa^{\pm}(\lambda))>0$. By \eqref{eq_analytic_extension_proof_2}, we can deform the $C_1$-integral in \eqref{eq_analytic_extension_proof_1} to make the semi-circles centered at $\kappa^{\pm}(z)$ and with their radius turning to zero, during which $\lambda^{E}(\kappa)\neq z$ stays true and hence no poles are swept: 
\begin{equation} \label{eq_analytic_extension_proof_3}
\begin{aligned}
&\frac{1}{2\pi}\int_{C_1}d\kappa \frac{\big(\mathcal{P}_{\kappa}\hat{u}_{\kappa}\big)(\tilde{\bm{x}})}{\lambda^{E}(\kappa)-z}e^{i\kappa (\bm{x}-\tilde{\bm{x}})\cdot\bm{e}_1} \\
&=\lim_{\varepsilon\to 0}\Big[\frac{1}{2\pi}\Big( \int_{-\pi}^{\kappa^{-}(z)-\varepsilon}+\int_{\kappa^{-}(z)+\varepsilon}^{\kappa^{+}(z)-\varepsilon} + \int_{\kappa^{+}(z)+\varepsilon}^{\pi}\Big)d\kappa \frac{\big(\mathcal{P}_{\kappa}\hat{u}_{\kappa}\big)(\tilde{\bm{x}})}{\lambda^{E}(\kappa)-z}e^{i\kappa (\bm{x}-\tilde{\bm{x}})\cdot\bm{e}_1} \\
&\quad\quad\quad +\frac{1}{2\pi} \int_{C_{z,\varepsilon}^{+}\cup C_{z,\varepsilon}^{-}} d\kappa \frac{\big(\mathcal{P}_{\kappa}\hat{u}_{\kappa}\big)(\tilde{\bm{x}})}{\lambda^{E}(\kappa)-z}e^{i\kappa (\bm{x}-\tilde{\bm{x}})\cdot\bm{e}_1} \Big] \\
&=\frac{1}{2\pi}\text{p.v.}\int_{-\pi}^{\pi}d\kappa \frac{\big(\mathcal{P}_{\kappa}\hat{u}_{\kappa}\big)(\tilde{\bm{x}})}{\lambda^{E}(\kappa)-z}e^{i\kappa (\bm{x}-\tilde{\bm{x}})\cdot\bm{e}_1}  + \sum_{\sigma\in\{\pm\}}\frac{i}{2|\partial_{\kappa}\lambda^{E}\big(\kappa^{\sigma}(z)\big)|}\big(\mathcal{P}_{\kappa^{\sigma}(z)}\hat{u}_{\kappa^{\sigma}(z)}\big)(\tilde{\bm{x}})e^{i\kappa^{\sigma}(z) (\bm{x}-\tilde{\bm{x}})\cdot\bm{e}_1},
\end{aligned}
\end{equation}
where the second equality follows by the definition of the Cauchy principal-value integral and the residue formula. Then we conclude the proof of \eqref{eq_analytic_extension_proof_4}.

Finally, we note that the above extension guarantees that
\begin{equation}  \label{eq_analytic_extension_proof_5}
\tilde{\mathcal{G}}^{E,out}(z)=(\mathcal{L}^{E}-z)^{-1}\quad \text{ for $z\in \mathcal{N}(\lambda)\cap \mathbb{C}^{+}$}.
\end{equation}
This also follows from \eqref{eq_analytic_extension_proof_2}, which guarantees that one can deform the $C_1$-integral in \eqref{eq_analytic_extension_proof_1} to be over $[-\pi,\pi]$ without sweeping any poles if $\Im z>0$. Then we prove \eqref{eq_analytic_extension_proof_5} by recalling the Floquet expansion formula \eqref{eq_out_going_green_1}.

In conclusion, we have constructed an analytic extension of the out-going Green operator via \eqref{eq_analytic_extension_proof_1}, locally near each $\lambda\in\mathcal{I}_0$, which satisfies properties \eqref{eq_analytic_extension_proof_4} and \eqref{eq_analytic_extension_proof_5}. Covering the interval $\mathcal{I}_0$ with finitely many balls $\mathcal{N}(\lambda)$ by a compactness argument, this analytic extension holds along the whole range of frequencies $\mathcal{I}_0$. We summarize these results as follows.
\begin{proposition} \label{prop_Green_operator_analytic_extension}
There exists a complex neighborhood $\mathcal{N}\supset \mathcal{I}_{0}$ and an analytic operator-valued map $\mathcal{N}\ni \lambda\mapsto  \tilde{\mathcal{G}}^{E,out}(\lambda)$, with $\tilde{\mathcal{G}}^{E,out}(z):H^{-1}_{x-comp}\to H^1_y$ satisfying
\begin{equation} \label{eq_Green_operator_analytic_extension_1}
\tilde{\mathcal{G}}^{E,out}(\lambda)=\mathcal{G}^{E,out}(\lambda)\quad \text{for any $\lambda\in \mathcal{N}\cap\mathbb{R}$},
\end{equation}
and
\begin{equation}  \label{eq_Green_operator_analytic_extension_2}
\tilde{\mathcal{G}}^{E,out}(\lambda)=(\mathcal{L}^{E}-\lambda)^{-1}\quad \text{ for $\lambda\in \mathcal{N}\cap \mathbb{C}^{+}$}.
\end{equation}
As in \eqref{eq_out_going_Green_operator_bound}, the operator $\tilde{\mathcal{G}}^{E,out}(\lambda)$ is bounded in the sense that
\begin{equation} \label{eq_Green_operator_analytic_extension_3}
\big\|\tilde{\mathcal{G}}^{E,out}(\lambda)u\big\|_{H^1(S)}\leq M\|u\|_{H^{-1}(\mathbb{R}^2)},
\end{equation}
with $M=M(\lambda)>0$. Finally, for all $\lambda\in \mathcal{N}$ and $u\in H^{-1}_{x-comp}$, we have
\begin{equation}  \label{eq_Green_operator_analytic_extension_4}
(\mathcal{L}^{E}-\lambda)\tilde{\mathcal{G}}^{E,out}(\lambda)u=u.
\end{equation}
\end{proposition}
Note that the bound \eqref{eq_Green_operator_analytic_extension_3} for $\lambda\notin \mathbb{R}$ is proved by establishing the estimate for the terms in \eqref{eq_analytic_extension_proof_1}, involving reduced resolvent and projection $\mathcal{P}_{\kappa}$, respectively. Specifically, since all $z\in \mathcal{N}$ are isolated from the spectrum of the Fredholm operator $\mathcal{L}^{E}_{\kappa,\bm{e}_1}\mathcal{Q}_{\kappa}:H^1_{\kappa,\bm{e}_1}(S)\to H^{-1}_{\kappa,\bm{e}_1}(S)$, with a uniform isolation distance $d_*>0$ by Assumption \ref{asmp_detailed_edge_band}(ii), the Fredholm alternative states that
\begin{equation*}
\big\| \big(\mathcal{L}^{E}_{\kappa,\bm{e}_1}\mathcal{Q}_{\kappa}-\lambda\big)^{-1}\mathcal{Q}_{\kappa} \big\|_{\mathcal{B}(H^{-1}(S),H^1(S))}\leq M_1, \quad \forall \kappa\in [-\pi,\pi].
\end{equation*}
where $M_1=M_1(d_*)>0$. On the other hand, as the analytic extension of $\mathcal{P}_{\kappa}$ holds in $H^1$ norm and $z\neq \lambda^{E}(\kappa)$ for all $\kappa\in C_1$, we can check that
%proved in P.358 of Joly's
\begin{equation*}
\big\| \frac{\mathcal{P}_{\kappa}}{\lambda^{E}(\kappa)-z} \big\|_{\mathcal{B}(H^{-1}(S),H^1(S))}\leq M_2, \quad \forall \kappa\in C_1.
\end{equation*}
with $M_2=M_2(\lambda)>0$. Then we conclude the proof of \eqref{eq_Green_operator_analytic_extension_3} by these estimates and recalling the definition \eqref{eq_analytic_extension_proof_1}. The details are omitted here. Finally, we note that $(\mathcal{L}^{E}-\lambda)\tilde{\mathcal{G}}^{E,out}(\lambda)u\equiv u$ holds for all $\lambda\in \mathcal{N}\cap \mathbb{R}$ by \eqref{eq_Green_operator_analytic_extension_1}; therefore, the identity \eqref{eq_Green_operator_analytic_extension_4} holds for all complex parameters $\lambda\in\mathcal{N}$ by the analyticity of the vector-valued map $\lambda\mapsto (\mathcal{L}^{E}-\lambda)\tilde{\mathcal{G}}^{E,out}(\lambda)u$.

\begin{remark} \label{rmk_analytic_Green_radiation}
One can also derive a radiation condition for $\tilde{\mathcal{G}}^{E,out}(\lambda)$ in the form of \eqref{eq_out_going_Green_asymptotics_positive} and \eqref{eq_out_going_Green_asymptotics_negative} for non-real $\lambda\in \mathcal{N}$. If $\Im \lambda>0$, the property \eqref{eq_Green_operator_analytic_extension_2} simply implies that $\tilde{\mathcal{G}}^{E,out}(\lambda)u$ decays exponentially for any $u\in H^{-1}_{x-comp}$ according to the Combes-Thomas estimate. The situation is a little different for $\Im \lambda<0$. In that case, when deforming the $C_1$-integral in \eqref{eq_analytic_extension_proof_1} to be over $[-\pi,\pi]$, as in the derivation of \eqref{eq_Green_operator_analytic_extension_2}, it will sweep two poles locating at $\kappa=\kappa^{\pm}(\lambda)$, i.e., the zeros of $\lambda^{E}(\kappa)=\lambda$, with 
\begin{equation*}
\pm \Re \kappa^{\pm}(\lambda)>0,\quad \mp \Im \kappa^{\pm}(\lambda)<0.
\end{equation*}
Hence, the residue formula will produce
\begin{equation*}
\begin{aligned}
\big[\tilde{\mathcal{G}}^{E,out}(\lambda)u\big](\bm{x})=\big[(\mathcal{L}^{E}-\lambda)^{-1}u\big](\bm{x})&+\frac{i}{\partial_{\kappa}\lambda^{E}\big(\kappa^{-}(\lambda)\big)}\big(\mathcal{P}_{\kappa^{-}(\lambda)}\hat{u}_{\kappa^{-}(\lambda)}\big)(\tilde{\bm{x}})e^{i\kappa^{-}(\lambda) (\bm{x}-\tilde{\bm{x}})\cdot\bm{e}_1} \\
&-\frac{i}{\partial_{\kappa}\lambda^{E}\big(\kappa^{+}(\lambda)\big)}\big(\mathcal{P}_{\kappa^{+}(\lambda)}\hat{u}_{\kappa^{+}(\lambda)}\big)(\tilde{\bm{x}})e^{i\kappa^{+}(\lambda) (\bm{x}-\tilde{\bm{x}})\cdot\bm{e}_1},
\end{aligned}
\end{equation*}
which implies $\tilde{\mathcal{G}}^{E,out}(\lambda)u$ blows up exponentially as $x_1\to\pm\infty$ if the coupling $\mathcal{P}_{\kappa^{\pm}(\lambda)}\hat{u}_{\kappa^{\pm}(\lambda)}\neq 0$. The details are left to the reader.
\end{remark}

\section{Layer-Potential Operators Associated with the Out-going Green Function}
\label{sec_LP_operators}

In this section, we introduce the layer-potential operators associated with the analytic-extended out-going Green function on the auxiliary interface $\Gamma$. In particular, we will prove that the trace of single-layer potential is a Fredholm operator with zero index, which is the key point in the proof of Theorem \ref{thm_analytic_Fredholm_SL}, as seen in Section \ref{sec_bend_immunity}.

\subsection{Boundedness, Analyticity, and Symmetry}
\label{sec_LP_operators_bound_analytic_symmetry}

As in \eqref{eq_layer_potential_left_evan_2}, the SL and DL potentials associated with $\tilde{G}^{E,out}$ (integral kernel of $\tilde{\mathcal{G}}^{E,out}$) are defined as 
\begin{equation} \label{eq_analytic_extended_SL_DL}
\begin{aligned}
&\mathcal{S}(\lambda;\tilde{G}^{E,out})\big[\varphi\big](\bm{x}):=\int_{\Gamma}\tilde{G}^{E,out}(\bm{x},\bm{y};\lambda)\varphi(\bm{y})d\sigma({\bm{y}}),\quad \varphi\in H^{-\frac{1}{2}}(\Gamma),\, \bm{x}\in \mathbb{R}^2\backslash \Gamma, \\
&\mathcal{D}(\lambda;\tilde{G}^{E,out})\big[\phi\big](\bm{x}):=\int_{\Gamma}\partial_{\nu,a^{E},\bm{y}}\tilde{G}^{E,out}(\bm{x},\bm{y};\lambda)\phi(\bm{y})d\sigma({\bm{y}}),\quad \phi\in H^{\frac{1}{2}}(\Gamma),\, \bm{x}\in \mathbb{R}^2\backslash \Gamma . \\
\end{aligned}
\end{equation}
These operators are bounded, mapping the boundary densities to functions in the unit strip $S$, in the following sense:
\begin{proposition}[Boundedness of the analytic SL/DL potentials] \label{prop_analytic_extended_SL_DL_bound}
For $\lambda\in \mathcal{N}$, it holds that
\begin{equation} \label{eq_analytic_extended_SL_DL_bound_1}
\|\mathcal{S}(\lambda;\tilde{G}^{E,out})\varphi\|_{H^1(S)}\leq C\|\varphi\|_{H^{-\frac{1}{2}}(\Gamma)},\quad \forall \varphi \in H^{-\frac{1}{2}}(\Gamma),
\end{equation}
\begin{equation} \label{eq_analytic_extended_SL_DL_bound_2}
\|\mathcal{D}(\lambda;\tilde{G}^{E,out})\phi\|_{H^1(S)}\leq C\|\phi\|_{H^{\frac{1}{2}}(\Gamma)},\quad \forall \phi \in H^{\frac{1}{2}}(\Gamma),
\end{equation}
where $C=C(\lambda)>0$
\end{proposition}

\begin{proof}
Here, we only prove the statement for $\lambda\in \mathcal{N}\cap \mathbb{R}$, and point out the necessary modification when considering non-real $\lambda$. The proof follows the same lines as in \cite[Theorem 6.11]{mclean2000strongly}. First, to show \eqref{eq_analytic_extended_SL_DL_bound_1}, we define $J_{\Gamma}:H^{-\frac{1}{2}}(\Gamma)\to H^{-1}_{x-comp}$ with $J_{\Gamma}\varphi:=\varphi\cdot \delta_{\Gamma}$ (i.e., boundary-to-volume lifting of densities). It follows immediately that
\begin{equation*}
\mathcal{S}(\lambda;\tilde{G}^{E,out})\varphi = \tilde{\mathcal{G}}^{E,out}\big[J_{\Gamma}\varphi \big],
\end{equation*}
which proves \eqref{eq_analytic_extended_SL_DL_bound_1} by recalling the estimate \eqref{eq_Green_operator_analytic_extension_3}. The proof of \eqref{eq_analytic_extended_SL_DL_bound_2} is carried out by expressing the DL potential in terms of the SL potential via the Green identity. To achieve this, we fix $\lambda_0<0$ and denote $\gamma^{R/L}:H^1(\Omega_{R/L})\to H^{\frac{1}{2}}(\Gamma)$ by the trace operator from the right/left half plane to the imaginary interface $\Gamma$. Then the Dirichlet problem
\begin{equation*}
\left\{
\begin{aligned}
&(\mathcal{L}^{E}-\lambda_0)u=0 \quad \text{in $\Omega_R$,} \\
&\gamma^R u=\phi\in H^{\frac{1}{2}}(\Gamma) \quad \text{in $\Gamma$,}
\end{aligned}
\right.
\end{equation*}
has a unique solution $u=u_{\phi}\in H^{1}(\Omega_R)$ because $\mathcal{L}^{E}$ has a positive spectrum. Hence, the solution operator $\mathscr{R}_{\lambda_0}:\phi\mapsto u_{\phi}$ is well-defined. Moreover, $u_{\phi}$ decays exponentially from $\Gamma$, by the Combes-Thomas estimate applied to the resolvent $\big(\mathcal{L}^{E}\big|_{H_0^1(\Omega_R)}-\lambda_0\big)^{-1}$, in the sense that
\begin{equation*}
\|u_{\phi}\|_{L^2(S_n)}\leq c_1e^{-c_2 n} \|\phi\|_{H^{\frac{1}{2}}(\Gamma)}\quad \text{for some $c_i=c_i(\lambda_0)>0$ and $n\geq 0$.}
\end{equation*}
This $L^2$ estimate is, in fact, raised to $H^1$ by the regularity theory for elliptic equations
\begin{equation} \label{eq_analytic_extended_SL_DL_bound_proof_1}
\|u_{\phi}\|_{H^2(S_n)}\leq C_1e^{-C_2 n} \|\phi\|_{H^{\frac{1}{2}}(\Gamma)}\quad \text{for $n\geq 0$.}
\end{equation}
Now, we note that $u_{\phi}$ solves the following equation:
\begin{equation*}
(\mathcal{L}^{E}-\lambda)u_{\phi}=-(\lambda-\lambda_0)u_{\phi}\quad \text{in $\Omega_R$.}
\end{equation*}
Thus, the Green identity gives, for any $\bm{x}\in \Omega_R$,
\begin{equation} \label{eq_analytic_extended_SL_DL_bound_proof_2}
\begin{aligned}
u_{\phi}(\bm{x})&= -(\lambda-\lambda_0)\int_{\cup_{0\leq n\leq N}S_{n,m}}G^{E,out}(\bm{x},\bm{y};\lambda)u_{\phi}(\bm{y})d\bm{y} \\
&\quad -\int_{\cup_{0\leq n\leq N}\partial S_{n,m}\cap \{x_2=m\}}
\Big[ \partial_{y_2}G^{E,out}(\bm{x},\bm{y};\lambda)u_{\phi}(\bm{y})-G^{E,out}(\bm{x},\bm{y};\lambda)\partial_{y_2}u_{\phi}(\bm{y}) \Big] d\sigma(\bm{y}) \\
&\quad +\int_{\cup_{0\leq n\leq N}\partial S_{n,m}\cap \{x_2=-m\}}
\Big[ \partial_{y_2}G^{E,out}(\bm{x},\bm{y};\lambda)u_{\phi}(\bm{y})-G^{E,out}(\bm{x},\bm{y};\lambda)\partial_{y_2}u_{\phi}(\bm{y}) \Big] d\sigma(\bm{y}) \\
&\quad -\int_{\Gamma_{N+1}\cap \{|x_2|<m\}}
\Big[ \partial_{\nu,a^{E},\bm{y}}G^{E,out}(\bm{x},\bm{y};\lambda)u_{\phi}(\bm{y})-G^{E,out}(\bm{x},\bm{y};\lambda)\partial_{\nu,a^{E}}u_{\phi}(\bm{y}) \Big] d\sigma(\bm{y}) \\
&\quad +\int_{\Gamma\cap \{|x_2|<m\}}
\Big[ \partial_{\nu,a^{E},\bm{y}}G^{E,out}(\bm{x},\bm{y};\lambda)u_{\phi}(\bm{y})-G^{E,out}(\bm{x},\bm{y};\lambda)\partial_{\nu,a^{E}}u_{\phi}(\bm{y}) \Big] d\sigma(\bm{y})
\end{aligned}
\end{equation}
for sufficiently large $m,N>0$, where the parallelogram $S_{n,m}:=S_{n}\cap \{|x_2|<m\}$ (i.e. truncation of the $n$-strip with a width of $2m>0$.) By Proposition \ref{prop_perpendicular_decay_Green_function}, one sees that the boundary integral in \eqref{eq_analytic_extended_SL_DL_bound_proof_2} over the upper and lower segments converges to zero as $m\to\infty$ (on which $u_{\phi}$ has bounded trace and conormal derivatives). Hence,
\begin{equation} \label{eq_analytic_extended_SL_DL_bound_proof_3}
\begin{aligned}
u_{\phi}(\bm{x})&= -(\lambda-\lambda_0)\int_{\cup_{0\leq n\leq N}S_{n}}G^{E,out}(\bm{x},\bm{y};\lambda)u_{\phi}(\bm{y})d\bm{y} \\
&\quad -\int_{\Gamma_{N+1}}
\Big[ \partial_{\nu,a^{E},\bm{y}}G^{E,out}(\bm{x},\bm{y};\lambda)u_{\phi}(\bm{y})-G^{E,out}(\bm{x},\bm{y};\lambda)\partial_{\nu,a^{E}}u_{\phi}(\bm{y}) \Big] d\sigma(\bm{y}) \\
&\quad +\int_{\Gamma}
\Big[ \partial_{\nu,a^{E},\bm{y}}G^{E,out}(\bm{x},\bm{y};\lambda)u_{\phi}(\bm{y})-G^{E,out}(\bm{x},\bm{y};\lambda)\partial_{\nu,a^{E}}u_{\phi}(\bm{y}) \Big] d\sigma(\bm{y}) .
\end{aligned}
\end{equation}
Next, pushing $N\to\infty$, the decay property \eqref{eq_analytic_extended_SL_DL_bound_proof_1} of $u_{\phi}$ together with the boundedness of out-going Green operator in \eqref{eq_out_going_Green_operator_bound} implies that the volume integral in \eqref{eq_analytic_extended_SL_DL_bound_proof_3} converges uniformly and the boundary integral over $\Gamma_{N+1}$ vanishes\footnote{This also holds for $\Im \lambda>0$ because the out-going Green function $G^{E,out}(\bm{x},\bm{y};\lambda)$ is still bounded as $|x_1-y_1|\to\infty$. For $\Im \lambda<0$, one just needs to let $|\lambda_0|$ to be sufficiently large such that the decay rate of $u_{\phi}$ in \eqref{eq_analytic_extended_SL_DL_bound_proof_1} overwhelms the blowing-up rate of $G^{E,out}$; see Remark \ref{rmk_analytic_Green_radiation}.}, which gives
\begin{equation} \label{eq_analytic_extended_SL_DL_bound_proof_4}
\begin{aligned}
u_{\phi}(\bm{x})&= -(\lambda-\lambda_0)\int_{\cup_{n\geq 0}S_{n}}G^{E,out}(\bm{x},\bm{y};\lambda)u_{\phi}(\bm{y})d\bm{y} \\
&\quad +\int_{\Gamma}
\Big[ \partial_{\nu,a^{E},\bm{y}}G^{E,out}(\bm{x},\bm{y};\lambda)u_{\phi}(\bm{y})-G^{E,out}(\bm{x},\bm{y};\lambda)\partial_{\nu,a^{E}}u_{\phi}(\bm{y}) \Big] d\sigma(\bm{y}) .
\end{aligned}
\end{equation}
In conclusion, recalling the identity $\gamma^R u_{\phi}=\phi$, we obtain
\begin{equation*}
\begin{aligned}
\big[\mathcal{D}(\lambda;G^{E,out})\phi\big](\bm{x})
=&\big[\mathcal{S}(\lambda;G^{E,out})(\partial_{a,\nu}\mathscr{R}_{\lambda_0}\phi)\big](\bm{x}) \\
&+[\mathscr{R}_{\lambda_0}\phi](\bm{x}) 
+(\lambda-\lambda_0)\sum_{n\geq 0}\big[\mathcal{G}^{E,out}(\mathscr{R}_{\lambda_0}\phi\circ \mathbbm{1}_{S_n}) \big](\bm{x}) .
\end{aligned}
\end{equation*}
    Hence, the estimate \eqref{eq_analytic_extended_SL_DL_bound_2} follows by the boundedness of the SL potential \eqref{eq_analytic_extended_SL_DL_bound_1}, the solution operator $\mathscr{R}_{\lambda_0}$ and the Green operator $\mathcal{G}^{E,out}$.
\end{proof}

As a byproduct of the Green-identity argument in the proof, we have justified the layer-potential formulation of evanescent waves claimed in Proposition \ref{prop_layer_potential_left_evan}:

\begin{proposition} \label{prop_green_formula_left_evan}
Let $\lambda\in\mathcal{I}_{0}$. Suppose that $w^{evan,L}\in L^2(\Omega_L)$ solves $(\mathcal{L}^{E}-\lambda)w^{evan,L}=0$ in $\Omega_L$ and satisfies the estimate
\begin{equation*}
\|w^{evan,L}\|_{H^1(S_n)}\leq e^{-c|n|},\quad \text{for $n<0$}.
\end{equation*}
Then, for $\bm{x}\in \Omega_{L}$, the Green formula holds for $w^{evan,L}$
\begin{equation*}
w^{evan,L}(\bm{x})=-\mathcal{D}(\lambda;G^{E,out})\big[\gamma^{-} w^{evan,L}\big]+\mathcal{S}(\lambda;G^{E,out})\big[\partial_{\nu,a^{E}}^{-}w^{evan,L} \big] .
\end{equation*}
\end{proposition}

By Proposition \ref{prop_analytic_extended_SL_DL_bound}, the trace/conormal derivatives of the layer-potential operators are bounded in their respective domains, i.e., 
\begin{equation*}
\begin{aligned}
&\gamma^{\pm}\mathcal{S}(\lambda;\tilde{G}^{E,out}):\, H^{-\frac{1}{2}}(\Gamma) \to H^{\frac{1}{2}}(\Gamma),
\quad \partial_{\nu,a^{E}}^{\pm}\mathcal{S}(\lambda;\tilde{G}^{E,out}):\, H^{-\frac{1}{2}}(\Gamma) \to H^{-\frac{1}{2}}(\Gamma), \\
&\gamma^{\pm}\mathcal{D}(\lambda;\tilde{G}^{E,out}):\, H^{\frac{1}{2}}(\Gamma) \to H^{\frac{1}{2}}(\Gamma),
\quad \partial_{\nu,a^{E}}^{\pm}\mathcal{D}(\lambda;\tilde{G}^{E,out}):\, H^{\frac{1}{2}}(\Gamma) \to H^{-\frac{1}{2}}(\Gamma) .
\end{aligned}
\end{equation*}
Recalling the following jumping identities of the layer-potential operators, which hold for $\lambda\in \mathcal{N}\cap \mathbb{C}^{+}$ (cf. \cite[Theorem 6.11]{mclean2000strongly}) and hence extend to the whole domain $\mathcal{N}$ by analyticity,
\begin{equation*}
\begin{aligned}
&(\gamma^{+}-\gamma^{-})\mathcal{S}(\lambda;\tilde{G}^{E,out})=0,\quad (\partial_{\nu,a^{E}}^{+}-\partial_{\nu,a^{E}}^{-})\mathcal{S}(\lambda;\tilde{G}^{E,out})[\varphi]=-\varphi , \\
&(\gamma^{+}-\gamma^{-}) \mathcal{D}(\lambda;\tilde{G}^{E,out})[\phi]=\phi,\quad (\partial_{\nu,a^{E}}^{+}-\partial_{\nu,a^{E}}^{-}) \mathcal{D}(\lambda;\tilde{G}^{E,out})[\varphi]=0 .
\end{aligned}
\end{equation*}
Thus, if we define
\begin{equation} \label{eq_analytic_boundary_operator_def}
\begin{aligned}
&S(\lambda;\tilde{G}^{E,out}):=\gamma \mathcal{S}(\lambda;\tilde{G}^{E,out}),\quad N(\lambda;\tilde{G}^{E,out}):=\partial_{\nu,a^{E}}\mathcal{D}(\lambda;\tilde{G}^{E,out}), \\
&K^{*}(\lambda;\tilde{G}^{E,out}):=\frac{1}{2}\Big[ \partial_{\nu,a^{E}}^{+}\mathcal{S}(\lambda;\tilde{G}^{E,out})+\partial_{\nu,a^{E}}^{-}\mathcal{S}(\lambda;\tilde{G}^{E,out}) \Big]. \\
&K(\lambda;\tilde{G}^{E,out}):=\frac{1}{2}\Big[ \gamma^{+} \mathcal{D}(\lambda;\tilde{G}^{E,out})+ \gamma^{-} \mathcal{D}(\lambda;\tilde{G}^{E,out}) \Big], \\
\end{aligned}
\end{equation}
we conclude that the following results hold. 
\begin{proposition}[Boundedness and Analyticity of Boundary Operators] \label{prop_boundary_operators_analytic_bounded_jump}
The following operator-valued maps are analytic and bounded in their respective domains:
\begin{equation*}
\begin{aligned}
&\mathcal{N}\ni \lambda \mapsto S(\lambda;\tilde{G}^{E,out})\in\mathcal{B}\big(H^{-\frac{1}{2}}(\Gamma),H^{\frac{1}{2}}(\Gamma)\big) ,\\
&\mathcal{N}\ni \lambda \mapsto N(\lambda;\tilde{G}^{E,out})\in\mathcal{B}\big(H^{\frac{1}{2}}(\Gamma),H^{-\frac{1}{2}}(\Gamma)\big) ,\\
&\mathcal{N}\ni \lambda \mapsto K^{*}(\lambda;\tilde{G}^{E,out})\in\mathcal{B}\big(H^{-\frac{1}{2}}(\Gamma),H^{-\frac{1}{2}}(\Gamma)\big) ,\\
&\mathcal{N}\ni \lambda \mapsto K(\lambda;\tilde{G}^{E,out})\in\mathcal{B}\big(H^{\frac{1}{2}}(\Gamma),H^{\frac{1}{2}}(\Gamma)\big) .\\
\end{aligned}
\end{equation*}
Moreover, it holds that
\begin{equation*}
\begin{aligned}
&\gamma^{\pm} \mathcal{S}(\lambda;\tilde{G}^{E,out})=S(\lambda;\tilde{G}^{E,out}),
\quad \partial_{\nu,a^{E}}^{\pm}\mathcal{S}(\lambda;\tilde{G}^{E,out})=\mp \frac{1}{2} + K^{*}(\lambda;\tilde{G}^{E,out}) ,\\
&\gamma^{\pm} \mathcal{D}(\lambda;\tilde{G}^{E,out})=\pm\frac{1}{2}+K(\lambda;\tilde{G}^{E,out}),
\quad \partial_{\nu,a^{E}}^{\pm}\mathcal{D}(\lambda;\tilde{G}^{E,out})=N(\lambda;\tilde{G}^{E,out}).
\end{aligned}
\end{equation*}
\end{proposition}
Similarly, since the following Calderón identity holds for all $\lambda\in\mathcal{N}\cap\mathbb{C}^{+}$ (cf. \cite[Exercise 7.6]{mclean2000strongly})
\begin{equation*}
\begin{aligned}
S(\lambda;\tilde{G}^{E,out})K^{*}(\lambda;\tilde{G}^{E,out})&=K(\lambda;\tilde{G}^{E,out})S(\lambda;\tilde{G}^{E,out}) \\
S(\lambda;\tilde{G}^{E,out})N(\lambda;\tilde{G}^{E,out})&=-\frac{1}{4}\Big(1- K^2(\lambda;\tilde{G}^{E,out})\Big)
\end{aligned}
\end{equation*}
it extends to all $\lambda\in\mathcal{N}$ by analyticity. 
\begin{proposition} \label{prop_calderon_identity}
For all $\lambda\in\mathcal{N}$, it holds that
\begin{equation*}
\begin{aligned}
S(\lambda;\tilde{G}^{E,out})K^{*}(\lambda;\tilde{G}^{E,out})&=K(\lambda;\tilde{G}^{E,out})S(\lambda;\tilde{G}^{E,out}) \\
S(\lambda;\tilde{G}^{E,out})N(\lambda;\tilde{G}^{E,out})&=-\frac{1}{4}\Big(1- K^2(\lambda;\tilde{G}^{E,out})\Big)
\end{aligned}
\end{equation*}
\end{proposition}

Finally, utilizing the reflection symmetry of $\mathcal{L}^{E}$, we prove the identities \eqref{eq_bend_immunity_proof_10}-\eqref{eq_bend_immunity_proof_11}, which relate the boundary integral operators associated with $G^{E,out}$ to that of $G^{R^{-1}E,out}$.
\begin{proposition} \label{prop_boundary_integral_operators_relation}
For $\lambda\in \mathcal{I}_{0}$, it holds that
\begin{equation} \label{eq_boundary_integral_operators_relation}
\begin{aligned}
&S(\lambda;G^{R^{-1}E,out})=S(\lambda;G^{E,out}),\quad N(\lambda;G^{R^{-1}E,out})=N(\lambda;G^{E,out}), \\
&K^*(\lambda;G^{R^{-1}E,out})=-K^*(\lambda;G^{E,out}),\quad K(\lambda;G^{R^{-1}E,out})=-K(\lambda;G^{E,out}),
\end{aligned}
\end{equation}
where $O(\lambda;G^{R^{-1}E,out})$ ($O\in\{S,N,K,K^*\}$) is defined similarly to \eqref{eq_analytic_boundary_operator_def} with the Green function being given by $G^{R^{-1}E,out}(\bm{x},\bm{y};\lambda):=G^{E,out}(R\bm{x},R\bm{y};\lambda)$.
\end{proposition}

\begin{proof}
We have proved $S(\lambda;G^{R^{-1}E,out})=S(\lambda;G^{E,out})$ in Section \ref{sec_LP_formulation}; see the arguments following Proposition \ref{prop_layer_potential_formulation}. The proof for the other identities is similar. For example, for $\bm{y}\in\Gamma$ and $\bm{x}\neq \bm{y}$,
\begin{equation*}
\begin{aligned}
\partial_{a^{E},\nu,\bm{x}}\big[G^{R^{-1}E,out}(\bm{x},\bm{y};\lambda)\big]
&=\nu\cdot a^{E}(\bm{x})\nabla_{\bm{x}}\big[G^{E,out}(R\bm{x},R\bm{y};\lambda)\big] \\
&=\nu\cdot a^{E}(\bm{x})R^{\top}\big(\nabla_{\bm{x}}G^{E,out}\big)(R\bm{x},R\bm{y};\lambda).
\end{aligned}
\end{equation*}
Restricting $\bm{x}\in\Gamma$, it holds that
\begin{equation*}
    R\bm{x}=F\bm{x},\quad R\bm{y}=F\bm{y}.
\end{equation*}
Hence, for $\bm{x}\in \Gamma$, we have
\begin{equation} \label{eq_boundary_integral_operators_relation_proof_1}
\begin{aligned}
\nu\cdot a^{E}(\bm{x})\nabla_{\bm{x}} \big[G^{R^{-1}E,out}(\bm{x},\bm{y};\lambda)\big]
&=\nu\cdot a^{E}(\bm{x})R^{\top}\big(\nabla_{\bm{x}}G^{E,out}\big)(F\bm{x},F\bm{y};\lambda) \\
&=\nu\cdot a^{E}(\bm{x})R^{\top}F\nabla_{\bm{x}}\big[G^{E,out}(F\bm{x},F\bm{y};\lambda)\big] \\
&\overset{(i)}{=}(F^{\top}R\nu)\cdot a^{E}(\bm{x})\nabla_{\bm{x}}\big[G^{E,out}(\bm{x},\bm{y};\lambda)\big] \\
&\overset{(ii)}{=}-\nu\cdot a^{E}(\bm{x})\nabla_{\bm{x}}\big[G^{E,out}(\bm{x},\bm{y};\lambda)\big],
\end{aligned}
\end{equation}
where (i) follows from the reflection symmetry of the out-going Green function (see Proposition \ref{prop_reflection_covariance}), and (ii) is derived by noting that $F_{\Gamma}=F^{\top}R=FR$ and $F_{\Gamma}\nu=-\nu$. By \eqref{eq_boundary_integral_operators_relation_proof_1}, we know that
\begin{equation*}
\begin{aligned}
&\partial_{\nu,a^{E}}^{+}\mathcal{S}(\lambda;\tilde{G}^{R^{-1}E,out})[\varphi] + \partial_{\nu,a^{E}}^{-}\mathcal{S}(\lambda;\tilde{G}^{R^{-1}E,out})[\varphi] \\
&=-\Big[\partial_{\nu,a^{E}}^{+} \mathcal{S}(\lambda;\tilde{G}^{E,out})[\varphi] + \partial_{\nu,a^{E}}^{-}\mathcal{S}(\lambda;\tilde{G}^{E,out})[\varphi]\Big]
\end{aligned}
\end{equation*}
and hence,
\begin{equation*}
K^*(\lambda;G^{R^{-1}E,out})=-K^*(\lambda;G^{E,out})
\end{equation*}
by recalling the jump identities in Proposition \ref{prop_boundary_operators_analytic_bounded_jump}. The remaining identities in \eqref{eq_boundary_integral_operators_relation} are proved similarly.
\end{proof}

\subsection{Fredholmness} \label{sec_LP_fredholm}
In this section, we prove that the operator $S(\lambda;\tilde{G}^{E,out})\in\mathcal{B}\big(H^{-\frac{1}{2}}(\Gamma),H^{\frac{1}{2}}(\Gamma)\big)$, i.e., the trace of the SL potential, satisfies the following properties, based on which the conclusions of Theorem \ref{thm_analytic_Fredholm_SL} follow immediately.
\begin{proposition} \label{prop_SL_positive_imaginary_invertibility}
For $\lambda\in\mathcal{N}\cap \mathbb{C}^{+}$, $S(\lambda;\tilde{G}^{E,out})$ is invertible.
\end{proposition}
\begin{proposition} \label{prop_SL_real_fredholm}
For $\lambda\in\mathcal{N}\cap \mathbb{R}$, $S(\lambda;\tilde{G}^{E,out})$ is a Fredholm operator with zero index.
\end{proposition}
We prove these two propositions in two separate subsections. In the sequel, we denote by $\tilde{S}(\lambda)=S(\lambda;\tilde{G}^{E,out})$, and in particular,  $S(\lambda)=S(\lambda;\tilde{G}^{E,out})$ if $\lambda\in\mathbb{R}$, for ease of notation.

\subsubsection{Invertibility at Non-real Frequencies}
\label{sec_nonreal_invertible}

For the proof of Proposition \ref{prop_SL_positive_imaginary_invertibility}, we note that, thanks to the property \eqref{eq_Green_operator_analytic_extension_2} of the analytic extension of the Green operator, $\tilde{S}(\lambda)$ is exactly the trace of the SL operator associated with the resolvent $(\mathcal{L}^{E}-\lambda)^{-1}$, whose invertibility is standard. In fact, for non-real $\lambda$, the Dirichlet-to-Neumann map is well-defined and bounded:
\begin{equation*}
\Pi_{\pm}(\lambda):H^{\frac{1}{2}}(\Gamma)\to H^{-\frac{1}{2}}(\Gamma),\quad
\Pi_{\pm}(\lambda)\phi=\partial_{\nu,a^{E}}^{\pm}u_{R/L}\big|_{\Gamma_{+/-}},
\end{equation*}
where $u_{+}$ is the unique solution to
\begin{equation} \label{eq_SL_positive_imaginary_invert_proof_1}
\left\{
\begin{aligned}
&(\mathcal{L}^{E}-\lambda)u=0 \quad \text{in $\Omega_{R}$,} \\
&\gamma^{+} u=\phi\in H^{\frac{1}{2}}(\Gamma) \quad \text{on $\Gamma$,}
\end{aligned}
\right.
\end{equation}
and $u_{-}$ solves
\begin{equation} \label{eq_SL_positive_imaginary_invert_proof_3}
\left\{
\begin{aligned}
&(\mathcal{L}^{E}-\lambda)u=0 \quad \text{in $\Omega_{L}$,} \\
&\gamma^{-} u=\phi\in H^{\frac{1}{2}}(\Gamma) \quad \text{on $\Gamma$.}
\end{aligned}
\right.
\end{equation}
We claim that
\begin{equation} \label{eq_SL_positive_imaginary_invert_proof_2}
\big(\Pi_{-}(\lambda)-\Pi_{+}(\lambda)\big)\tilde{S}(\lambda)
=\tilde{S}(\lambda)\big(\Pi_{-}(\lambda)-\Pi_{+}(\lambda)\big)
=\mathbbm{1}.
\end{equation}
The first identity is clear. In fact, we have
\begin{equation*}
\big(\Pi_{-}(\lambda)-\Pi_{+}(\lambda)\big)\tilde{S}(\lambda)[\varphi]
=\partial_{\nu,a^{E}}^{-}\tilde{S}(\lambda)[\varphi]-\partial_{\nu,a^{E}}^{+}\tilde{S}(\lambda)[\varphi]
=\varphi ,
\end{equation*}
where the last equality follows from the jumping identity in Proposition \ref{prop_boundary_operators_analytic_bounded_jump}. On the other hand, for any $\phi\in H^{\frac{1}{2}}(\Gamma)$, the functions $u_{\pm}$ defined via \eqref{eq_SL_positive_imaginary_invert_proof_1} and \eqref{eq_SL_positive_imaginary_invert_proof_3} admit the following layer-potential expression, which is proved by the Green identity as in Proposition \ref{prop_analytic_extended_SL_DL_bound}, 
\begin{equation*}
\begin{aligned}
&u_{+}=\mathcal{D}(\lambda;\tilde{G}^{E,out})[\phi]-\mathcal{S}(\lambda;\tilde{G}^{E,out})[\Pi_{+}\phi]\quad \text{in $\Omega_{R}$,} \\
&u_{-}=-\mathcal{D}(\lambda;\tilde{G}^{E,out})[\phi]+\mathcal{S}(\lambda;\tilde{G}^{E,out})[\Pi_{-}\phi]\quad \text{in $\Omega_{L}$.}
\end{aligned}
\end{equation*}
Taking the trace on $\Gamma$ and applying the jump identities, we arrive at
\begin{equation*}
\begin{aligned}
&\phi=\big(\frac{1}{2}+K(\lambda;\tilde{G}^{E,out})\big)[\phi]-\tilde{S}(\lambda)[\Pi_{+}\phi], \\
&\phi=\big(\frac{1}{2}-K(\lambda;\tilde{G}^{E,out})\big)[\phi]+\tilde{S}(\lambda)[\Pi_{-}\phi] .
\end{aligned}
\end{equation*}
Adding these two identities concludes the proof of \eqref{eq_SL_positive_imaginary_invert_proof_2}.

\subsubsection{Fredholmness at Real Frequencies}

The proof of Proposition \ref{prop_SL_real_fredholm} is more involved. As we have remarked after Theorem \ref{thm_analytic_Fredholm_SL}, the key idea lies in the fact that the out-going Green function $G^{E,out}(\bm{x},\bm{y};\lambda)$ turns to the Green function $G^{\pm}(\bm{x},\bm{y};\lambda)$, which are associated with the bulk operator $\mathcal{L}^{a}\pm\delta\mathcal{L}^{b}$ and are well-posed thanks to the gap condition, as $\bm{x},\bm{y}\to\pm\infty$ along the auxiliary interface $\Gamma$.

We first fix the notation. For $\lambda\in\mathcal{I}_{0}$, let $G^{\pm}(\bm{x},\bm{y};\lambda)$ be the integral kernel of the resolvent $\mathcal{G}^{\pm}(\lambda):=\big(\mathcal{L}^{\pm}-\lambda\big)^{-1}$, where $\mathcal{L}^{\pm}:=\mathcal{L}^{a}\pm\delta\mathcal{L}^{b}$ with the coefficient function $a^{\pm}(\bm{x}):=a(\bm{x})\pm \delta b(\bm{x})$. Let $S^{\pm}(\lambda)$ be the trace of SL operators associated with $G^{\pm}(\bm{x},\bm{y};\lambda)$, i.e.,
\begin{equation*}
S^{\pm}(\lambda):H^{-\frac{1}{2}}(\Gamma)\to H^{\frac{1}{2}}(\Gamma),\quad 
S^{\pm}(\lambda)=\gamma \mathcal{G}^{\pm}(\lambda)J_{\Gamma}.
\end{equation*}
As the spectral parameter $\lambda$ is fixed throughout this section, we omit the $\lambda$-dependence and write $S=S(\lambda),S^{\pm}=S^{\pm}(\lambda)$ if no confusion arises.

We aim to show that the operator $S$ admits a left and right pseudo-inverse, respectively.
\begin{proposition} \label{prop_pesudo_inverse}
There exist $Q,W\in\mathcal{B}(H^{\frac{1}{2}}(\Gamma),H^{-\frac{1}{2}}(\Gamma))$ such that
\begin{equation} \label{eq_real_fredholmness_proof_1}
QS=\mathbbm{1}_{H^{-\frac{1}{2}}(\Gamma)}+E_1+T_1,\quad SW=\mathbbm{1}_{H^{\frac{1}{2}}(\Gamma)}+E_2+T_2,
\end{equation}
where $E_1,T_1\in \mathcal{B}(H^{-\frac{1}{2}}(\Gamma))$, $E_2,T_2\in \mathcal{B}(H^{\frac{1}{2}}(\Gamma))$ with $T_i$ being compact and
\begin{equation*}
\|E_1\|_{\mathcal{B}(H^{-\frac{1}{2}}(\Gamma))}<1,\quad \|E_2\|_{\mathcal{B}(H^{\frac{1}{2}}(\Gamma))}<1.
\end{equation*}
\end{proposition}
By Proposition \ref{prop_pesudo_inverse}, we know that $S=S(\lambda)=\tilde{S}(\lambda)$ is a Fredholm operator for all $\lambda\in\mathcal{N}\cap \mathbb{R}$. On the other hand, Proposition \ref{prop_SL_positive_imaginary_invertibility} states that the index of $\tilde{S}(\lambda+i\epsilon)$ is zero. Hence, by the continuity of the index, we conclude that the index of $S$ is also zero, which completes the proof of Proposition \ref{prop_SL_real_fredholm}.

For Proposition \ref{prop_pesudo_inverse}, we only construct the left pseudo-inverse $Q$, and the right inverse is constructed similarly. This is achieved in several steps. To prepare, for $n\geq 1$, we select functions $\alpha_{n}^{+},\psi_{n}^{+},\theta_{n}^{+},\beta_{n}^{+}\in C^{\infty}(\Gamma)$ which satisfy
\begin{itemize}
    \item[(i)] $f_{n}^{+}(\bm{x})\in [0,1],\quad \forall \bm{x}\in\Gamma ,\quad f\in\{\alpha,\psi,\theta,\beta\}$;
    \item[(ii)] \begin{equation*}
\begin{aligned}
&\alpha_{n}^{+}(\bm{x})=1\,\, \text{for $\bm{x}\cdot(\bm{e}_2+2\bm{e}_1)>2n $}\quad\text{and}\quad \alpha_{n}^{+}(\bm{x})=0\,\, \text{for $\bm{x}\cdot(\bm{e}_2+2\bm{e}_1)<n $,} \\
&\psi_{n}^{+}(\bm{x})=1\,\, \text{for $\bm{x}\cdot(\bm{e}_2+2\bm{e}_1)>4n $}\quad\text{and}\quad \psi_{n}^{+}(\bm{x})=0\,\, \text{for $\bm{x}\cdot(\bm{e}_2+2\bm{e}_1)<3n $,} \\
&\theta_{n}^{+}(\bm{x})=1\,\, \text{for $\bm{x}\cdot(\bm{e}_2+2\bm{e}_1)>6n $}\quad\text{and}\quad \theta_{n}^{+}(\bm{x})=0\,\, \text{for $\bm{x}\cdot(\bm{e}_2+2\bm{e}_1)<5n $,} \\
&\beta_{n}^{+}(\bm{x})=1\,\, \text{for $\bm{x}\cdot(\bm{e}_2+2\bm{e}_1)>8n $}\quad\text{and}\quad \beta_{n}^{+}(\bm{x})=0\,\, \text{for $\bm{x}\cdot(\bm{e}_2+2\bm{e}_1)<7n $;} \\
\end{aligned}
\end{equation*}
    \item[(iii)] $\sup_{n\geq 1}\|f_{n}^{+}\|_{C^1(\Gamma)}<\infty,\quad f\in\{\alpha,\psi,\theta,\beta\}$.
\end{itemize}
In other words, the supports of these functions move to infinity as $n\to\infty$, and the transition areas of each function are separated by a distance greater than $n$; see Figure \ref{fig_real_fredholm} for a sketch of their profiles. Then we define
\begin{equation*}
f_{n}^{-}(\bm{x}):=f_{n}^{+}(-\bm{x}),\quad f_{n}^{0}(\bm{x}):=1-f_{n}^{+}(\bm{x})-f_{n}^{-}(\bm{x}),\quad f\in\{\alpha,\psi,\theta,\beta\}.
\end{equation*}
We will denote the multiplication operator induced by any function $f$ on $\Gamma$ as $M_{f}$.

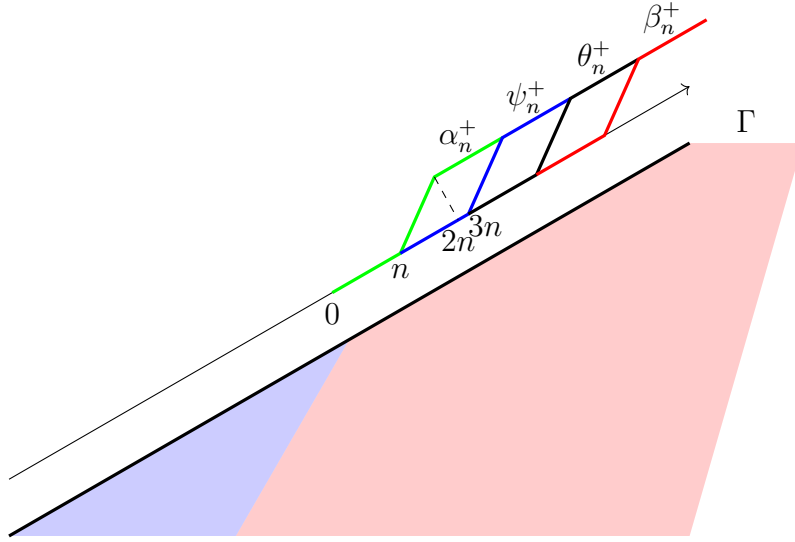
\begin{figure}
\centering
\begin{tikzpicture}[scale=1.5]
\tikzset{
    evan/.style   = {very thick, dashed, -{Stealth[length=3.2mm,width=2.8mm]}, draw=gray!70},
    prop/.style   = {decorate, decoration={snake, amplitude=1.3pt, segment length=6pt}},
  }
%bulks
\path [fill=red,opacity=0.2] ({1},{0})--({0},{-1*sqrt(3)})--({4},{-1*sqrt(3)})--({5},{sqrt(3)})--({4},{sqrt(3)})--({1},{0});
\path [fill=blue,opacity=0.2] ({1},{0})--({-2},{-1*sqrt(3)})--({0},{-1*sqrt(3)})--({1},{0});
%imaginary matching interface
\draw[very thick,black] ({-2},{-1*sqrt(3)})--({4},{sqrt(3)});
\node[above,scale=1] at ({4.5},{sqrt(3)}) {$\Gamma$};
%axis
\draw[->,black] ({-2},{-1*sqrt(3)+0.5})--({4},{sqrt(3)+0.5});
%alpha^+
\draw[very thick, green] ({1-0.05*3},{0.5-0.05*sqrt(3)})--({1+0.15*3},{0.5+0.15*sqrt(3)})--({1+0.25*3},{1+0.25*sqrt(3)})--({1+0.45*3},{1+0.45*sqrt(3)});
\node[above,scale=1] at ({1+0.32*3},{1+0.32*sqrt(3)}) {$\alpha_{n}^{+}$};
%\psi^+
\draw[very thick, blue] ({1+0.15*3},{0.5+0.15*sqrt(3)})--({1+0.35*3},{0.5+0.35*sqrt(3)})--({1+0.45*3},{1+0.45*sqrt(3)})--({1+0.65*3},{1+0.65*sqrt(3)});
\node[above,scale=1] at ({1+0.52*3},{1+0.52*sqrt(3)}) {$\psi_{n}^{+}$};
%\theta^+
\draw[very thick, black] ({1+0.35*3},{0.5+0.35*sqrt(3)})--({1+0.55*3},{0.5+0.55*sqrt(3)})--({1+0.65*3},{1+0.65*sqrt(3)})--({1+0.85*3},{1+0.85*sqrt(3)});
\node[above,scale=1] at ({1+0.72*3},{1+0.72*sqrt(3)}) {$\theta_{n}^{+}$};
%beta^+
\draw[very thick, red] ({1+0.55*3},{0.5+0.55*sqrt(3)})--({1+0.75*3},{0.5+0.75*sqrt(3)})--({1+0.85*3},{1+0.85*sqrt(3)})--({1+1.05*3},{1+1.05*sqrt(3)});
\node[above,scale=1] at ({1+0.92*3},{1+0.92*sqrt(3)}) {$\beta_{n}^{+}$};
%buffer area
\node[below,scale=1] at ({1-0.05*3},{0.5-0.05*sqrt(3)}) {$0$};
\node[below,scale=1] at ({1+0.15*3},{0.5+0.15*sqrt(3)}) {$n$};
\draw[dashed] ({1+0.25*3},{1+0.25*sqrt(3)})--({1+0.32*3},{0.5+0.32*sqrt(3)});
\node[below,scale=1] at ({1+0.32*3},{0.5+0.32*sqrt(3)}) {$2n$};
\node[below,scale=1] at ({1+0.4*3},{0.5+0.4*sqrt(3)}) {$3n$};
\end{tikzpicture}
\caption{Profile of the auxiliary functions. The domain in red (or blue) refers the medium with coefficient function $a^{+}=a+\delta b$ (or $a^{-}=a-\delta b$). }
\label{fig_real_fredholm}
\end{figure}

To start the construction of the left inverse, we first note that, as we will prove, the half-space Dirichlet problem associated with $\mathcal{L}^{a}\pm\delta\mathcal{L}^{b}$ is well-posed. Hence, the operator $S^{\pm}$ is invertible, as proved in Section \ref{sec_nonreal_invertible}, and, as will be seen, its inverse exhibits exponential off-diagonal decay as does $S^{\pm}$.
\begin{lemma} \label{lem_bulk_SL_inverse}
$S^{\pm}$ has a bounded inverse $Q^{\pm}$. Moreover, there exists $c,C,d_0>0$ such that for any disjoint open intervals $U_1,U_2\subset \Gamma$ with $d:=\text{dist}(U_1,U_2)>d_0$,
\begin{equation} \label{eq_bulk_SL_inverse_off_diag_estimate}
\big\|M_{\zeta_1}Q^{\pm}M_{\zeta_2} \big\|_{\mathcal{B}(H^{\frac{1}{2}}(\Gamma),H^{-\frac{1}{2}}(\Gamma))}\leq Ce^{-cd}\|\zeta_1\|_{C^1(\Gamma)}\|\zeta_2\|_{C^1(\Gamma)}
\end{equation}
for all $\zeta_i\in C^1(\Gamma)$ supported in $U_i$ ($i=1,2$).
\end{lemma}
The proof is postponed to the end of this section. In fact, we will use $S^{\pm}$ as the far-field component (i.e., at the ends of $\Gamma$) of the left inverse $Q$. For the middle part, it is well-known that the restriction of the SL potential associated with a  second-order elliptic operator is Fredholm. Moreover, by truncating the off-diagonals, the pseudo-inverse of the restricted single-layer potential is supported in a fixed neighborhood of the diagonals. The following result holds. 
\begin{lemma} \label{lem_middle_SL_pseudo_inverse}
For any $n>0$, there exists $Q_{n}^{0}\in \mathcal{B}(H^{\frac{1}{2}}(\Gamma),H^{-\frac{1}{2}}(\Gamma))$ such that 
\begin{equation*}
Q_{n}^{0}\big(M_{\beta_{n}^{0}}SM_{\beta_{n}^{0}} \big)=M_{\theta_{n}^{0}}+T_{n}^{0}
\end{equation*}
with $T_{n}^{0}$ being compact. Moreover, there exists $\rho>0$ such that for any disjoint open intervals $U_1,U_2\subset \Gamma$ with $\text{dist}(U_1,U_2)>\rho$,
\begin{equation*}
M_{\zeta_1}Q^{0}_{n}M_{\zeta_2}=0,
\end{equation*}
for all $\zeta_i\in C^1(\Gamma)$ supported in $U_i$ ($i=1,2$).
\end{lemma}
The proof is based on a standard parametrix construction based on symbol calculus. As this technique is standard and somewhat removed from the main line of this paper, we attach the proof in the appendix. With these preparations, we define
\begin{equation} \label{eq_real_fredholmness_proof_2}
Q_{n}:=M_{\theta_{n}^{+}}Q^{+}M_{\psi_{n}^{+}}+M_{\theta_{n}^{-}}Q^{-}M_{\psi_{n}^{-}}+M_{\theta_{n}^{0}}Q^{0}_{n}M_{\beta_{n}^{0}}.
\end{equation}
Let us calculate $Q_n S$, by which we will conclude the proof of \eqref{eq_real_fredholmness_proof_1}. First,
\begin{equation} \label{eq_real_fredholmness_proof_3}
\begin{aligned}
M_{\theta_{n}^{+}}Q^{+}M_{\psi_{n}^{+}}S
&=M_{\theta_{n}^{+}}Q^{+}M_{\psi_{n}^{+}}S^{+}+M_{\theta_{n}^{+}}Q^{+}M_{\psi_{n}^{+}}(S-S^{+}) \\
&=M_{\theta_{n}^{+}}Q^{+}S^{+}-M_{\theta_{n}^{+}}Q^{+}M_{1-\psi_{n}^{+}}S^{+} \\
&\quad +M_{\theta_{n}^{+}}Q^{+}M_{\psi_{n}^{+}}(S-S^{+})M_{\alpha_{n}^{+}}
+M_{\theta_{n}^{+}}Q^{+}M_{\psi_{n}^{+}}(S-S^{+})M_{1-\alpha_{n}^{+}} \\
&=M_{\theta_{n}^{+}}-M_{\theta_{n}^{+}}Q^{+}M_{1-\psi_{n}^{+}}S^{+} \\
&\quad +M_{\theta_{n}^{+}}Q^{+}M_{\psi_{n}^{+}}(S-S^{+})M_{\alpha_{n}^{+}}
+M_{\theta_{n}^{+}}Q^{+}M_{\psi_{n}^{+}}(S-S^{+})M_{1-\alpha_{n}^{+}} \\
&=:M_{\theta_{n}^{+}}+E_{n}^{+}.
\end{aligned}
\end{equation}
The three parts that constitute the remainder $E_{n}^{+}$ are small in norm, as $n\to\infty$, for the following reasons:
\begin{itemize}
    \item[(i)] $M_{\theta_{n}^{+}}Q^{+}M_{1-\psi_{n}^{+}}S^{+}$: the operator $M_{\theta_{n}^{+}}Q^{+}M_{1-\psi_{n}^{+}}$ converges to zero because the supports of $\theta_{n}^{+}$ and $1-\psi_{n}^{+}$ are separated by a distance $\geq n$, and thanks to the estimate \eqref{eq_bulk_SL_inverse_off_diag_estimate};
    \item[(ii)] $M_{\theta_{n}^{+}}Q^{+}M_{\psi_{n}^{+}}(S-S^{+})M_{\alpha_{n}^{+}}$: as we have indicated at the beginning of this section, the out-going Green function $G^{E,out}(\bm{x},\bm{y};\lambda)$ converges to $G^{\pm}(\bm{x},\bm{y};\lambda)$ as $\bm{x},\bm{y}\to\pm\infty$ along $\Gamma$. This implies the convergence of the single-layer operator $M_{\psi_{n}^{+}}(S-S^{+})M_{\alpha_{n}^{+}} \to 0$, because the support of $\psi_{n}^{+}$ and $\alpha_{n}^{+}$ moves to infinity as $n\to\infty$;
    \item[(iii)] $M_{\theta_{n}^{+}}Q^{+}M_{\psi_{n}^{+}}(S-S^{+})M_{1-\alpha_{n}^{+}}$: the operator $M_{\psi_{n}^{+}}(S-S^{+})M_{1-\alpha_{n}^{+}}$ converges to zero for a similar reason as in (i).
\end{itemize}
To be precise, we will prove, at the end of this section, that the following convergence results hold. 
\begin{lemma} \label{lem_remainder_pseudo_inverse_positive}
As $n\to\infty$, it holds that
\begin{equation} \label{eq_remainder_pseudo_inverse_positive_1}
\big\|M_{\theta_{n}^{+}}Q^{+}M_{1-\psi_{n}^{+}}S^{+}\big\|_{\mathcal{B}(H^{-\frac{1}{2}}(\Gamma),H^{-\frac{1}{2}}(\Gamma))}\to 0 ,
\end{equation}
\begin{equation} \label{eq_remainder_pseudo_inverse_positive_2}
\big\|M_{\theta_{n}^{+}}Q^{+}M_{\psi_{n}^{+}}(S-S^{+})M_{\alpha_{n}^{+}}\big\|_{\mathcal{B}(H^{-\frac{1}{2}}(\Gamma),H^{-\frac{1}{2}}(\Gamma))}\to 0 ,
\end{equation}
and
\begin{equation} \label{eq_remainder_pseudo_inverse_positive_3}
\big\|M_{\theta_{n}^{+}}Q^{+}M_{\psi_{n}^{+}}(S-S^{+})M_{1-\alpha_{n}^{+}}\big\|_{\mathcal{B}(H^{-\frac{1}{2}}(\Gamma),H^{-\frac{1}{2}}(\Gamma))}\to 0 .
\end{equation}
\end{lemma}
Similarly, one can calculate
\begin{equation} \label{eq_real_fredholmness_proof_4}
\begin{aligned}
M_{\theta_{n}^{-}}Q^{-}M_{\psi_{n}^{-}}S
&=M_{\theta_{n}^{-}}-M_{\theta_{n}^{-}}Q^{-}M_{1-\psi_{n}^{-}}S^{-} \\
&\quad +M_{\theta_{n}^{-}}Q^{-}M_{\psi_{n}^{-}}(S-S^{-})M_{\alpha_{n}^{-}}
+M_{\theta_{n}^{-}}Q^{-}M_{\psi_{n}^{-}}(S-S^{-})M_{1-\alpha_{n}^{-}} \\
&=:M_{\theta_{n}^{-}}+E_{n}^{-},
\end{aligned}
\end{equation}
and prove that the following result holds. 
\begin{lemma} \label{lem_remainder_pseudo_inverse_negative}
As $n\to\infty$, it holds that
\begin{equation*}
\big\|M_{\theta_{n}^{-}}Q^{-}M_{1-\psi_{n}^{-}}S^{-}\big\|_{\mathcal{B}(H^{-\frac{1}{2}}(\Gamma),H^{-\frac{1}{2}}(\Gamma))}\to 0 ,
\end{equation*}
\begin{equation*}
\big\|M_{\theta_{n}^{-}}Q^{-}M_{\psi_{n}^{-}}(S-S^{-})M_{\alpha_{n}^{-}}\big\|_{\mathcal{B}(H^{-\frac{1}{2}}(\Gamma),H^{-\frac{1}{2}}(\Gamma))}\to 0 ,
\end{equation*}
and
\begin{equation*}
\big\|M_{\theta_{n}^{-}}Q^{-}M_{\psi_{n}^{-}}(S-S^{-})M_{1-\alpha_{n}^{-}}\big\|_{\mathcal{B}(H^{-\frac{1}{2}}(\Gamma),H^{-\frac{1}{2}}(\Gamma))}\to 0 .
\end{equation*}
\end{lemma}
Finally, let us compute the $M_{\theta_{n}^{0}}$ part of \eqref{eq_real_fredholmness_proof_2}. We have
\begin{equation} \label{eq_real_fredholmness_proof_5}
\begin{aligned}
M_{\theta_{n}^{0}}Q^{0}_{n}M_{\beta_{n}^{0}}S
&=M_{\theta_{n}^{0}}Q^{0}_{n}M_{\beta_{n}^{0}}SM_{\beta_{n}^{0}}+M_{\theta_{n}^{0}}Q^{0}_{n}M_{\beta_{n}^{0}}SM_{1-\beta_{n}^{0}} \\
&\overset{(i)}{=}M_{\theta_{n}^{0}}M_{\beta_{n}^{0}}+M_{\theta_{n}^{0}}T_{n}^{0} \\
&\quad + M_{\theta_{n}^{0}}Q^{0}_{n}M_{\beta_{n}^{0}\cdot (1-\tau_{n}^{0})}SM_{1-\beta_{n}^{0}} \\
&\quad + M_{\theta_{n}^{0}}Q^{0}_{n}M_{\beta_{n}^{0}\cdot \tau_{n}^{0} }SM_{1-\beta_{n}^{0}} \\
&\overset{(ii)}{=:}M_{\theta_{n}^{0}}+M_{\theta_{n}^{0}}T_{n}^{0} + E_{n}^{0}.
\end{aligned}
\end{equation}
Here, $\tau_{n}^{0}\in C^{\infty}(\Gamma)$ satisfies
\begin{equation*}
\tau_{n}^{0}(\bm{x})=1\,\, \text{for $|\bm{x}\cdot(\bm{e}_2+2\bm{e}_1)|\leq 6.4n $}\quad\text{and}\quad \tau_{n}^{0}(\bm{x})=0\,\, \text{for $|\bm{x}\cdot(\bm{e}_2+2\bm{e}_1)|>6.6n $.}
\end{equation*}
The identity (i) follows from Lemma \ref{lem_middle_SL_pseudo_inverse} and (ii) is derived by noting that $\beta_{n}^{0}\equiv 1$ on the support of $\theta_{n}^{0}$. Note that since the supports of $\beta_{n}^{0}\cdot (1-\tau_{n}^{0})$ and $\theta_{n}^{0}$ are separated by a distance $\geq 0.4n$, we have
\begin{equation*}
M_{\theta_{n}^{0}}Q^{0}_{n}M_{\beta_{n}^{0}\cdot (1-\tau_{n}^{0})}SM_{1-\beta_{n}^{0}}=0
\end{equation*}
for sufficiently large $n$. On the other hand, arguing similarly to \eqref{eq_remainder_pseudo_inverse_positive_1}, the disjointness between $\beta_{n}^{0}\cdot \tau_{n}^{0}$ and $1-\beta_{n}^{0}$ implies that
\begin{equation*}
\big\|M_{\theta_{n}^{0}}Q^{0}_{n}M_{\beta_{n}^{0}\cdot \tau_{n}^{0}}SM_{1-\beta_{n}^{0}}\big\|_{\mathcal{B}(H^{-\frac{1}{2}}(\Gamma),H^{-\frac{1}{2}}(\Gamma))}\to 0 .
\end{equation*}
In conclusion, by \eqref{eq_real_fredholmness_proof_3}, \eqref{eq_real_fredholmness_proof_4}, and \eqref{eq_real_fredholmness_proof_5}, we conclude that
\begin{equation*}
\begin{aligned}
Q_n S=M_{\theta_{n}^{+}}+E_{n}^{+}+M_{\theta_{n}^{-}}+E_{n}^{-}+M_{\theta_{n}^{0}}+M_{\theta_{n}^{0}}T_{n}^{0}+E_{n}^{0}=:\mathbbm{1}+E_n+T_n
\end{aligned}
\end{equation*}
with $E_n:=E_{n}^{+}+E_{n}^{-}+E_{n}^{0}$ converging to zero in operator norm as $n\to\infty$, and $T_n:=M_{\theta_{n}^{0}}T_{n}^{0}$ being compact. This completes the proof of \eqref{eq_real_fredholmness_proof_1}.

\begin{proof}[Proof of Lemma \ref{lem_bulk_SL_inverse}]
We only prove the statement regarding the inverse $Q^+$, while the other case is treated similarly. First we note that the following half-space Dirichlet problem is well-posed
\begin{equation} \label{eq_bulk_SL_inverse_1}
\left\{
\begin{aligned}
&(\mathcal{L}^{+}-\lambda)u=0 \quad \text{in $\Omega_{\sigma}$,} \\
&\gamma^{\sigma^{\prime}} u=\phi\in H^{\frac{1}{2}}(\Gamma) \quad \text{on $\Gamma$,}
\end{aligned}
\right.
\end{equation}
for both $(\sigma,\sigma^{\prime})=(R,+)$ and $(\sigma,\sigma^{\prime})=(L,-)$ (i.e., the right/left half plane). In fact, if it is not well-posed for $\Omega_{R}$, we know that $\mathcal{I}\cap\text{Spec}\big(\mathcal{L}^{+}\big|_{H_{0}^{1}(\Omega_{R})}\big)\neq \emptyset$. However, we note that the coefficient of $\mathcal{L}^{+}$ is reflectional symmetric about $\Gamma$ because
\begin{equation*}
    a(F_{\Gamma}\bm{x})=a(\bm{x}),\quad b(F_{\Gamma}\bm{x})=b(\bm{x}),\quad \forall \bm{x}\in\mathbb{R}^2,
\end{equation*}
which is derived using Assumptions \ref{asmp_unperturbed_coef}, \ref{asmp_perturbed_coef}, and the identity $F_{\Gamma}=F\cdot R$. Thus, using the odd extension, we conclude that $\mathcal{I}\cap\text{Spec}(\mathcal{L}^{+})\neq \emptyset$, which clearly contradicts the gap condition, i.e., Proposition \ref{prop_gap_open}. The same argument works to conclude that \eqref{eq_bulk_SL_inverse_1} is well-posed for the left half-plane. Hence, as in Section \ref{sec_nonreal_invertible}, the following Dirichlet-to-Neumann maps are well-defined
\begin{equation*}
\Pi_{\pm}:H^{\frac{1}{2}}(\Gamma)\to H^{-\frac{1}{2}}(\Gamma),\quad
\Pi_{R/L}^{+}\phi=\partial_{\nu,a^+}^{\pm}u_{\pm},
\end{equation*}
where $u_{+/-}$ is the unique solution to \eqref{eq_bulk_SL_inverse_1} with $(\sigma,\sigma^{\prime})=(R,+)$ or $(\sigma,\sigma^{\prime})=(L,-)$. Moreover, if we define
\begin{equation}
Q^{+}:=\Pi_{-}-\Pi_{+} ,
\end{equation}
then the argument is Section \ref{sec_nonreal_invertible} applies to show that $Q^{+}=(S^+)^{-1}$. Hence, to conclude the proof, it is sufficient to show that the Dirichlet-to-Neumann map satisfies the decay estimate \eqref{eq_bulk_SL_inverse_off_diag_estimate}. We only prove it for $\Pi_{R}^{+}$, i.e., if $d:=\text{dist}(U_1,U_2)>2$,
\begin{equation} \label{eq_bulk_SL_inverse_2}
\big\|M_{\zeta_1}\Pi_{+}M_{\zeta_2} \big\|_{\mathcal{B}(H^{\frac{1}{2}}(\Gamma),H^{-\frac{1}{2}}(\Gamma))}\leq Ce^{-cd}\|\zeta_1\|_{C^1(\Gamma)}\|\zeta_2\|_{C^1(\Gamma)}.
\end{equation}
The proof for $\Pi_{-}$ is similar. The proof of \eqref{eq_bulk_SL_inverse_2} follows the standard procedure by writing $\Pi_{+}$ in terms of the resolvent. Let $\Xi_{+}:H^{\frac{1}{2}}(\Gamma)\to H^{1}(\Omega_R)$ be an extension map, which satisfies
\begin{equation*}
\gamma^{+}\Xi_{+}=\mathbbm{1}_{H^{\frac{1}{2}}(\Gamma)},
\end{equation*}
and
\begin{equation*}
\text{supp}(\Xi_{+}\phi)\subset U^{\prime} \quad \text{with}\quad U^{\prime}:=\{\bm{x}\in \Omega_R:\, \text{dist}(\bm{x},U)<\frac{1}{2}\},
\end{equation*}
for any $\phi\in H^{\frac{1}{2}}(\Gamma)$ supported in $U$. Let $u\in H^{1}(\Omega_{R})$ be the unique function that solves \eqref{eq_bulk_SL_inverse_1} with the trace $\gamma^{+}u=\phi\zeta_2$. Then, defining the function
\begin{equation*}
w:=u-\Xi_{+}M_{\zeta_2}\phi ,
\end{equation*}
one sees that $w$ solves the following equation with the $H^{-1}$ source $f_{\zeta_2,\phi}:=-\nabla\cdot a^{+}\nabla (\Xi_{+}M_{\zeta_2}\phi)$ compactly supported in $U_2^{\prime}$:
\begin{equation*}
\left\{
\begin{aligned}
&(\mathcal{L}^{+}-\lambda)w=f_{\zeta_2,\phi} \quad \text{in $\Omega_{R}$,} \\
&\gamma^{+} u=0 \quad \text{on $\Gamma$.}
\end{aligned}
\right.
\end{equation*}
This implies that
\begin{equation*}
u=\Xi_{+}M_{\zeta_2}\phi  +   \big(\mathcal{L}^{+}\big|_{H_{0}^{1}(\Omega_{R})}-\lambda\big)^{-1}f_{\zeta_2,\phi}.
\end{equation*}
As the support of $\Xi^{R}M_{\zeta_2}\phi$ is disjoint from $U_1$ by a distance $\geq \frac{3}{2}$, only the resolvent part contributes when evaluating the conormal derivative of $u$ on $\text{supp}(U_1)$. This implies that
\begin{equation*}
M_{\zeta_1}\Pi_{+}M_{\zeta_2}\phi
=\zeta_1\partial_{\nu,a^+}^{+}\big(\mathcal{L}^{+}\big|_{H_{0}^{1}(\Omega_{R})}-\lambda\big)^{-1}f_{\zeta_2,\phi}. 
\end{equation*}
Then, the proof of \eqref{eq_bulk_SL_inverse_2} follows from the Combes-Thomas estimate by applying it to the resolvent $\big(\mathcal{L}^{+}\big|_{H_{0}^{1}(\Omega_{R})}-\lambda\big)^{-1}$, and noting that
\begin{equation*}
\text{dist}(\text{supp}(\zeta_1),\text{supp}(f_{\zeta_2,\phi}))>1 .
\end{equation*}
The details are left to the reader.
\end{proof}

\begin{proof}[Proof of Lemma \ref{lem_remainder_pseudo_inverse_positive}]
As remarked below \eqref{eq_real_fredholmness_proof_3}, estimate \eqref{eq_remainder_pseudo_inverse_positive_1} follows directly from Lemma \ref{lem_bulk_SL_inverse}. On the other hand, \eqref{eq_remainder_pseudo_inverse_positive_3} is derived by a similar reasoning, noting that the operator $S$ exhibits exponential off-diagonal decay (thanks to Proposition \ref{prop_perpendicular_decay_Green_function}) and so does $S^{+}$ (directly by the Combes-Thomas estimate applied to $\big(\mathcal{L}^{+}-\lambda\big)^{-1}$).

To prove \eqref{eq_remainder_pseudo_inverse_positive_2}, we note that the integral kernel of $S-S^{+}$ admits the following expression derived by the resolvent identity:
\begin{equation*}
\begin{aligned}
(S-S^{+})(\bm{x},\bm{y})&\overset{(i)}{=}\int_{\bm{z}\in \mathbb{R}^2}G^{E,out}(\bm{x},\bm{z};\lambda)\big(-\nabla_{z}\cdot (a^{+}(\bm{z})-a^{E}(\bm{z}))\nabla_{\bm{z}}\big)G^{+}(\bm{z},\bm{y};\lambda) d\bm{z} \\
&\overset{(ii)}{=}\int_{\bm{z}\in \mathbb{R}^2}\nabla_{z}G^{E,out}(\bm{x},\bm{z};\lambda)\cdot (a^{+}(\bm{z})-a^{E}(\bm{z}))\nabla_{\bm{z}}G^{+}(\bm{z},\bm{y};\lambda) d\bm{z}.
\end{aligned}
\end{equation*}
Here, when applying the resolvent identity to derive (i), one starts with
\begin{equation*}
G^{E}(\bm{z},\bm{y};\lambda+i\epsilon)-G^{+}(\bm{z},\bm{y};\lambda+i\epsilon)=
\int_{\bm{z}\in \mathbb{R}^2}G^{E}(\bm{x},\bm{z};\lambda+i\epsilon)\big(-\nabla_{z}\cdot (a^{+}(\bm{z})-a^{E}(\bm{z}))\nabla_{\bm{z}}\big)G^{+}(\bm{z},\bm{y};\lambda+i\epsilon) d\bm{z},
\end{equation*}
then send $\epsilon\to 0^+$. The equality (ii) follows by applying the integral by parts.\footnote{This is justified using the local regularity of the Green function; see, e.g., \cite[Theorem 1.6, Step 4]{qiu2025bec_finite}.} The key point is that $a^{+}(\bm{z})=a^{E}(\bm{z})$ for $z_2>0$; hence,
\begin{equation} \label{eq_remainder_pseudo_inverse_positive_proof_1}
\begin{aligned}
(S-S^{+})(\bm{x},\bm{y})=\int_{\{z_2<0\}}\nabla_{z}G^{E,out}(\bm{x},\bm{z};\lambda)\cdot (a^{+}(\bm{z})-a^{E}(\bm{z}))\nabla_{\bm{z}}G^{+}(\bm{z},\bm{y};\lambda) d\bm{z}.
\end{aligned}
\end{equation}
Note that, for $\bm{x}\in \text{supp}(\psi_{n}^{+})$ and $\bm{y}\in \text{supp}(\alpha_{n}^{+})$,
\begin{itemize}
    \item[(i)] $|\nabla_{\bm{x}}\nabla_{\bm{z}}G^{E,out}(\bm{x},\bm{z};\lambda)|$ decays exponentially by Proposition \ref{prop_perpendicular_decay_Green_function}, thanks to the fact that
    \begin{equation*}
    \text{dist}(\bm{x},\{z_2<0\})\geq n ;
    \end{equation*}
    \item[(ii)] Similarly, $|\nabla_{\bm{z}}G^{E,out}(\bm{z},\bm{y};\lambda)|$ also decays exponentially, by the Combes-Thomas estimate and the fact that $\text{dist}(\bm{y},\{z_2<0\})\geq n $.
\end{itemize}
As a consequence, upon taking first-order derivative in $\bm{x}$, it holds that
\begin{equation} \label{eq_remainder_pseudo_inverse_positive_proof_2}
\begin{aligned}
\big|\partial_{t,\bm{x}}^{(k)}(S-S^{+})(\bm{x},\bm{y})\big|
&\leq C\int_{\{z_2<0\}}e^{-c|\bm{x}-\bm{z}|}e^{-c|\bm{z}-\bm{y}|}d\bm{z} \\
&\overset{(i)}{\leq}Ce^{-\frac{c}{4}|\bm{x}-\bm{y}|}\int_{\{z_2<0\}}e^{-\frac{3c}{4}|\bm{x}-\bm{z}|}e^{-\frac{3c}{4}|\bm{z}-\bm{y}|}d\bm{z} \\
&\overset{(ii)}{\leq}Ce^{-\frac{c}{4}|\bm{x}-\bm{y}|}e^{-2\cdot\frac{c}{4}\frac{n}{2}}\int_{\{z_2<0\}}e^{-\frac{c}{2}|\bm{x}-\bm{z}|}e^{-\frac{c}{2}|\bm{z}-\bm{y}|}d\bm{z} \\
&\leq Ce^{-\frac{c}{4}|\bm{x}-\bm{y}|}e^{-\frac{cn}{4}}
\end{aligned}
\end{equation}
for some $C,c>0$ and $k=0,1$, where $\partial_{t}$ denotes the tangential derivative along $\Gamma$. Here, to derive (i) and (ii), we have applied the triangular inequality and the estimate
\begin{equation*}
    \text{dist}(\bm{x},\{z_2<0\})\geq n,\quad \text{for }\bm{x}\in \text{supp}(\psi_{n}^{+})\cup \text{supp}(\alpha_{n}^{+}) .
\end{equation*}
By Schur's test, \eqref{eq_remainder_pseudo_inverse_positive_proof_2} implies the bound
\begin{equation*}
\big\|M_{\psi_{n}^{+}}(S-S^{+})M_{\alpha_{n}^{+}}\big\|_{\mathcal{B}(L^2(\Gamma),H^{1}(\Gamma))} \to 0,\quad\text{as $n\to\infty$}.
\end{equation*}
Similarly, repeating \eqref{eq_remainder_pseudo_inverse_positive_proof_1}-\eqref{eq_remainder_pseudo_inverse_positive_proof_2} but taking the tangential derivative in $\bm{y}$, one can derive
\begin{equation*}
\big\|M_{\psi_{n}^{+}}(S-S^{+})M_{\alpha_{n}^{+}}\big\|_{\mathcal{B}(H^{-1}(\Gamma),L^2(\Gamma))} \to 0,\quad\text{as $n\to\infty$}.
\end{equation*}
Then a standard interpolation argument gives the operator bound in $\mathcal{B}(H^{-\frac{1}{2}},H^{\frac{1}{2}})$, which completes the proof of \eqref{eq_remainder_pseudo_inverse_positive_2}.
\end{proof}

\subsection{Projection Identities Associated with Straight-Interface Modes}

In this section, we prove the following identities.
\begin{proposition} \label{prop_projection_identity}
For $\lambda\in \mathcal{I}_{0}$, it holds that
\begin{equation} \label{eq_projection_identity_1}
\begin{aligned}
\begin{pmatrix} \gamma u_{-}^{E} \\ \partial_{\nu,a^{E}}u_{-}^{E} \end{pmatrix}
=\frac{1}{2} \begin{pmatrix}\gamma u_{-}^{E} \\ \partial_{\nu,a^{E}}u_{-}^{E} \end{pmatrix}
+\begin{pmatrix}
-K(\lambda;G^{E,out}) & S(\lambda;G^{E,out}) \\ -N(\lambda;G^{E,out}) & K^{*}(\lambda;G^{E,out})
\end{pmatrix}
\begin{pmatrix}\gamma u_{-}^{E} \\ \partial_{\nu,a^{E}}u_{-}^{E}\end{pmatrix},
\end{aligned}
\end{equation}
and
\begin{equation} \label{eq_projection_identity_2}
\begin{aligned}
\begin{pmatrix} \gamma u_{+}^{E} \\ \partial_{\nu,a^{E}}u_{+}^{E} \end{pmatrix}
=\frac{1}{2} \begin{pmatrix} \gamma u_{+}^{E} \\ \partial_{\nu,a^{E}}u_{+}^{E} \end{pmatrix}
+\begin{pmatrix}
K(\lambda;G^{E,out}) & -S(\lambda;G^{E,out}) \\ N(\lambda;G^{E,out}) & -K^{*}(\lambda;G^{E,out})
\end{pmatrix}
\begin{pmatrix} \gamma u_{+}^{E} \\ \partial_{\nu,a^{E}}u_{+}^{E} \end{pmatrix},
\end{aligned}
\end{equation}
where $u_{\pm}^{E}=u_{\kappa^{\pm}(\lambda)}^{E}$ are the straight-interface modes.
\end{proposition}
Before showing the proof, we briefly illustrate the meaning of Proposition \ref{prop_projection_identity}. Note that the operators on the right side of \eqref{eq_projection_identity_1}-\eqref{eq_projection_identity_2} take the form of the Calderón projectors associated with the right/left half-plane:
\begin{equation*}
\begin{aligned}
&P_{Cal}^{L}(\lambda):=\frac{1}{2}\mathbbm{1}+\begin{pmatrix}
K(\lambda;G^{E,out}) & -S(\lambda;G^{E,out}) \\ N(\lambda;G^{E,out}) & -K^{*}(\lambda;G^{E,out})
\end{pmatrix},\\
&P_{Cal}^{R}(\lambda):=\frac{1}{2}\mathbbm{1}+\begin{pmatrix}
-K(\lambda;G^{E,out}) & S(\lambda;G^{E,out}) \\ -N(\lambda;G^{E,out}) & K^{*}(\lambda;G^{E,out}) 
\end{pmatrix} .
\end{aligned}
\end{equation*}
It is well-known (for the case of bounded domains) that the range of the Calderón projector consists of boundary data of functions solving the elliptic equation in the respective domain; see \cite[Exercise 7.6]{mclean2000strongly}. For our case, since the domain is infinite, the Calderón projector is radiation-selective: the range of $P_{Cal}^{L}(\lambda)$ consists of functions in $\Omega_L$, solving $(\mathcal{L}^{E}-\lambda)u=0$ and the left-going radiation condition. As the left-going straight-interface mode $u_{-}^{E}$ satisfies all these conditions, the identity \eqref{eq_projection_identity_1} holds naturally and verifies that the understanding of Calderón projectors extends to infinite domains. The other identity \eqref{eq_projection_identity_2} is explained in a similar manner.

\begin{proof}[Proof of Proposition \ref{prop_projection_identity}]
We only prove \eqref{eq_projection_identity_1}, since \eqref{eq_projection_identity_2} can be treated in the same way. Applying the Green identity between $u_{-}^{E}$ and $G^{E,out}$ within the domain $\cup_{-N\leq n\leq -1}S_{n}$, as in \eqref{eq_analytic_extended_SL_DL_bound_proof_3} (but in the left half-plane), leads to
\begin{equation*}
\begin{aligned}
u_{-}^{E}(\bm{x})&=  \int_{\Gamma_{-N-1}}
\Big[ \partial_{\nu,a^{E},\bm{y}}G^{E,out}(\bm{x},\bm{y};\lambda)u_{-}^{E}(\bm{y})-G^{E,out}(\bm{x},\bm{y};\lambda)\partial_{\nu,a^{E}}u_{-}^{E}(\bm{y}) \Big] d\sigma(\bm{y}) \\
&\quad -\int_{\Gamma}
\Big[ \partial_{\nu,a^{E},\bm{y}}G^{E,out}(\bm{x},\bm{y};\lambda)u_{-}^{E}(\bm{y})-G^{E,out}(\bm{x},\bm{y};\lambda)\partial_{\nu,a^{E}}u_{-}^{E}(\bm{y}) \Big] d\sigma(\bm{y}) ,
\end{aligned}
\end{equation*}
for $\bm{x}\in \Omega_{L}$ and sufficiently large $N>0$. As $N\to\infty$, the first integral converges as follows by the asymptotics in Proposition \ref{prop_parallel_decay_Green_function}:
\begin{equation*}
\begin{aligned}
&\lim_{N\to\infty}\int_{\Gamma_{-N-1}}
\Big[ \partial_{\nu,a^{E},\bm{y}}G^{E,out}(\bm{x},\bm{y};\lambda)u_{-}^{E}(\bm{y})-G^{E,out}(\bm{x},\bm{y};\lambda)\partial_{\nu,a^{E}}u_{-}^{E}(\bm{y}) \Big] d\sigma(\bm{y}) \\
&=\frac{i}{|\partial_{\kappa}\lambda^{E}\big(\kappa^{+}(\lambda)\big)|}u^{E}_{\kappa^{+}(\lambda)}(\bm{x})\lim_{N\to\infty}\int_{\Gamma_{-N-1}}
\Big[ \overline{\partial_{\nu,a^{E}}u_{+}^{E}(\bm{y})}u_{-}^{E}(\bm{y})-\overline{u^{E}_{\kappa^{+}(\lambda)}(\bm{y})}\partial_{\nu,a^{E}}u_{-}^{E}(\bm{y}) \Big] d\sigma(\bm{y}).
\end{aligned}
\end{equation*}
Remark that the right side equals zero by Proposition \ref{prop_energy_flux}. Hence, we conclude that
\begin{equation*}
u_{-}^{E}(\bm{x})=\mathcal{S}(\lambda;G^{E,out})[\partial_{\nu,a^{E}}u_{-}^{E}](\bm{x})-\mathcal{D}(\lambda;G^{E,out})[\gamma u_{-}^{E}](\bm{x}),\quad \bm{x}\in\Omega_{L}.
\end{equation*}
Taking the trace and the conormal derivative to $\Gamma_{-}$ and applying the jump identities in Proposition \ref{prop_boundary_operators_analytic_bounded_jump}, we obtain
\begin{equation*}
\begin{aligned}
\begin{pmatrix} \gamma u_{-}^{E} \\ \partial_{\nu,a^{E}}u_{-}^{E} \end{pmatrix}
=\begin{pmatrix}
\frac{1}{2}-K(\lambda;G^{E,out}) & S(\lambda;G^{E,out}) \\ -N(\lambda;G^{E,out}) & \frac{1}{2}+K^{*}(\lambda;G^{E,out})
\end{pmatrix}
\begin{pmatrix} \gamma u_{-}^{E} \\ \partial_{\nu,a^{E}}u_{-}^{E} \end{pmatrix} ,
\end{aligned}
\end{equation*}
which proves \eqref{eq_projection_identity_1}.
\end{proof}

\appendix
\setcounter{secnumdepth}{0}
\section{Appendix: Microlocal Construction of the Local Parametrix}
\label{app_microlocal}
\setcounter{equation}{0}
\setcounter{subsection}{0}
\setcounter{theorem}{0}
\renewcommand{\theequation}{A.\arabic{equation}}
\renewcommand{\thesubsection}{A.\arabic{subsection}}
\renewcommand{\thetheorem}{A.\arabic{theorem}}

In this appendix, we prove Lemma \ref{lem_middle_SL_pseudo_inverse} following the standard procedure based on symbol calculus, for which we will adapt the notation of \cite[Chapter 7]{taylor2010partial}.

Note that $\Gamma$ is straight, parameterized by $h:\mathbb{R}\to \Gamma$ with $h(t):=t(2\bm{e}_1+\bm{e}_2)$. Hence, defining the induced operator $\tilde{S}\in\mathcal{B}(H^{-\frac{1}{2}}(\mathbb{R}),H^{\frac{1}{2}}(\mathbb{R}))$ specified by the kernel $\tilde{S}(s,t):=S(h(s),h(t))$, it is sufficient to find $\tilde{Q}_{n}^{0}\in\mathcal{B}(H^{\frac{1}{2}}(\mathbb{R}),H^{-\frac{1}{2}}(\mathbb{R}))$ such that
\begin{itemize}
    \item[(i)] \begin{equation} \label{eq_microlocal_1}
\tilde{Q}_{n}^{0}\big(M_{\tilde{\beta}_{n}^{0}}\tilde{S}M_{\tilde{\beta}_{n}^{0}} \big)=M_{\tilde{\theta}_{n}^{0}}+\tilde{T}_{n}^{0}
\end{equation}
with $\tilde{T}_{n}^{0}$ being compact, and $\tilde{\beta}_{n}^{0}(t):=\beta_{n}^{0}(h(t))$, $\tilde{\theta}_{n}^{0}(t):=\theta_{n}^{0}(h(t))$;
    \item[(ii)] there exists $\rho>0$ such that for any disjoint open intervals $U_1,U_2\subset \mathbb{R}$ with $\text{dist}(U_1,U_2)>\rho$,
\begin{equation} \label{eq_microlocal_2}
M_{\tilde{\zeta}_1}\tilde{Q}^{0}_{n}M_{\tilde{\zeta}_2}=0,
\end{equation}
for all $\tilde{\zeta}_i\in C^1(\mathbb{R})$ supported in $U_i$ ($i=1,2$).
\end{itemize}

First, it is well-known that the single-layer boundary operator is locally a classical elliptic pseudodifferential (PDO) operator of order $-1$, i.e., $\tilde{S}\in OPS^{-1}_{loc}(\mathbb{R})$ (cf. \cite[Chapter 7, Prop. 11.2]{taylor2010partial}), as its principal near-diagonal behavior is the same logarithm-type as for the standard Green kernel and the out-going correction contributes only a smooth remainder. This means that upon the two-sided multiplication of $\sqrt{\tilde{\beta}_{n}^{0}}$ and $\beta_{n}^{0}$, which are compactly supported and equal to one in a neighborhood of $(-7n,7n)$, there exists a symbol $a_n$ of order $-1$, elliptic near $(-7n,7n)$, whose quantization equals $M_{\tilde{\beta}_{n}^{0}}\tilde{S}M_{\sqrt{\tilde{\beta}_{n}^{0}}}$ , i.e., 
$$
M_{\tilde{\beta}_{n}^{0}}\tilde{S}M_{\sqrt{\tilde{\beta}_{n}^{0}}}=\text{op}(a_n).
$$
Hence, recalling that $\text{supp}(\theta_{n}^{0})\subset (-7n,7n)$, we can choose $p_n(s,\xi)\in S^1$ with $s$-support contained in $(-7n,7n)$ that satisfies
\begin{equation*}
    p_{n}\#a_{n}=\tilde{\theta}_{n}^{0}\,\text{ mod $S^{-\infty}$}.
\end{equation*}
Define
\begin{equation*}
\tilde{P}_{n}^{0}:=\text{op}(p_{n}).
\end{equation*}
Then we have
\begin{equation} \label{eq_microlocal_3}
\tilde{P}_{n}^{0}M_{\tilde{\beta}_{n}^{0}}\tilde{S}M_{\sqrt{\tilde{\beta}_{n}^{0}}}=M_{\tilde{\theta}_{n}^{0}}+\tilde{K}_{n}^{0}
\end{equation}
with $\tilde{K}_{n}^{0}$ being compact. Next, select $\chi_{\rho}\in C^{\infty}(\mathbb{R})$ such that
\begin{equation*}
\text{supp}(\chi_{\rho})\subset (-\rho,\rho),\quad \chi_{\rho}(t)=1\quad (|t|<\rho/2).
\end{equation*}
Then the desired operator $\tilde{Q}_{n}^{0}$ is defined by its kernel
\begin{equation} \label{eq_microlocal_4}
\tilde{Q}_{n}^{0}(s,t):=\sqrt{\tilde{\beta}_{n}^{0}}(s)\chi_{\rho}(s-t)\tilde{P}_{n}^{0}(s,t).
\end{equation}
We now check that the operator $\tilde{Q}_{n}^{0}$ satisfies the properties \eqref{eq_microlocal_1} and \eqref{eq_microlocal_2}. First, $\eqref{eq_microlocal_2}$ follows immediately by the definition \eqref{eq_microlocal_4} as $\chi_{\rho}(s-t)$ vanishes when $|s-t|\geq \rho$. On the other hand, we note that
\begin{equation} \label{eq_microlocal_5}
\tilde{H}(s,t):=\big[1-\chi_{\rho}(s-t)\big]\tilde{P}_{n}^{0}(s,t)\in C^{\infty}(\mathbb{R}\times \mathbb{R})
\end{equation}
as $\tilde{P}_{n}^{0}(s,t)$ is smooth away from the off-diagonals. Hence, denoting $\tilde{H}$ as the integral operator specified by the kernel $\tilde{H}(s,t)$, one sees that
\begin{equation} \label{eq_microlocal_6}
\begin{aligned}
\tilde{Q}_{n}^{0}M_{\tilde{\beta}_{n}^{0}}\tilde{S}M_{\tilde{\beta}_{n}^{0}}
&=M_{\sqrt{\tilde{\beta}_{n}^{0}}}\tilde{P}_{n}^{0}M_{\tilde{\beta}_{n}^{0}}\tilde{S}M_{\tilde{\beta}_{n}^{0}}
+M_{\sqrt{\tilde{\beta}_{n}^{0}}}\tilde{H}M_{\tilde{\beta}_{n}^{0}}\tilde{S}M_{\tilde{\beta}_{n}^{0}} \\
&\overset{\eqref{eq_microlocal_3}}{=}M_{\sqrt{\tilde{\beta}_{n}^{0}}}\big(  M_{\tilde{\theta}_{n}^{0}}+\tilde{K}_{n}^{0}\big)M_{\sqrt{\tilde{\beta}_{n}^{0}}}
+M_{\sqrt{\tilde{\beta}_{n}^{0}}}\tilde{H}M_{\tilde{\beta}_{n}^{0}}\tilde{S}M_{\tilde{\beta}_{n}^{0}} \\
&=M_{\tilde{\theta}_{n}^{0}}+M_{\tilde{\theta}_{n}^{0}}  M_{\sqrt{\tilde{\beta}_{n}^{0}}\cdot (1-\tilde{\theta}_{n}^{0})}+M_{\sqrt{\tilde{\beta}_{n}^{0}}\cdot (1-\tilde{\theta}_{n}^{0})}M_{\tilde{\theta}_{n}^{0}}M_{\sqrt{\tilde{\beta}_{n}^{0}}} \\
&\quad +M_{\sqrt{\tilde{\beta}_{n}^{0}}}\tilde{K}_{n}^{0}M_{\sqrt{\tilde{\beta}_{n}^{0}}}
+M_{\sqrt{\tilde{\beta}_{n}^{0}}}\tilde{H}M_{\tilde{\beta}_{n}^{0}}\tilde{S}M_{\tilde{\beta}_{n}^{0}},
\end{aligned}
\end{equation}
where the last inequality follows by the fact that $\tilde{\beta}_{n}^{0}=1$ on the support of $\tilde{\theta}_{n}^{0}$. Then, by \eqref{eq_microlocal_5}, one sees that $M_{\sqrt{\tilde{\beta}_{n}^{0}}}\tilde{H}M_{\tilde{\beta}_{n}^{0}}$ is compact because its kernel is smooth and compactly supported. It is also direct to see that all the other operators in \eqref{eq_microlocal_6} are compact. Hence, the proof of \eqref{eq_microlocal_1} is complete.

\footnotesize
\bibliographystyle{plain}
\bibliography{ref}

\end{document}